\documentclass[12pt,a4paper]{article}
\usepackage{amscd,amsmath,amssymb,amsfonts,xspace,mathrsfs,mathtools,amsthm}
\usepackage{xr-hyper}
\usepackage{jheppub}
\usepackage[dvipsnames]{xcolor}
\usepackage[utf8]{inputenc}
\usepackage{bm}
\usepackage{bbm}
\usepackage{graphicx}
\usepackage[color,matrix,arrow]{xy}
\usepackage{stmaryrd}
\SetSymbolFont{stmry}{bold}{U}{stmry}{m}{n}
\usepackage[shortlabels]{enumitem}
\usepackage{array,amsmath,booktabs}
\usepackage{slashed}
\usepackage{tensor}
\usepackage{caption}
\usepackage{subcaption}
\usepackage{enumitem}
\usepackage[textsize=tiny]{todonotes}
\usepackage{diagbox}
\usepackage{pdflscape}
\usepackage{silence}
\usepackage{booktabs, cellspace, hhline}
\makeatletter
\newcommand*{\addFileDependency}[1]{
  \typeout{(#1)}
  \@addtofilelist{#1}
  \IfFileExists{#1}{}{\typeout{No file #1.}}}
\makeatother
\newcommand*{\myexternaldocument}[1]{
\externaldocument[][nocite]{#1}
\addFileDependency{#1.tex}
\addFileDependency{#1.aux}}
\myexternaldocument{uplane_part1_v2}

\newtheorem{theorem}{Theorem}
\newtheorem{lemma}{Lemma}
\newtheorem{proposition}{Proposition}
\newtheorem{definition}{Definition}

\newtheorem{corollary}{Corollary}
\newcommand{\AD}{\text{AD}}
\newcommand{\SW}{\text{SW}}
\newcommand{\ws}{Weierstra{\ss}}
\newcommand{\Oi}{\Omega_{\mathbb Q(i)}}
\newcommand{\Oa}{\Omega_{\mathbb Q(\alpha)}}
\newcommand{\ima}{\text{Im}}

\newcommand{\ttu}{\mathtt{u}}

\newcommand{\ubp}{u_{\text{bp}}}
\newcommand{\tbp}{\tau_{\text{bp}}}
\newcommand{\ybp}{y_{\text{bp}}}

\newcommand{\mad}{m_{\text{AD}}} 
\newcommand{\uad}{u_{\text{AD}}}

\newcommand{\ord}{\text{ord}\,}

\newcommand{\pt}{\mathbb P^2} 
\newcommand{\spinc}{\text{Spin}^c} 
 
\newcommand{\Res}{\text{Res}} 
 
\newcommand{\psl}{\text{PSL}(2,\mathbb Z)} 
\newcommand{\slz}{\text{SL}(2,\mathbb Z)}

\newcommand{\nstar}{\CN=2^*}

\newcommand{\jt}{\vartheta}
\newcommand{\tn}{\tau_0} 

\newcommand{\chih}{\chi_{\rm h}}

\newcommand{\tuv}{\tau_{\text{\tiny{UV}}}}

\newcommand{\bea}{\begin{equation} \begin{aligned}}
\newcommand{\eea}{\end{aligned} \end{equation}}
\newcommand{\be}{\begin{equation}} 
\newcommand{\ee}{\end{equation}} 
\newcommand{\bes}{\begin{equation*}}
\newcommand{\ees}{\end{equation*}}
\newcommand{\sgn}{\mathrm{sgn}}

\newcommand{\im}{i}

\newcommand{\CA}{\mathcal{A}} 
\newcommand{\CB}{\mathcal{B}} 
\newcommand{\CC}{\mathcal{C}}  
\newcommand{\cD}{\mathcal{D}}  
\newcommand{\CE}{\mathcal{E}}  
\newcommand{\CF}{\mathcal{F}}

\newcommand{\CI}{\mathcal{I}}
\newcommand{\CJ}{\mathcal{J}}
\newcommand{\CK}{\mathcal{K}}
\newcommand{\CL}{\mathcal{L}} 
\newcommand{\CM}{\mathcal{M}}  
\newcommand{\CN}{\mathcal{N}}
\newcommand{\CO}{\mathcal{O}} 

\newcommand{\CQ}{\mathcal{Q}}

\newcommand{\CS}{\mathcal{S}}
\newcommand{\CT}{\mathcal{T}}

\newcommand{\BZ}{\mathbb{Z}}

\newcommand{\fz}{\mathfrak{z}}

\newcommand{\bfb}{{\boldsymbol b}}

\newcommand{\bfm}{{\boldsymbol m}} 
\newcommand{\bfn}{{\boldsymbol n}}

\newcommand{\bfx}{{\boldsymbol x}}

\newcommand{\bfrho}{{\boldsymbol \rho}}

\newcommand{\bfk}{{\boldsymbol k}}
\newcommand{\bfz}{{\boldsymbol z}}
\newcommand{\bfnull}{{\boldsymbol 0}}
\newcommand{\bfmu}{{\boldsymbol \mu}}

\title{Topological twists of massive SQCD, Part II}

\abstract{
This is the second and final part of ``Topological twists of massive SQCD''. Part I is available at \href{https://arxiv.org/abs/2206.08943}{arXiv:2206.08943}. 
In this second part, we evaluate the contribution of the Coulomb branch to topological path integrals for $\mathcal{N}=2$ supersymmetric QCD with $N_f\leq 3$ massive hypermultiplets on compact four-manifolds.  Our analysis includes the decoupling of hypermultiplets, the massless limit and the merging of mutually non-local singularities at the Argyres-Douglas points. We give explicit mass expansions for the four-manifolds $\mathbb{P}^2$ and $K3$. For $\mathbb{P}^2$, we find that the correlation functions are polynomial as function of the masses, while infinite series and (potential) singularities occur for $K3$. The mass dependence corresponds mathematically to the integration of the equivariant Chern class of the matter bundle over the moduli space of $Q$-fixed equations. We demonstrate that the physical partition functions agree with mathematical results on Segre numbers of instanton moduli spaces. 
\\
\\
 \noindent\today
}

\author{Johannes Aspman${}^\sharp$, Elias Furrer${}^\flat$, Jan Manschot${}^\natural$ \\
\vspace{5pt}
{${}^\sharp\ $\it Department of Computer Science, Czech Technical University in Prague, Karlovo nam. 13, Prague 2, Czech Republic \\ \vspace{5pt}
${}^\flat\ $ \it School of Mathematics, University of Birmingham,\\
\it Watson Building, Edgbaston, Birmingham B15 2TT, United Kingdom \\ \vspace{5pt}
${}^\natural\ $ \it School of Mathematics, Trinity College, Dublin 2, Ireland\\
${}^\natural\ $ \it Hamilton Mathematical Institute, Trinity College, Dublin 2 \\
${}^\natural\ $ \it School of Natural Sciences, Institute for Advanced Study, 1 Einstein Drive, Princeton, NJ 08540 USA
\vspace{20pt} 

}}

\preprint{}
\begin{document}
\maketitle
\setcounter{figure}{8}
\setcounter{section}{6}

This is the second and final part of ``Topological twists of massive SQCD''. Part I is available as preprint at \href{https://arxiv.org/abs/2206.08943}{arXiv:2206.08943} \cite{Aspman:2022sfj}. The numbering of sections is consecutive to that of Part I, while each part contains its own reference list. Since Part II has developed to a larger text than anticipated, the following interlude provides a complementary and extended introduction to Part II.
A combined document with part I and part II can be found on \href{https://www.maths.tcd.ie/~manschot/}{here}.

\section{Interlude} 
\label{Interlude}

In this part II, we study topological partition functions for four-dimensional  $\CN=2$ supersymmetric QCD with $N_f=0,\dots, 3$ massive hypermultiplets.  The low-energy theory in flat space has a rather rich structure: The $2+N_f$ singular vacua move on the Coulomb branch smoothly as a function of the masses, which we denote by $\bfm=(m_1,\dots, m_{N_f})$. The vacua can collide in two distinct ways, depending on the Kodaira type of the corresponding singular fibre in the Seiberg-Witten geometry. If $r$ $I_1$ singularities for $r$ mutually local dyons merge, they form a new singularity of Kodaira type $I_r$. When singularities corresponding to mutually non-local dyons collide, they rather lead to Kodaira type $II$, $III$ or $IV$ singularities, which give rise to superconformal Argyres-Douglas (AD) theories \cite{Argyres:1995xn,Argyres:1995jj}. In general, if two or more masses of the hypermultiplets align, the flavour symmetry enhances and a Higgs branch opens up. 

It is an interesting question how this singularity structure is reflected in the topological theory. For the mass deformation of $\CN=4$ Yang-Mills, the $\CN=2^*$ theory, this has been analysed in \cite{Labastida:1998sk,Manschot:2021qqe}, which connects Vafa-Witten and Donaldson-Witten invariants. The structure for SQCD bears much resemblance to that case, yet the multiple masses and AD singularities give rise to richer structure with more intricacies. Before discussing our findings and results, we give an overview of previous literature, including part I.

\subsection{Literature overview}
For a generic compact four-manifold $X$, the topological partition function of SQCD takes the form of a sum of a $u$-plane integral $\Phi_\bfmu^J$ and a Seiberg-Witten (SW) contribution \cite{Moore:1997pc},
\be \label{Z_contr}
Z_\bfmu^J(\bfm)=\Phi_\bfmu^J(\bfm)+\sum_{j=1}^{2+N_f} Z^J_{\text{SW},j,\bfmu}(\bfm).
\ee 
The partition function depends on three distinct collections of parameters: The masses $\bfm$, the metric $J$ and  a set of fluxes $\bfmu$ for the theory (such as a 't Hooft flux for the gauge bundle and background fluxes for the flavour group). Geometrically, the mass dependence of  $Z_\bfmu^J(\bfm)$ contains  information on intersection numbers of Chern classes on gauge theoretic moduli spaces \cite{LoNeSha}.

The $u$-plane integral $\Phi^J_\bfmu(\bfm)$ vanishes for manifolds with $b_2^+>1$ \cite{Moore:1997pc}. Manifolds with $b_2^+=0,1$ are therefore of special interest, since they have the right topology to probe the full Coulomb branch. We will restrict to $X$ with $b_2^+\geq 1$. For $b_2^+=1$, the SW contribution can be found from the $u$-plane integral by   wall-crossing as a function of the metric $J$. While the $u$-plane integral  depends on $X$ only through its intersection form on $H^2(X,\mathbb Z)$, the SW invariants can distinguish between homeomorphic manifolds with distinct smooth structures.

In part I, we have defined the  $u$-plane integral $\Phi_\bfmu^J(\bfm)$ of massive SQCD. For fixed fluxes $\bfmu$ on a given four-manifold $X$, it is essentially determined by the SW solution for the Coulomb branch or $u$-plane of the theory. The fibration of the SW curve  over the $u$-plane has been identified as a rational elliptic surface (RES) $\CS(\bfm)$, which is also known as the Seiberg-Witten surface 
\cite{Sen:1996vd,Banks:1996nj,Minahan:1998vr,Eguchi:2002fc,Malmendier:2008yj,Caorsi:2018ahl,Caorsi:2019vex,Closset:2021lhd,Magureanu:2022qym}.\footnote{See \cite{Closset:2021lhd} for a recent review on Seiberg-Witten geometry.} This geometry encodes much of the data of the supersymmetric low-energy effective theory.
The analytical structure of the $u$-plane integral is therefore to a great extent determined by that of the surface $\CS(\bfm)$. As explained in part I, the $u$-plane integral can be mapped to a fundamental domain $\CF(\bfm)$ associated with the elliptic surface $\CS(\bfm)$, and collapses to a finite sum over cusps, elliptic points and interior singular points of the fundamental domain (see \eqref{integrationresult}). In terms of the SW surface, we calculate the sum of the $u$-plane integrand over the singular fibres of $\CS(\bfm)$, which fall into Kodaira's classification. The possible configurations of singularities for rational elliptic surfaces have been classified as well \cite{Persson:1990,Miranda:1990}.

A notable intricacy for the evaluation  is the fact that the mass dependence of the surface $\CS(\bfm)$ is not globally smooth, which gives rise to  branch points and branch cuts for $N_f\geq 1$ \cite{Aspman:2021vhs}. This requires a careful regularisation of the fundamental domain: It must be chosen to not cross any branch points in the renormalisation of the integral. Moreover, as the masses are varied, the singular fibres in $\CS(\bfm)$ can split or merge. In the limit where an Argyres-Douglas (AD) point emerges, the fundamental domain is `pinched' at the AD point and it splits into two \cite{Aspman:2021vhs}. See e.g. Figure \ref{fig:nf1ADlimit}.

As mentioned above, the SW contribution $Z_{\text{SW},j,\bfmu}$ can be determined by a wall-crossing argument from their corresponding cusps of the $u$-plane integral. Due to their application to Donaldson invariants in the pure $N_f=0$ theory, they have been studied predominantly for singularities of type $I_1$, corresponding to one massless monopole or dyon. The generalisation to SQCD proceeds analogously, since in such configurations all singularities are of type $I_1$ as well \cite{Moore:1997pc}. Partition functions for the massless theories are determined in \cite{Malmendier:2008db,Kanno:1998qj}.

The partition functions of Argyres-Douglas theories on four-manifolds have been studied from various perspectives \cite{Nishinaka:2012kn,Kimura:2020krd,Fucito:2023txg,Fucito:2023plp,Marino:1998uy,Marino:1998uy_short,Gukov:2017,Moore:2017cmm,Dedushenko:2018bpp,moore_talk2018,Marino:1998bm}.
While the $u$-plane integrand is regular at any smooth point on the Coulomb branch, it can diverge at the elliptic AD points. In contrast to the strong coupling singularities of type $I_r$, their contribution to correlators exhibits continuous metric dependence rather than discrete wall-crossing. Besides, the expansion of the integrand at  elliptic points has a very different flavour than at cusps, and has been largely unexplored in the literature. The study of such elliptic points is also of interest due to other types of singularities, such as the Minahan--Nemeschansky SCFTs \cite{Minahan:1996fg,Minahan:1996cj}.

 Other intriguing connections between theories can be realised by compactification, which relates invariants associated with geometries of different dimensions. This connects for instance the Donaldson invariants, Floer homology, Gromov-Witten invariants and K-theoretic versions \cite{Taubes1994,Bershadsky:1995vm,Gottsche:2006,Nakajima:2005fg,Harvey_1995,donaldson1995floer,kim2023}, and allows to conjecture QFTs themselves as invariants \cite{Gadde:2013sca,Dedushenko:2017tdw}.

\subsection{Summary of results}
 In order to study the analytical structure of topological partition functions explicitly, we focus on two manifolds: the complex projective plane $\mathbb{P}^2$ and $K3$ surfaces. For $\mathbb{P}^2$, only the $u$-plane integral contributes, while for $K3$ there is only the SW contribution. The dependence on the masses can be studied in various special limits, such as large and small masses, and limits to AD theories. In part I, we argued that the twisted theory can be coupled to background fluxes for the flavour group. In this part II, by explicit computation we demonstrate that this indeed provides a refined family of theories with nonzero partition functions.

As announced in part I, we evaluate $u$-plane integrals using mock modular forms and Appell--Lerch sums. For $\mathbb{P}^2$, various choices of mock modular forms have appeared in the literature, which all differ by an integration `constant', in this case a holomorphic modular form. Since the anti-derivative of the integrand must transform under all possible monodromies on the $u$-plane of SQCD with arbitrary masses, this singles out a specific $\slz$ mock modular form: It is the $q$-series $H^{(2)}$ of Mathieu moonshine \cite{Eguchi:2010ej, Dabholkar:2012nd}, which relates the dimensions of irreducible representations of the sporadic group $M_{24}$ to the elliptic genus of the $K3$ sigma model with ${\CN=(4,4)}$ supersymmetry. Including either surface observables or nontrivial background fluxes, this function generalises to an $\slz$ mock Jacobi form, giving an interesting refinement.

For four-manifolds with $b_2=1$, the weak coupling cusp contributes to all correlation functions, while the strong coupling cusps never contribute. For all four-manifolds with $b_2^+>0$ that admit a Riemannian metric of positive scalar curvature, the SW invariants are zero due to a well-known vanishing theorem \cite{Witten:1994cg,Taubes1994,bryan1996, Dedushenko:2017tdw}. Hence by SW wall-crossing, the strong coupling contributions to the $u$-plane integral are expected to vanish as well. We confirm this by an analysis of the $u$-plane integrand at the singularities for such manifolds, including the del Pezzo surfaces $dP_n$.  Furthermore, we prove that in absence of background fluxes for the flavour group, the branch points never contribute to $u$-plane integrals.

Our calculations for $\mathbb P^2$ agree with previous results in the literature, which were available for massless SQCD \cite{Malmendier:2008db}. A consistency check available only for \emph{massive} SQCD is the infinite mass decoupling limit, which precisely matches with that of the proposed form of correlation functions in the UV theory. The limit of the $u$-plane integral takes the form as given in \eqref{IR-decoupling}, and we use it to check our explicit results for $\mathbb P^2$: If all hypermultiplets are decoupled, one recovers the Donaldson invariants of $\mathbb P^2$. Our results agree precisely with \cite{ellingsrud1995wall,Gottsche:1996} for $N_f=0$ and \cite{Malmendier:2008db} for massless $N_f=2$ and $N_f=3$. 
The UV formula provides another consistency check in the form of a selection rule  for observables. For instance, correlation functions of point observables on $\mathbb P^2$ with canonical 't Hooft flux are valued in the polynomial ring of the masses, where the virtual rank and  degree of the Chern class of the matter bundle as well as the virtual dimension of the instanton moduli space can be read off from the exponents of the masses and dynamical scale. The coefficients are then (rational) intersection numbers on the moduli space of solutions to the $\CQ$-fixed equations.

Coupling the hypermultiplets to background fluxes for the flavour group allows to formulate the theory for arbitrary 't Hooft fluxes. We determine the couplings to the background fluxes for $N_f=1,2$ by  integration of the SW periods, and evaluate the correlation functions on $\mathbb P^2$.
For nontrivial background flux, the results depend on the expansion point, i.e. small or large masses. 
This is due to the pole structure of the (elliptic) mock Jacobi form $H^{(2)}$, which we determine precisely.

As discussed in part I, the superconformal Argyres-Douglas theories present themselves in the fundamental domain of massive SQCD as elliptic points. We expand the $u$-plane integrand around any singularity of type $II$, $III$ and $IV$. The anti-derivative of the photon path integral is a non-holomorphic modular form, which we evaluate at elliptic points using the Chowla--Selberg formula. This formula expresses the value of modular forms at elliptic points as products of the Euler gamma function at rational numbers. Interestingly, elliptic points are all zeros of the  function $H^{(2)}$. Together with the holomorphic expansion of the measure factor, whose order of vanishing at any elliptic point we determine,
we show that for four-manifolds with odd intersection form and canonical 't Hooft flux the $u$-plane integrand is regular and thus there is no contribution from AD points in those cases. Our results for the expansion at elliptic points can be readily generalised to other $u$-plane integrals containing elliptic points.

Further, we derive the general form of SW contributions for SU(2) SQCD and evaluate correlation functions of point observables for $X=K3$.
If the masses are large, $N_f$ singularities move on the $u$-plane to infinity, while two converge to $\pm \Lambda_0^2$, giving the SW singularities of the pure $N_f=0$ theory. This allows to attribute the $N_f$ singularities at large $|u|$ to the monopole component of the moduli space of $\CQ$-fixed equations, while the union of the monopole, dyon and weak coupling contribution corresponds to the instanton component. See Fig. \ref{fig:branches}. Note the distinction between the monopole contribution to the instanton component and the monopole component. 

\begin{figure}[ht]\centering
	\includegraphics[scale=1]{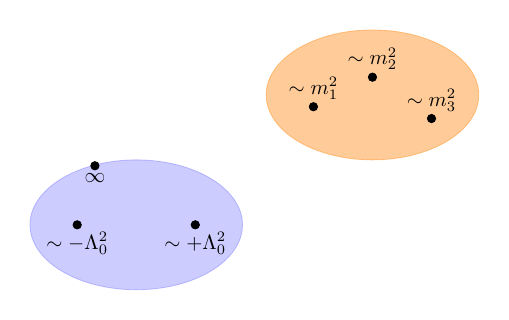}
	\caption{The singularity structure on the Coulomb branch with $N_f$ heavy hypermultiplets consists of two singularities approaching the monopole and dyon points $\pm\Lambda_0^2$ of pure SU(2) SYM, while $N_f$ singularities are asymptotically large. We associate the two singularities of $\CO(\Lambda_0^2)$ and weak coupling $u=\infty$ to the instanton component (blue), while the other $N_f$ singularities are attributed to the abelian (or monopole) component (orange). In contrast to the case of small masses, for large masses the two components are well separated.}
 \label{fig:branches}
\end{figure}

The contributions of the `instanton' singularities to point observables on $K3$ are Laurent series in the inverse mass $\frac 1m$, which turn out to be generating functions of \emph{Segre numbers}. Segre classes first appeared in the context of moduli of vector bundles in an article by Tyurin \cite{Tyurin1994}. They were later generalised for higher rank bundles over projective surfaces in \cite{Gottsche:2020ass}. Recently, the correspondence between higher rank Segre numbers on moduli spaces of stable sheaves on surfaces and their Verlinde numbers \cite{Losev:1995cr, LoNeSha, Gottsche:2006} has been proven \cite{Oberdieck_2022}. 
We establish the relation between the physical partition functions and these geometric invariants by an explicit mapping using the SW geometry. The coefficients of the `monopole' contribution lack such a mathematical interpretation. However, combining it with the instanton contribution eliminates the (infinite) principal part of the series, resulting in a polynomial in the masses. See for instance Table \ref{PtObsNf1K3}. 

 For $I_r$ singularities with $r\geq 2$, the SW invariants are not readily well-defined, since the moduli space can be non-compact and the integrals require regularisation \cite{Dedushenko:2017tdw,bryan1996,Moore:1997dj}. The SW invariants for $I_r$ singularities with $r>1$ are invariants for the multi-monopole equations, and 
require higher order corrections in the local variables \cite{Kanno:1998qj}. 
  The SW invariants are in this case non-vanishing for nonzero $\spinc$ structures.  We calculate the simplest nontrivial case, which is an $I_2$ singularity in $N_f=2$ with equal masses on $K3$. An apparent feature is the potentially divergent behaviour of the SW partition function near the superconformal Argyres-Douglas points. We propose that these divergences are  rendered finite by sum rules for the  $I_r$ SW invariants for different $r$ (see \eqref{SWhyp}), generalising earlier results for sum rules for $I_1$ SW invariants  
\cite{Marino:1998uy,Marino:1998uy_short,Moore:2017cmm}.   A calculation in $N_f=3$ with the same type $III$ AD point as in $N_f=2$ reproduces the same sum rules for the $I_1$ and $I_2$ SW invariants. This suggests that the constraints on the topology from the regularity at any AD point is determined completely by the universality class of the superconformal theory. We note that the collision of $I_1$ points to an $I_r$ point also exhibits a mass singularity. 
This singularity is expected to be related to the appearance of a non-compact Higgs branch \cite{Dedushenko:2017tdw,AFMM:future}, and therefore does not give rise to new sum rules.

 Imposing the sum rules, in all cases we study the correlation functions are then polynomials in the masses. The type $III$ AD point can also be approached away from the equal mass locus in $N_f=2$, where three $I_1$ singularities collide rather than an $I_2$ and an $I_1$. The two limits agree up to a divergent term $\sim(m_1-m_2)^{-2}$, which is a consequence of a non-compact Higgs branch appearing as $m_1\to m_2$ \cite{AFMM:future,Dedushenko:2017tdw,LoNeSha}. The $I_2$ SW contribution thus naturally regularises the singular limit of colliding $I_1$ singularities. 
 
 Organising the SW contributions of the pure $N_f=0$ theory to a correlation function with exponentiated observables, the generating function in many cases satisfies an ordinary differential equation with respect to the point observable. Four-manifolds whose SW invariants enjoy this property are said to be of (generalised) simple type. An example are the $K3$ surfaces, where the corresponding generating functions for massless SQCD have been studied \cite{Kanno:1998qj}. We generalise this analysis to the massive theories, and show that for generic masses the differential equation is determined by the physical discriminant associated with the massive theory (see \eqref{SWST_Nf}). When mutually local singularities collide (as is the case in massless $N_f=2,3$),  zeros of the discriminant collide, and give rise to higher order zeros of the characteristic polynomial of the ODE. Such general structure results on generating functions are also of interest regarding the asymptology of correlation functions for many fields \cite{Korpas:2019cwg}. Due to the rich phase structure of the SQCD Coulomb branch, it is not obvious if the generating function of correlation functions as a formal series is well-defined, that is, if it defines an entire function on the homology ring $H_*(X,\mathbb C)$.
 \\

\subsection{Outline of Part II}
This part II is organised as follows. In section \ref{sec:further_aspects}, we discuss various aspects of the $u$-plane integral of massive SQCD, such as the contributions from singular points, the measure factor including gravitational couplings, and the decoupling limit.
In section \ref{sec:behaviour_special_points} we calculate auxiliary expansions of various Coulomb branch functions near special points, such as weak coupling, strong coupling cusps, and branch points. 
In section \ref{sec:uplane_int_mock} we formulate $u$-plane integrals over fundamental domains in the presence of branch points, analyse the components of the integrand in detail and derive conditions on the cusps to contribute to correlation functions.
In section \ref{sec:p2}, we calculate $u$-plane integrals of massive SQCD on the complex projective plane $\mathbb P^2$, with $N_f\leq 3$ arbitrary masses and with nontrivial background fluxes.
In section \ref{sec:SWcontributions} we rederive the SW contributions for $I_1$ singularities by a wall-crossing argument at the strong coupling cusps. We furthermore propose the form of SW contributions for $I_2$ singularities,  calculate point correlators on $K3$, discuss the AD limit and the relation to Segre invariants. Finally, we propose generalised simple type conditions for generic as well as coincident masses.
In section \ref{sec:ADcontribution} we discuss the  contributions of AD points to the $u$-plane integral. We conclude with a brief discussion in section \ref{sec:discussion}. 
Various useful expansions, derivations, proofs and formulas can be found in the appendices \ref{app:modularforms2}, \ref{app:mass_expansions}, \ref{sec:bkgFluxes} and \ref{sec:Kodaira_invariant}.

\section{Further aspects of topological path integrals}\label{sec:further_aspects}
This section discusses further aspects and preliminaries of topological path integrals. Subsection \ref{sec:Contributions} discusses the different contributions to the topological path integral. Subsection \ref{sec:measure_factor} discusses the measure of the $u$-plane contribution. In Subsection \ref{sec:scalinglimit} we study the behaviour of the path integral under decoupling of hypermultiplets in the infinite mass limit.

\subsection{General structure}
\label{sec:Contributions}
For a generic compact four-manifold, the topological partition function of SQCD takes the form (\ref{Z_contr}). 
As discussed before, the $u$-plane integral $\Phi^J_\bfmu$ receives contributions from the weak coupling cusp, $\tau\to i\infty$, and the $N_f+2$ strong coupling singularities, such that we can express $Z_\bfmu^J$ as
\be \label{Z_contributions_generic}
Z_\bfmu^J(\bfm)=\Phi_{\bfmu,\infty}^J(\bfm)+\sum_{j=1}^{2+N_f} \Phi^J_{\bfmu,j}(\bfm)+Z^J_{\text{SW},\bfmu,j}(\bfm).
\ee 
In Sections \ref{sec:uplane_int_mock} and \ref{sec:p2}, we will discuss and calculate the $u$-plane integral for generic and specific four-manifolds. In Section \ref{sec:SWcontributions} we will derive the action $Z_{\text{SW},j}$ for the theory near $u_j^*$ from the $u$-plane integral using wall-crossing. The reason for this is that wall-crossing of the total partition function can only be due to the non-compact direction in field space, i.e. $|u|\to \infty$ \cite{Moore:1997pc}.  Thus the wall-crossing of the strong coupling $u$-plane contributions $\Phi^J_{\bfmu,j}(\bfm)$ must cancel that of $Z^J_{\text{SW},\bfmu,j}(\bfm)$.

If the masses are tuned to an AD point, the partition function naturally splits into a contribution from a small neighbourhood of the AD-point, and its complement in the $u$-plane \cite{Moore:2017cmm}. This works out rather nicely when lifted to domains in the $\tau$-plane. On the AD mass locus,  the fundamental domain splits into a component including the original weak coupling regime, and a strong coupling component associated to the vicinity of the AD point in the $u$-plane \cite{Aspman:2021vhs}. 
The strong coupling singularities $\{1,\dots ,2+N_f \}$ accordingly split in two sets $S$ and $S'$, with $S\cup S'=\{1,\dots ,2+N_f \}$ and the singularities in $S'$ merging in the AD point. Furthermore, the fundamental domains in $\tau$-space include the elliptic points $e$ in $S$ and its complement $e'$ in $S'$. Schematically, we arrive at the following 
\be \label{Z_decomp_AD}
\begin{split} 
Z_\bfmu^J(\bfm_{\text{AD}})&=\Phi_{\bfmu,\infty}^J(\bfm_{\text{AD}})+\Phi_{\bfmu,e}^J(\bfm_{\text{AD}})+\lim_{\bfm\to \bfm_{\text{AD}}}\sum_{j\in S} \Phi^J_{\bfmu,j}(\bfm)+Z_{\text{SW},\bfmu,j}(\bfm)\\
&\quad +\Phi_{\bfmu,e'}^J(\bfm_{\text{AD}})+\lim_{\bfm\to \bfm_{\text{AD}}}\sum_{j'\in S'} \Phi^J_{\bfmu,j'}(\bfm)+Z_{\text{SW},\bfmu,j'}(\bfm).
    \end{split} 
\ee 
The limit on the right hand side occurs since each summand for fixed $j$ can diverge. The sum over $j$ may  remain finite as a consequence of sum rules \cite{Marino:1998uy}. The terms on the second line correspond to the vicinity of the AD point in the $u$-plane,  
\be
\label{ZmAD}
Z_{\widetilde{\text{AD}},\bfmu}^J=\Phi_{\bfmu,e'}^J(\bfm_{\text{AD}})+\lim_{\bfm\to \bfm_{\text{AD}}}\sum_{j'\in S'} \Phi^J_{\bfmu,j'}(\bfm)+Z_{\text{SW},\bfmu,j'}(\bfm),
\ee 
where the tilde on $\widetilde{\text{AD}}$ indicates that $Z_{\widetilde{\text{AD}},\bfmu}^J$ is the contribution to the partition function of SQCD from the neighbourhood of the AD point, rather than the partition function of the  intrinsic AD theory.

\begin{figure}[ht]\centering
	\includegraphics[scale=1.3]{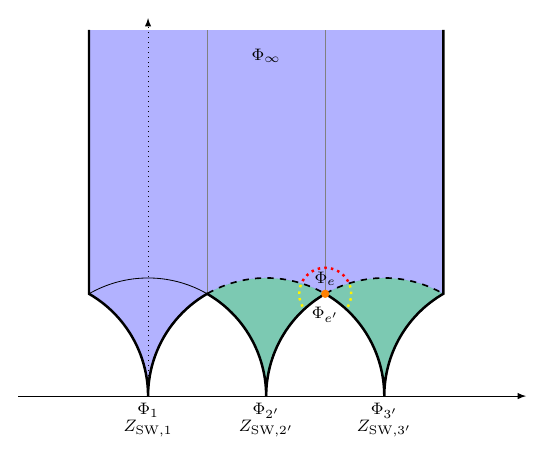}
	\caption{When the mass for the $N_f=1$ theory is tuned to the AD value, the partition function  \eqref{Z_decomp_AD} receives contributions from two disjoint regions. The first region (purple) is precisely the fundamental domain in the limit $m\to\mad$, while the other one (green) is the `zoomed in' domain, as studied for instance in  \cite{Argyres:1995xn,Moore:2017cmm,Aspman:2021vhs,Marino:1998uy}. The two domains are connected through the AD point (orange), whose boundary arcs (red and yellow) both have angle $\frac{2\pi}{3}$.}
 \label{fig:nf1ADlimit}
\end{figure}

The $2+N_f$ singularities of \eqref{Z_contributions_generic} are split up into the sum over $j$ and $j'$. 
We illustrate this in Fig. \ref{fig:nf1ADlimit}. In order for the limit $Z_\bfmu^J(\bfm) \to Z_\bfmu^J(\bfm_{\text{AD}})$ to be smooth, it is natural to expect that $\Phi_{\bfmu,e}^J(\bfm)=-\Phi_{\bfmu,e'}^J(\bfm)$. We will discuss the contributions from the AD points in detail in Section \ref{sec:ADcontribution}. It is an interesting question how to extract from $Z_{\widetilde{\text{AD}},\bfmu}^J$ (\ref{ZmAD}) the partition function $Z_{\AD,\bfmu}^J$ of the superconformal AD theory based on the `zoomed in' AD curve \cite{Argyres:1995xn}. The latter partition function is determined for instance in \cite{Moore:2017cmm}. Roughly, $Z_{\AD,\bfmu}^J$ is the leading term of $Z_{\widetilde{\text{AD}},\bfmu}^J$. 
We make some further comments in Section \ref{sec:results_K3}, and leave a more thorough analysis for future work \cite{AFMM:future}.

We close this subsection with some further notation. We will often omit the mass $\bfm$ from the argument of $Z_\bfmu^J$. Moreover, we denote the insertion of observables by straight brackets,
\be 
\label{Zcorr}
Z_\bfmu^J[\CO]=\left< \CO \right>^J_\bfmu,
\ee 
and similarly for the terms on the rhs of (\ref{Z_decomp_AD}). Two common observables are the exponentiated point and surface observables $e^{2pu/\Lambda_{N_f}^2}$ and $e^{I(\bfx)}$. For these observables, we also use $p$ and $\bfx$ as arguments of $Z_\bfmu^J$,
\be 
Z_{\bfmu}^J(\bfm,p,\bfx)=Z_{\bfmu}^J(p,\bfx)=Z_\bfmu^J[e^{2pu/\Lambda_{N_f}^2+I(\bfx)}].
\ee 
See Sec. \ref{sec:contaccterm} in part I for more details.

\subsection{Measure factor}\label{sec:measure_factor}
We recall that the metric-independent part of the path integral is the measure factor \eqref{measurefactornf}. It contains the topological couplings 
\begin{equation}\label{P2:ABtopological}
	A=\alpha_{N_f}\left(\frac{du}{da}\right)^{1/2},\qquad B=\beta_{N_f}\Delta_{N_f}^{1/8}
\end{equation}
of the theory to the Euler characteristic $\chi$ and the signature $\sigma$ of the four-manifold $X$. While the functions $\alpha_{N_f}$ and $\beta_{N_f}$ are independent on $\tau$, they can be functions of other moduli such as the masses $\bfm$ and the dynamical scale $\Lambda_{N_f}$, or the UV coupling $\tuv$ for $N_f=4$ or $\CN=2^*$. The functions $A$ and $B$ essentially do not change in form by including matter because the kinetic terms of the hypermultiplets have no explicit $\tau$-dependence \cite{Moore:1997pc}. This suggests that $\alpha_{N_f}$ and $\beta_{N_f}$ do not have a strong dependence on $N_f$. 

Furthermore, $\alpha_{N_f}$ and $\beta_{N_f}$  cannot  depend on any masses since otherwise the path integral would have additional global mass singularities, which are not physically motivated \cite{Marino:1998uy}. This argument includes the conformal fixed points \cite{Argyres:1995xn,Argyres:1995jj}, as these singularities occur only for special values and not on the whole $u$-plane.  Thus $\alpha_{N_f}$ and $\beta_{N_f}$ depend only on the scale $\Lambda_{N_f}$. This furthermore agrees with the fact that the gravitational factors need to reproduce the anomaly associated to the fields that have been integrated out, which eliminates a possible dependence of $\alpha$ and $\beta$ on $\bfm$ \cite{Marino:1998uy}.

 Since the couplings $A$ and $B$ are contained in the low-energy effective action as $\chi \log A+\sigma\log B$,  both $A$ and $B$  are necessarily dimensionless. With $\left[\frac{du}{da}\right]=1$ and $[\Delta_{N_f}]=2(N_f+2)$, this fixes the dimensionality 
\begin{equation}
\label{alphabetaNf}
\begin{aligned}
    \alpha_{N_f}&=\alpha_{N_f,0}\Lambda_{N_f}^{-\frac12},\\
    \beta_{N_f}&=\beta_{N_f,0}\Lambda_{N_f}^{-\frac14(N_f+2)},
    \end{aligned}
\end{equation}
where $\alpha_{N_f,0},\beta_{N_f,0}\in\mathbb C$ are dimensionless numbers. These gravitational couplings have been recently calculated for several families of  theories \cite{Manschot:2019pog,Closset:2021lhd,John:2022yql,Ashok:2023tsm}.

Using the decoupling limits, we find for the normalisation \eqref{measurefactornf} and the constants $\alpha$ and $\beta$,  
\begin{equation}\begin{aligned}\label{numerical_constants_Nf}
  \CK_{N_f,0}=\frac 1\pi, \qquad  \alpha_{N_f,0}=2^{\frac14}, \qquad \beta_{N_f,0}=2\,e^{\frac{\pi i}{8}N_f}.
\end{aligned}\end{equation}
The phase in $\beta_{N_f,0}$ originates from the decoupling of the discriminant, as we will discuss momentarily.

The effective gravitational couplings appear in the $u$-plane integrand as a product  $\CK_{N_f}\alpha_{N_f}^\chi\beta_{N_f}^\sigma$. Due to the fact that $\chi+\sigma=4$ for manifolds with $(b_1,b_2^+)=(0,1)$,  there is a normalisation ambiguity \cite{Moore:2017cmm}
\begin{equation}
(\CK_{N_f},\alpha_{N_f},\beta_{N_f})\sim (\kappa^{-4}\CK_{N_f},\kappa\alpha_{N_f},\kappa\beta_{N_f}),
\end{equation}
giving the same result for any $\kappa\in\mathbb C$.
In particular, in the $u$-plane integral only the ratio
\begin{equation}
    \frac{\beta_{N_f}}{\alpha_{N_f}}=\frac{\beta_{N_f,0}}{\alpha_{N_f,0}}\Lambda_{N_f}^{-\frac{N_f}{4}}
\end{equation}
is fixed. This agrees with  \cite{Manschot:2019pog} for $N_f=0,\dots 4$ and the general considerations in \cite{Marino:1998uy}. From our result \eqref{numerical_constants_Nf}, we find 
\begin{equation}
    \frac{\beta_{N_f,0}}{\alpha_{N_f,0}}=2^{\frac34}e^{\frac{\pi\im}{8} N_f}.
\end{equation}
For $N_f=0$, the unambiguous ratio $ \frac{\beta_{0,0}}{\alpha_{0,0}}=2^{\frac34}$ agrees with \cite{Korpas:2019cwg}, and matches with explicit computations of Donaldson invariants.

Since the $u$-plane integral computes intersection numbers on the moduli space, it should be properly normalised to be dimensionless. With $A$ and $B$ dimensionless, the only dimensionful quantity in the measure factor is the Jacobian $\frac{da}{d\tau}$. Thus
\begin{equation}
\label{defKNf}
    \CK_{N_f}=\CK_{N_f,0}\Lambda_{N_f}^{-1}=\frac{1}{\pi\,\Lambda_{N_f}}
\end{equation}
produces a dimensionless $u$-plane integral, with $\CK_{N_f,0}\in\mathbb C$ the number in \eqref{numerical_constants_Nf}.

Combining all the scales fixed by dimensional analysis and using $\chi+\sigma=4$, we have 
\begin{equation}\label{CNNf}
    \CN_{N_f}\coloneqq \CK_{N_f}\alpha_{N_f}^\chi\beta_{N_f}^\sigma=\CK_{N_f,0}\alpha_{N_f,0}^\chi\beta_{N_f,0}^\sigma\Lambda_{N_f}^{-(3+\sigma N_f/4)}. 
\end{equation}
This total normalisation factor $\CN_{N_f}$ will be important in the decoupling of the $u$-plane integral, as we will study in the subsequent subsection.

With the above analyses, we can now present the measure $\nu(\tau,\{\bfk_j\})$ (\ref{measurefactornf}) in a more tangible fashion. Let us first consider $\bfk_j=0$ for all $j$ and set $\nu(\tau,\{\bf 0\})=\nu(\tau)$. Substituting Eq. \eqref{ABtopological} for $A$ and $B$, and \cite[Eq. (3.13)]{Aspman:2021vhs}, we find
\be 
\nu(\tau)=-\frac{16\pi i}{4-N_f}\CN_{N_f} \frac{1}{P^{\rm M}_{N_f}}\left( \frac{du}{da} \right)^{- \sigma/2-1}\Delta_{N_f}^{1+\sigma/8},
\ee 
where the polynomials $P^{\rm M}_{N_f}$ appearing in  generalizations of the Matone relation \cite{Matone:1995rx} are given in \cite[Eq. (3.15)]{Aspman:2021vhs}, and $\CN_{N_f}$ is given by  \eqref{CNNf_explicit}. We can further substitute \cite[Eq. (3.9)]{Aspman:2021vhs} for $\Delta_{N_f}$, to express $\nu(\tau)$ as 
\be
\begin{split}
\label{nutauNf}
\nu(\tau)=\CN_{N_f}\frac{(-1)^{N_f+1}\pi i}{4-N_f} 2^{-\frac 34\sigma-2} e^{-\pi iN_f \frac \sigma8} \Lambda_{N_f}^{(2N_f-8)(\sigma/8+1)}   \frac{\eta(\tau)^{3\sigma+24}}{P^{\rm M}_{N_f}}  \left(\frac{du}{da} \right)^{\sigma+11}.
\end{split}
\ee 
These substitutions significantly simplify explicit calculations.

\subsection{Decoupling limit}\label{sec:scalinglimit}
This section will discuss the decoupling limit of the $u$-plane integral and compare with the decoupling (\ref{UV-decoupling}) from the UV in Section \ref{corrfunctions}. 
The $u$-plane integral (\ref{generaluplaneintegral}) reads
\begin{equation}\label{P2:generaluplaneintegral}
\begin{aligned}
&\Phi_{\bfmu, \{\bfk_j\}}^{J,N_f}(p,\bfx, \bfm_{N_f}, \Lambda_{N_f})=\\
&\qquad \CN_{N_f}\,\int_{\CF(\bfm_{N_f})}\!\!\!\!\!\!\!\!d\tau\wedge d\bar \tau \, \frac{da}{d\tau} \,\left(\frac{du}{da}\right)^{\frac\chi2} \Delta_{N_f}^{\frac\sigma8}\prod_{i,j=1}^{N_f}e^{-\pi i w_{ij}B(\bfk_i,\bfk_j)}\,\Psi_\bfmu^J(\tau,\bfz)\,e^{2p\tfrac{u}{\Lambda_{N_f}^2}+\bfx^2 G_{N_f}},
\end{aligned}
\end{equation}
where $\bfz$ is given in \eqref{bfz_def}. Here, we have combined the $N_f$-dependent normalisation factors in $\CN_{N_f}$ \eqref{CNNf}, which facilitates the decoupling analysis.

In the scaling limit \eqref{decouplimit} the curve \eqref{eq:curves} for given $N_f$ flows to the curve with $N_f-1$. The flow of some of the ingredients of the $u$-plane integrand has been determined in \cite{Ohta_1997,Ohta:1996hq,Aspman:2021vhs}. We summarise them in Table \ref{Decoupling_uplane}.

\begin{table}[ht]\begin{center}
\renewcommand{\arraystretch}{1}
\begin{tabular}{|>{$\displaystyle}Sl<{$}|>{$\displaystyle}Sl<{$}|}
		\hline
		(m_1,\dots, m_{N_f}) &  	(m_1,\dots, m_{N_f-1}) \\
		\hline
		\CF(\bfm_{N_f}) & \CF(\bfm_{N_f-1}) \\
		u & u \\
		\tfrac{du}{da} & \tfrac{du}{da} \\
		\tfrac{da}{d\tau} &  \tfrac{da}{d\tau}\\
		-\Delta_{N_f}/m_{N_f}^2  & \Delta_{N_f-1} \\
		\Lambda_{N_f}^2G_{N_f} & \Lambda_{N_f-1}^2G_{N_f-1} \\
		\hline
		\end{tabular}
		\caption{Decoupling flow of components of the $u$-plane integrand.} \label{Decoupling_uplane}\end{center}
\end{table}
The formulation of the $u$-plane integral in the presence of background fluxes introduces the further couplings $v_j$ and $w_{jk}$, which we defined in \eqref{Defvw} \cite{Manschot:2021qqe,Aspman:2022sfj} . While these are difficult to determine in general, we can study their behaviour under decoupling hypermultiplets using the semi-classical prepotential \eqref{prepotential}. If we send $m_{N_f}\to \infty$ while keeping $m_{N_f}\Lambda_{N_f}^{4-N_f}=\Lambda_{N_f-1}^{4-(N_f-1)}$ and $a$ fixed, we find that
\begin{equation}\label{expansionsfromF}
    \begin{aligned}
        &v^{j}_{N_f}\to v^j_{N_f-1},\\
        &v^{N_f}_{N_f}\to -\frac{1}{2}(n_{N_f}+1),\\
        &w^{jk}_{N_f}\to w^{jk}_{N_f-1}-\delta_{jk}\frac{i}{\pi}\log\left(\frac{\Lambda_{N_f-1}}{\Lambda_{N_f}}\right),\\
        &w^{N_fN_f}_{N_f}\to -\frac{i}{\pi}\log\left(\frac{m_{N_f}}{\Lambda_{N_f}}\right),\\
        &\tau_{N_f}\to \tau_{N_f-1},
    \end{aligned}
\end{equation}
for all $1\leq j< N_f$ and $1\leq k\leq  N_f$. The off-diagonal components of $w^{jk}$ only receive contributions from higher order terms in the large $a$ expansion of the prepotential, and we leave a determination of their decoupling limit for future work.

The point and surface observables $p$ and $\bfx$ are multiplying the dimensionless quantities $\frac{u}{\Lambda_{N_f}^2}$ and $G_{N_f}$. However, as is apparent from Table \ref{Decoupling_uplane}, they both rather need to be multiplied by $\Lambda_{N_f}^2$ in order to enjoy a well-defined scaling limit. This can be achieved by multiplying $p$ by $\left(\tfrac{\Lambda_{N_f}}{\Lambda_{N_f-1}}\right)^2$ and multiplying $\bfx$ by $\tfrac{\Lambda_{N_f}}{\Lambda_{N_f-1}}$.\footnote{Such redefinitions of the point and surface observables are familiar from superconformal rank one theories, such as the $\nstar$ theory \cite{Manschot:2021qqe}, the SU(2) $N_f=4$ theory \cite{moore_talk2018} and the $(A_1,A_2)$ Argyres-Douglas theory \cite{Moore:2017cmm}.} Then the resulting exponentials will simply flow to the ones for the theory with $N_f-1$ flavours.

With the above kept in mind, it is now straightforward to decouple every term in the expression \eqref{P2:generaluplaneintegral} separately. By multiplying it with the inverse normalisation $\CN_{N_f}^{-1}$, it becomes dimensionful, however all components decouple as given in Table \ref{Decoupling_uplane} and the result is $\CN_{N_f}^{-1}\Phi^{N_f-1}$. 
The decoupling $-m_{N_f}^{-2}\Delta_{N_f}\to \Delta_{N_f-1}$ tells us that we need to multiply $\Phi^{N_f}$ also by $(-m_{N_f}^{-2})^{\frac\sigma 8}$, which combines with the discriminant $\Delta_{N_f}^{\sigma/8}$ to have a well-defined limit. The minus sign is then absorbed in $\beta_{N_f,0}$, see \eqref{numerical_constants_Nf}. Using the definition of the double scaling limit $m_{N_f}\Lambda_{N_f}^{4-N_f}=\Lambda_{N_f-1}^{4-(N_f-1)}$, we find the useful relation which holds in the limit,
\begin{equation}\label{decoupling_prefactor}
    (-m_{N_f}^{-2})^{\frac\sigma 8}\frac{\CN_{N_f-1}}{\CN_{N_f}}=\left(\frac{\Lambda_{N_f}}{\Lambda_{N_f-1}}\right)^{3+\sigma}, 
\end{equation}
where the  exponent on the rhs for general four-manifolds is $3+\sigma=\frac14(7\sigma+3\chi)$.
This can be confirmed using the definition \eqref{CNNf}, with the numerical constants \eqref{numerical_constants_Nf} inserted, giving the normalisation factor $\CN_{N_f}$ for all $N_f$ and all $\sigma$,
\be \label{CNNf_explicit}
\CN_{N_f}=\frac{2^{1+3\sigma/4}\,e^{\pi i \sigma N_f/8}}{\pi\,\Lambda_{N_f}^{3+N_f\sigma /4}}.
\ee 
From \eqref{expansionsfromF} we also see that we have non-trivial decouplings of the couplings $v_{N_f}$ and $w_{N_f}^{jk}$. From the double scaling limit we have another useful formula
\begin{equation}
\label{dscalinglimit}
    \frac{\Lambda_{N_f}}{m_{N_f}}=\left(\frac{\Lambda_{N_f}}{\Lambda_{N_f-1}}\right)^{5-N_f},
\end{equation}
which combined with the decoupling limits \eqref{expansionsfromF} tells us that, to leading order ($C_{jk}=e^{-\pi iw_{jk}}$), 
\begin{equation}
    \begin{aligned}
        &C_{jj}^{(N_f)}\to \frac{\Lambda_{N_f}}{\Lambda_{N_f-1}}C_{jj}^{(N_f-1)},\quad j\neq N_f, \\
        &C_{N_fN_f}^{(N_f)}\to \left(\frac{\Lambda_{N_f}}{\Lambda_{N_f-1}}\right)^{5-N_f}.
    \end{aligned}
\end{equation}
The dependence on $v^j$ is only through the elliptic variable of the theta function \eqref{psi}, and from \eqref{expansionsfromF} we see that we pick up an extra phase
\begin{equation}\label{couplingfromv}
    e^{\pi i(n_{N_f}+1)B(\bfk_{N_f},\bfk)},
\end{equation}
in the decoupling of $v_{N_f}^{N_f}$. In Section \ref{sec:uplane_integral}  and in particular in  \eqref{u_plane_contraint} we concluded that the $u$-plane integral is only well-defined if the magnetic winding numbers $n_j$ satisfy $n_j\equiv -1\mod 4$. 
Using that $n_{N_f}+1\equiv 0\mod 4$, one finds that the phase \eqref{couplingfromv} equals 1, such that the decoupling does not introduce any additional phases in the theta function.

Combining everything, the decoupling limit of the full $u$-plane integral now reads 
\begin{equation}\boxed{\begin{aligned}\label{IR-decoupling}
    \lim_{\substack{m_{N_f}\to\infty \\ \Lambda_{N_f}\to 0}} \left(\frac{\Lambda_{N_f}}{\Lambda_{N_f-1}}\right)^{-\alpha} &\Phi^{J,N_f}_{\bfmu,\{\bfk_j\}_{N_f}}\left(\left(\tfrac{\Lambda_{N_f}}{\Lambda_{N_f-1}}\right)^2p,\tfrac{\Lambda_{N_f}}{\Lambda_{N_f-1}} \bfx,\bfm_{N_f},\Lambda_{N_f}\right) \\
    =\, \, \, &  \Phi^{J,N_f-1}_{\bfmu,\{\bfk_j\}_{N_f-1}}\left(p,\bfx,\bfm_{N_f-1},\Lambda_{N_f-1}\right),
\end{aligned}}\end{equation}
where we repeat the exponent
\be 
\alpha=\frac{1}{4}\Big(-3\chi-7\sigma+(5-N_f)\,c_1(\CL_{N_f})^2+\sum_{j=1}^{N_f-1}c_1(\CL_j)^2 \Big).
\ee 
from \eqref{alphadef}. We see that the decoupling matches precisely with the UV calculation \eqref{UV-decoupling}. In the case where all the $c_1(\CL_j)=0$, this reproduces the result \cite[(2.10)]{Marino:1998uy}.

\subsubsection*{Remarks about the phase of the partition function}
Our convention \eqref{decouplingConv} for the weak coupling limit is favourable since it is valid for all $N_f$ \cite{Aspman:2021vhs}. On the other hand, it differs from previous literature. Notably for $N_f=0$, $u$ differs by a sign, such that the monopole and dyon singularity are interchanged. 
As a consequence, the partition functions determined here differ by a phase compared to the literature. In particular,  $Z_{\bfmu}[e^{2pu}]$ differs from  $\left< e^{p \CO} \right>_z$ of \cite[Eq. (2.17)]{Witten:1994cg} by a phase,
\be
\label{relationOursLit}
Z_{\bfmu}[e^{2pu}]=-e^{\pi i \lambda/2}\,e^{2\pi i \bfmu^2} \left< e^{p \CO} \right>_z
\ee 
with $z=2\bfmu$, and $\alpha_a=0$. 

A related aspect is the choice of fundamental domain.
Reference \cite{Aspman:2021vhs} described a framework for mapping out the fundamental domain for $\CN=2$ SQCD with generic masses. Yet there is some ambiguity in the choice of this domain. In this brief subsection, we study this ambiguity here for $N_f=0$ and connect it to characteristic classes.

In the decoupling limit, it is important to choose a consistent frame for $\tau\to i\infty$ such that the decoupling does not involve shifts. The frame found in \cite{Aspman:2021vhs} differs from the one in the broad literature by $T^2$, or alternatively by the action of $r$, with $r$ the generator of the unbroken $\mathbb{Z}_4$ $R$-symmetry for non-vanishing $a$. Let us thus study the effect of this transformation on the $u$-plane integral. 

Let us denote by $I_\bfmu(\tau,p,\bfx)$ the integrand of \eqref{P2:generaluplaneintegral}, such that we have $\Phi_\bfmu(p,\bfx) =\int_{\CF_0}I_{\bfmu}(\tau,p,\bfx)$. Assuming that we are integrating over the 'standard' choice $\CF_0$ (see Fig. \ref{fig:fundgamma0(4)}), we can simply determine the difference between $I_{\bfmu}(\tau,p,\bfx)$ and $I_{\bfmu}(\tau+2,p,\bfx)$. Under $T^2$, we have the following transformations for $N_f=0$:
\begin{equation}\label{T2transformation}
   \begin{aligned}
    T^2:\begin{cases}
    \frac{da}{d\tau}&\to -i\frac{da}{d\tau} \\
    \frac{du}{da}&\to -i \frac{du}{da}
   \\
    \Delta^{1/8}&\to e^{-\pi i/4}\Delta^{1/8} \\
    u&\to -u \\
    \Psi_{\bfmu}^J(\tau,\bar \tau)\!\!\!\!\!\! &\to e^{-2\pi i \bfmu^2} \Psi_{\bfmu}^J(\tau, \bar \tau) \\
    \end{cases}
     \end{aligned} 
\end{equation}
Then, using $\chi+\sigma=4$ we find that the integrand $I_\bfmu$ of the general $u$-plane integral \eqref{P2:generaluplaneintegral} transforms as
\begin{equation}
\label{changeconvs}
    I_\bfmu(\tau+2,p)=ie^{-2\pi i \bfmu^2}I_\bfmu(\tau,-p).
\end{equation}
Thus, the correlation function $\Phi_\bfmu^J(p)$ computed with some frame $\tau$ and $\widetilde\Phi_\bfmu(p)=\int_{\CF_0}I_{\bfmu}(\tau+2,p)$ computed with a frame $\tau+2$ relative to the first, differ by a factor
\begin{equation}\label{PhiT2equiv}
\Phi_\bfmu^J(p)=ie^{-2\pi i \bfmu^2}\widetilde \Phi_\bfmu^J(-p)= ie^{-\frac{\pi i}{2} \int_X P_2(w_2(E)) } \widetilde \Phi_\bfmu^J(-p).
\end{equation}
with $P_2$ the Pontryagin square, $P_2:H^2(X,\mathbb{Z}_2)\to H^4(X,\mathbb{Z}_4)$.
This is the mixed anomaly between the $U(1)_R$ symmetry and the $\mathbb{Z}_2$ 1-form symmetry of the $N_f=0$ theory \cite{Cordova:2018acb}. Eq. (\ref{PhiT2equiv}) demonstrates that the shift $\tau\to \tau+2$ in the integrand couples the theory to an invertible TQFT \cite{Aharony:2013hda, Gaiotto:2014kfa}. It is straightforward to also include the dependence on the surface observable $\bfx$ here. Its transformation is $\bfx\to -i\,\bfx$.

We note that for $N_f>0$, the theories with fundamental matter do not have a $\mathbb{Z}_2$ 1-form symmetry. Instead given the background fluxes for the flavor symmetry group, $\{\bfk_j\}$, the 't Hooft flux $\bfmu\in (L/2)/L$ is fixed. A sum over $\bfmu\in (L/2)/L$ as occurs in gauging of the 1-form symmetry is thus not meaningful.

\section{Behaviour near special points}\label{sec:behaviour_special_points}
This section collects various data of the ingredients of the $u$-plane integral near special points, such as weak coupling, strong coupling, and branch points.  Readers mainly interested in the results of the evaluation can skip this section.

\subsection{Behaviour at weak coupling} 
\label{weakcoupexp}
The evaluation of $u$-plane integrals requires the expansion of  various quantities at weak coupling. We will concentrate on either small or large mass expansions. In the large mass expansion, we express various quantities in terms of the order parameter $u_0$ of the $N_f=0$ theory, or $u_{N_f}$ of the theory with $N_f$ flavours. 

For example, for $N_f$ large and equal masses $m$, we can find the exact coefficients of $u_{N_f}$ as functions of $u_{0}$ by making an ansatz $u_{N_f}=\sum_{n}f_n(u_{0})\,m^{-n}$ and iteratively find $f_n$ by satisfying the relation $\CJ_{N_f}(u_{N_f})=\CJ_{0}(u_{0})$ order by order in $m^{-1}$. We list the results below and in Appendix \ref{app:mass_expansions} for $N_f=1,2,3$. For the evaluation of the correlation functions in Sections \ref{sec:p2} and \ref{sec:SWcontributions} higher order terms  than presented are required.
\\
\\
{\it {\bf $\boldsymbol{N_f=1}$}}\\
We can consider the $q$-expansion of $u_1(\tau,m)$ as in \cite[Eq. (4.18)]{Aspman:2021vhs}. The coefficients of the $q$-series are polynomials in the mass. While they are easily determined to all orders, the modular properties are not manifest in this expansion. They are more apparent if we consider an expansion in the mass $m$. We find for the large mass expansion of $u_1$,
\begin{equation}\label{u1to0}
\begin{split}
    u_1(\tau,m)&=u_0(\tau)-\frac{1}{16}(4u_0(\tau)^2 -3\Lambda_0^4)m^{-2} \\
    &\quad -\frac{1}{2^7}u_0(\tau)(4u_0(\tau)^2-3\Lambda_0^4)m^{-4}+\CO(m^{-6}),
    \end{split}
\end{equation}
with $u_0$ as in (\ref{nf0parameter}). We observe that this expansion for $u_1$ obviously reduces to $u_0$ in the $m\to \infty$ limit. The expression is left invariant under $\Gamma^0(4)$ transformations since $u_0$ is a Hauptmodul for $\Gamma^0(4)$. Moreover, these terms are polynomial in $u_0$, such that $\mathbb{H}/\Gamma_0(4)$ with $\text{Im}\, \tau \ll \infty$ is a good fundamental domain for this regime. Further subleading terms are given in Table \ref{Nf1_large_m} in Appendix \ref{app:large_masses}.

For $(da/du)_{N_f=1}$, we find using the definition (\ref{dadu_def}) the following large mass expansion,
\be 
\left( \frac{da}{du}\right)_1= \left( \frac{da}{du}\right)_0\left(1+\frac{u_0}{8\,m^2}+\frac{10u_0^2+3\Lambda_0^4}{256\,m^4}+\CO(m^{-6})\right),
\ee 
where $(da/du)_0$ is the corresponding period for $N_f=0$. Further subleading terms are given in Table \ref{duda_large_m}.

In the presence of background fluxes, we also need the couplings $v$ and $w$. We determine expansions for these couplings from the prepotential $\CF(a, m)$. To this end, we determine using the Matone relation (\ref{Matone})  expansions for $a_1$ in terms of small $q=e^{2\pi i\tau}$ and large $m$. The $q$-series for fixed powers of $m$ can be identified with a quasi-modular form for the group $\Gamma^0(4)$. We find for the first few terms
\be 
\begin{split}
a_1(\tau,m)&=-i\Lambda_0 \left( \textstyle{\frac{\vartheta_2(\tau)^4+\vartheta_3(\tau)^4+2E_2(\tau)}{6\vartheta_2(\tau)\vartheta_3(\tau)}} \right)\\
&\quad  - \frac{i\Lambda_0^3}{288m^2}\left( \textstyle{ \frac{7\vartheta_2(\tau)^8+7\vartheta_3(\tau)^8-10\vartheta_2(\tau)^4\vartheta_3(\tau)^4+2(\vartheta_2(\tau)^4+\vartheta_2(\tau)^4)E_2(\tau)}{\vartheta_2(\tau)^3\vartheta_3(\tau)^3} }\right) +\dots.
\end{split}
\ee 
The leading term corresponds to the one for  $N_f=0$ \cite{Moore:1997pc}. Using this expression, we can verify the identities of observables derived from the SW curve, and from the prepotential \eqref{prepotential}, such as Eq. (\ref{udlambdadF}).
For the prepotential, we use the expansion of \cite{Ohta:1996hq} up to $a^{-18}$. As a result, expansions are valid up to about $\CO(q^{3})$.

Substitution of $a_1(\tau,m)$ in the couplings $w_1$ and $v_1$ provides large mass expansions for the couplings $C=C_{11}=e^{-\pi i w_{11}}$ \eqref{cij} and $e^{2\pi i v_1}$ \eqref{Defvw}. For the coupling $C$, we find 
\be
\label{Cseries}
\begin{split}
C&=\left(\frac{\Lambda_0}{m} \right)^{4/3} \left( 1-\left(\frac{\Lambda_0}{m}\right)^2\frac{2\vartheta_2(\tau)^4+2\vartheta_3(\tau)^4+E_2(\tau)}{12\vartheta_2(\tau)^2\vartheta_3(\tau)^2}+O(m^{-4})\right).
\end{split}
\ee 
For $v_1$, we find
\be 
\label{vseries}
v_1=-\frac{1}{\sqrt{2}\,\pi}\frac{\Lambda_0}{m}\frac{1}{\vartheta_2(\tau)\vartheta_3(\tau)}+\CO(m^{-2}),
\ee 
such that
\be
\label{wseries}
e^{2\pi i v_1}=1-\frac{2i}{\sqrt{2}}\frac{\Lambda_0}{m}\frac{1}{\vartheta_2(\tau)\vartheta_3(\tau)}+\CO(m^{-2}).
\ee
Using modular transformations, these expansions also provide large mass expansions for the couplings near the strong coupling singularities. 

Alternatively, we can make expansions for small $m$. Making only the $m$-dependence manifest, $u_1(\tau,m)$, we have
\be 
u_1(\tau, m)=u_1(\tau, 0)+\frac{3\Lambda_1^3}{8u_1(\tau, 0)}m-\left(\frac{1}{3}+\frac{9}{128\, u_1(\tau, 0)^2} \right)m^2+\CO(m^3),
\ee 
with $u_1(\tau, 0)$ given in Eq. (\ref{nf1u}). Here we see that this expansion is left invariant under the monodromy group which leaves $u_1(0)$ invariant \cite{Aspman:2020lmf,Aspman:2021vhs}. On the other hand $u_1(\tau, 0)$ vanishes for $\tau=\alpha=e^{2\pi i/3}$, such that this expansion is not a good function on the full domain. It would be interesting to understand the nature of these poles, which we leave for future work.
\\
\\
{\it {\bf $\boldsymbol{N_f=2}$}}\\
For $N_f=2$ with equal masses, $\bfm=(m,m)$, we consider first the large mass expansion, relevant for the decoupling $N_f=2\to 0$. From the exact expression for the order parameter \eqref{umm}, we find
\begin{equation}
\begin{split}
    u_2(\tau,\bfm)&=u_0(\tau)-\frac18(4u_0(\tau)^2-3\Lambda_0^4)m^{-2}+\frac18u_0(\tau)(u_0(\tau)^2-\Lambda_0^4)m^{-4}\\
    &\quad  -\frac{1}{2^7}u_0(\tau)(u_0(\tau)^2-\Lambda_0^4)^2m^{-8}+\CO(m^{-12}).
\end{split}
\end{equation}
Further subleading terms are listed in Table \ref{Nf2_large_m}.

We can also consider the mass $\bfm=(0,m)$, which is relevant for the decoupling limit from $N_f=2$ to $N_f=1$. For this choice, we find
\begin{equation}
\begin{split}
    u_2(\tau,\bfm)&=u_1(\tau,0)-\frac{2^7u_1(\tau,0)^3+3^3\Lambda_1^6}{3 \cdot 2^7u_1(\tau,0)}m^{-2}\\
    &\quad -\frac{\Lambda_1^6(2^7u_1^3(\tau,0)+3^3\Lambda_1^6)}{3\cdot 2^{15}u_1(\tau,0)^3}m^{-4}+\CO(m^{-6}).
    \end{split}
\end{equation}
This expansion is again singular for $\tau=\alpha$ since $u_1(\alpha,0)=0$. 
\\
\\
{\it {\bf $\boldsymbol{N_f=3}$}}\\
We can similarly determine large mass expansions for $N_f=3$. For equal masses $\bfm=(m,m,m)$, we have
\begin{equation}
\begin{split}
    u_3(\tau,\bfm)&=u_0(\tau)-\frac{3}{16}(4u_0(\tau)^2-3\Lambda_0^4)m^{-2}\\
    &\quad +\frac{3}{2^7}u_0(\tau)(20u_0(\tau)^2-19\Lambda_0^4)m^{-4}+\CO(m^{-6}).
\end{split}
\end{equation}
Further subleading terms are given in Table \ref{Nf3_large_m}.
Finally, for the large $m$ expansion of $\bfm=(0,0,m)$, we find
\begin{equation}
\begin{split}
    u_3(\tau,\bfm)&=u_2(\tau,0)-\frac{1}{2^7}(2^6u_2(\tau,0)^2-\Lambda_2^4)m^{-2}\\
    &\quad +\frac{u_2(\tau,0)}{2^9}(2^6u_2(\tau,0)^2-\Lambda_2^4)m^{-4} \\
    &\quad -\frac{u_2(\tau,0)}{2^{19}}(2^6u_2(\tau,0)^2-\Lambda_2^4)^2m^{-8} +\CO(m^{-10}).
\end{split}
\end{equation}

\subsubsection*{Singularities for \texorpdfstring{{\bf $\boldsymbol{N_f=1}$}}{Nf1}}
We also list expansions for the strong coupling singularities. For $N_f=1$, these are the roots of the discriminant which is a cubic equation, and can be determined explicitly. Their large and small mass expansions are
\be \label{nf1_sing_exp}
\begin{split} 
u^*_1&=\left\{ \begin{array}{r} \displaystyle \hspace{2cm} -  \Lambda_0^2 -\frac{1}{16} \frac{\Lambda_0^4}{m^2}+\frac{1}{128} \frac{\Lambda_0^6}{m^4}+\CO(m^{-6}),\\
\\
\displaystyle -\frac{3}{2^{8/3}}\Lambda_1^2-\frac{1}{2^{1/3}}\Lambda_1m+\frac{m^2}{3}+\CO(m^3),
\end{array}\right.
\\
&\\
u^*_2&= \left\{\begin{array}{r}\displaystyle \qquad \Lambda_0^2 -\frac{1}{16} \frac{\Lambda_0^4}{m^2}-\frac{1}{128} \frac{\Lambda_0^6}{m^4} +\CO(m^{-6}),\\
\\
\displaystyle -e^{4\pi i/3} \frac{3}{2^{8/3}}\Lambda_1^2-e^{2\pi i/3}\frac{1}{2^{1/3}}\Lambda_1m+\frac{m^2}{3}+\CO(m^3),
\end{array}\right.
\\
&\\
u^*_3&=\left\{ \begin{array}{r} \displaystyle m^2+\frac{1}{8}\frac{\Lambda_0^4}{m^2}+\CO(m^{-6}),\\
\displaystyle -e^{2\pi i/3} \frac{3}{2^{8/3}}\Lambda_1^2-e^{4\pi i/3}\frac{1}{2^{1/3}}\Lambda_1m+\frac{m^2}{3}+\CO(m^3).
\end{array}\right.
\end{split}
\ee 
The large  mass expansions for $u^*_1$ and $u^*_2$ agree with the expansion \eqref{u1to0} for $u_0\to \pm \Lambda_0^2$.

These singularities have two special properties. First, by Vieta's formula, $u_1^*+u_2^*+u_3^*=m^2$.  For general $N_f=0,\dots, 3$, we have 
\begin{equation}\label{Vieta_sing}
\sum_{j=1}^{2+N_f}u^*_j=\sum_{j=1}^{N_f}m_j^2+\frac{\Lambda_3^2}{2^8}\delta_{N_f,3}.
\end{equation}
More generally, if $P$ is a polynomial, then $\sum_{j=1}^{2+N_f}P(u_j^*)$ is a symmetric function in the $u_j^*$, and by the fundamental theorem of symmetric polynomials can be written as a rational function of the coefficients of the polynomial $\Delta_{N_f}$.

Furthermore, for $N_f=1$ we have a special case that the curve depends only on $\Lambda_1^3$. This means that the discriminant locus $\{u_1^*,u_2^*,u_3^*\}$ can only depend on $\Lambda_1^3$, while the individual $u_j^*$ depend explicitly only on $\Lambda_1$. This symmetry forces the $u_j^*=u_j^*(\Lambda_1)$ to depend on $\Lambda_1$ in a $\mathbb Z_3$ symmetric fashion,
\begin{equation}\label{Nf1_Z3sym}
    u_j^*(\alpha\Lambda_1)=u^*_{j+1}(\Lambda_1),
\end{equation}
with $\alpha=e^{2\pi i/3}$, and the labels $j$ being modulo 3. This holds as long as the mass $m$ is finite and generic. For instance, the expansions around $m=0$ in  \eqref{nf1_sing_exp} obey this symmetry. If we pick a specific mass, for instance $m=\mad=\frac34\Lambda_1$, this symmetry is broken. Furthermore, expanding around $m=\infty$ singles out the singularity $u_3^*$, which goes as $u_3^*\sim m^2$, while $u_1^*$ and $u_2^*$ are related under $\Lambda_1\mapsto \alpha \Lambda_1$. Thus  the infinite mass expansion \eqref{nf1_sing_exp} does not obey the $\mathbb Z_3$ symmetry  \eqref{Nf1_Z3sym}.

\subsection{Behaviour near strong coupling singularities}\label{sec:behaviour_strong_c_s}

We list various general formulas near the strong coupling singularities $u_j^*$, $j=1,\dots, N_f+2$. Similar formulas have also appeared for example in \cite{Marino:1998uy}. To analyze the behaviour of $u_{N_f}$ near $u_j^*$, we introduce a ``local" order parameter $u_{N_f,j}$, which is a function of the local coupling $\tau_j$. For example for $j=1$, $u_{N_f,1}(\tau_1)=u_{N_f}(-1/\tau_1)$. We let $u_1^*$ be the monopole singularity for $\tau\to 0$, $u_2^*$ the dyon singularity for $\tau\to 2$, and $j\geq 3$ label the additional hypermultiplet singularities.

From the  invariant $\CJ$ of the SW curve, we deduce that near a strong coupling singularity $u_j^*$ of Kodaira type $I_1$, $u_{N_f,j}$ reads
\be 
u_{N_f,j}(\tau_j)=u_j^*+(-1)^{N_f} \Lambda_{N_f}^{2N_f-8}\frac{12^3\,g_2(u^*_j)^3}{\prod_{\ell\neq j} (u^*_j-u^*_\ell)}\,q_j+\CO(q_j^2),\qquad q_j=e^{2\pi i \tau_j},
\ee 
where $u_{N_f,j}(\tau_j)\to u_j^*$ as $\tau_j\to i \infty$ (see \cite{Aspman:2021vhs} for details). The product in the denominator has $N_f+1$ terms.

We define the coupling $(da/du)_{N_f,1}$ near $u_1^*$ in terms of the weak coupling period  $(da/du)_{N_f}$ as
\be
\left( \frac{da}{du}\right)_{N_f,1}(\tau_1)=\tau_1^{-1} \left( \frac{da}{du}\right)_{N_f}(-1/\tau_1),
\ee 
and analogously near the other singularities. 
From \eqref{dadu_def} it follows that  near the strong coupling singularity $u_j^*$, the local expansion $(da/du)_{N_f,j}$  reads
\be 
\label{daduj}
\begin{split}
\left( \frac{da}{du}\right)_{N_f,j}(\tau_j)&=\frac{1}{6}\sqrt{\frac{g_2(u_j^*)}{g_3(u_j^*)}}+\CO(q_j)\\
&=\frac{1}{2\sqrt{3}}\,s_j^{1/2}\,g_3(u_j^*)^{-1/6}+\CO(q_j),
\end{split}
\ee 
where we introduced the phase $s_j$ as
\be
s_j\coloneqq s_j(\bfm)\coloneqq\frac{g_2(u_j^*(\bfm))}{3g_3(u_j^*(\bfm))^{\frac23}}.
\ee
The period $\frac{da}{du}$ evaluates thus to a constant at any $I_n$ singularity $u_j^*$.

The $s_j$ are locally constant functions, with phase transitions at AD points. For $N_f=1$ for instance, the phase changes depending on whether the ratio of $g_2$ and $g_3$ is calculated as a series with $m<\mad$ or $m>\mad$. More generally, this function is locally constant on $\mathbb R^{N_f}\setminus \CL^{\text{AD}}_{N_f}$, where $\CL^{\text{AD}}_{N_f}$ is the locus in mass space where AD points emerge on the $u$-plane (see \cite[Section 2.3]{Aspman:2021vhs}). 
 This is because $g_2$ and $g_3$ are strictly nonzero away from the AD locus $ \CL^{\text{AD}}_{N_f}$, and $s_j(\bfm)^3=1$ by definition. Thus for any $j$, 
\begin{equation}\label{s_j_phases}
    s_j: \mathbb R^{N_f}\setminus \CL^{\text{AD}}_{N_f}\longrightarrow \mathbb Z_3=\{1,\alpha,\alpha^2\}
\end{equation}
is a smooth function on the finite union $\mathbb R^{N_f}\setminus \CL^{\text{AD}}_{N_f}=\bigcup_i U_i$ of open sets, valued in $\mathbb Z_3$ and thus locally constant. 
For $N_f=2$ for instance, this partitions the real mass space into three regions, on which the phases $s_j$ \eqref{s_j_phases} are constant. In Fig. \ref{fig:DDlocus2} of Part I, these are the three regions separated by the AD locus (blue).
We list values of the $s_j$ in the  Tables \ref{sj_phasesLight} and \ref{sj_phasesHeavy} for $N_f=0,1$ and equal mass $N_f=2$.

\begin{table}
\parbox{.45\linewidth}{
\centering
\begin{tabular}{|l|r|r|r|}
\hline 
\diagbox[width=1.3cm, height=1cm]{$j$}{ $N_f$} & 0 & 1 & 2   \\
\hline 
$s_1$ & 1 & 1 & 1 \\
$s_2$ & $\alpha^2$ & $\alpha$ & 1 \\
$s_3$ & & $\alpha^2$ & $\alpha^2$\\
$s_4$ & &  & $\alpha^2$ \\
\hline
\end{tabular}
\caption{Table with values of $s_j$, $j=1,\dots, N_f+2$ for $N_f=0,1,2$, for small equal masses, $\alpha=e^{2\pi i/3}$.}
\label{sj_phasesLight}
}
\hfill
\parbox{.45\linewidth}{
\centering
\begin{tabular}{|l|r|r|r|}
\hline 
\diagbox[width=1.3cm, height=1cm]{$j$}{ $N_f$} & 0 & 1 & 2   \\
\hline 
$s_1$ & 1 & 1 & 1 \\
$s_2$ & $\alpha^2$ & $\alpha^2$ & $\alpha^2$ \\
$s_3$ & & 1 & 1\\
$s_4$ & &  & 1 \\
\hline
\end{tabular}
\caption{Table with values of $s_j$, $j=1,\dots, N_f+2$ for $N_f=0,1,2$, for large equal masses, $\alpha=e^{2\pi i/3}$.}
\label{sj_phasesHeavy}
}
\end{table}

The local coordinate $a_j=a-m_j/\sqrt{2}$ vanishes near the  singularity $u_j^*$. From \eqref{Matone} and \eqref{dadtau}, we obtain
\be \label{dadtaun=1}
\frac{da_j}{d\tau_j}= (-1)^{N_f+1} \Lambda_{N_f}^{2N_f-8} \frac{16\pi i}{4-N_f} \frac{2^3 g_2(u^*_j)^3}{P_{N_f}^{\rm M}(u^*_j)} \left(\frac{g_2(u^*_j)}{g_3(u^*_j)} \right)^{3/2}\,q_j+\CO(q_j^2).
\ee 
Therefore $a_j$ behaves as
\be \label{aj_I1}
\begin{split}
a_j&=(-1)^{N_f+1} \Lambda_{N_f}^{2N_f-8} \frac{64}{4-N_f} \frac{g_2(u^*_j)^3}{P_{N_f}^{\rm M}(u^*_j)} \left(\frac{g_2(u^*_j)}{g_3(u^*_j)} \right)^{3/2}\,q_j+\CO(q_j^2)\\
&=(-1)^{N_f+1} \Lambda_{N_f}^{2N_f-8}\frac{64}{4-N_f} s_j^{3/2} \frac{\left(27g_3(u^*_j)\right)^{3/2}}{P_{N_f}^{\rm M}(u^*_j)} \,q_j+\CO(q_j^2).
\end{split}
\ee 
Note that since $s_j$ is a third root of unity, $s_j^{3/2}=\pm 1$.

We can  generalise this to $u_j^*$ being an $I_n$ singularity, where for SU(2) $\CN=2$ SQCD the four cases  $n=1,2,3,4$ are possible.  Then we consider the expansions around $u_j^*$. The discriminant reads
\begin{equation}\begin{aligned}\label{discriminant_exp}
    \Delta_{N_f}(u)&=\frac{1}{n!}\Delta^{(n)}(u_j^*)(u-u_j^*)^n+\dots,
\end{aligned}\end{equation}
where $\Delta^{(n)}(u_j^*)=n!\prod_{l\neq j}(u_j^*-u_l^*)$.
The expansion of $u(\tau)$ we can read off from 
\begin{equation}\label{CJ}
    \CJ(u)=N_{N_f}\frac{g_2(u)^3}{\Delta_{N_f}(u)},
\end{equation}
with $N_{N_f}=12^3(-1)^{N_f}\Lambda_{N_f}^{2(N_f-4)}$. It is 
\begin{equation}\label{u_j_exp}
u=u_j^*+\left(n!N_{N_f}\frac{g_2(u_j^*)^{3}}{\Delta^{(n)}(u_j^*)}\right)^{\frac 1n}q_j^{\frac1n}+\dots,
\end{equation}
where we used \eqref{discriminant_exp}. Note that the coefficient of $q_j^{\frac{1}{n}}$ has an ambiguity by an $n$'th root of unity. We refrain from introducing another symbol for this ambiguity, but this should be kept in mind here and in the formulae below. The exact solution of the theory fixes the ambiguity, such as for $N_f=2$ with equal masses.

The discriminant $\Delta_{N_f}$ has leading term $q_j^1$ near each strong coupling singularity $u_j^*$, which can be read off form \eqref{CJ},
\begin{equation}
    \Delta_{N_f}=N_{N_f}g_2(u_j^*)^3 q_j+\dots. 
\end{equation}
This holds for any value of $n$.
In \cite{Aspman:2021vhs} we showed that $\Delta_{N_f}/P^{\text M}_{N_f}=\widehat \Delta_{N_f}/\widehat P^{\text M}_{N_f}$, where if $u_j^*$ is an $n$-th order zero of $\Delta_{N_f}$ of multiplicity $n$, then its multiplicity in $\widehat \Delta_{N_f}$ is 1. In $P^{\text M}_{N_f}$, it has multiplicity $n-1$, and therefore it is not a root of  $\widehat P^{\text M}_{N_f}$. Therefore, we can write $P^{\text M}_{N_f}(u)=(u-u_j^*)^{n-1}\widehat P^{\text M}_{N_f}(u)$ as a polynomial. 

Finally, the period $\frac{da}{du}$ evaluates to a constant for any $I_n$ singularity $u_j^*$. Using Matone's relation \eqref{Matone}, we then compute
\begin{equation}\label{dadtau1}
    \frac{da_j}{d\tau_j}=-\frac{2\pi i}{27(4-N_f)}\frac{N_{N_f}^{\frac1n}}{(n!)^{\frac{n-1}{n}}}\frac{g_2(u_j^*)^{\frac 3n+\frac 32}}{g_3(u_j^*)^{\frac32}}\frac{\Delta_{N_f}^{(n)}(u_j^*)^{\frac{n-1}{n}}}{\widehat P^{\text M}_{N_f}(u_j^*)}q_j^{\frac 1n}+\dots. 
\end{equation}
 This agrees exactly with the earlier result \eqref{dadtaun=1} for $n=1$.
Instead of using Matone's relation, we can also calculate $\frac{du}{d\tau}$ from \eqref{u_j_exp} directly. This gives a simpler result,
\begin{equation}\label{dadtau2}
    \frac{da_j}{d\tau_j}=\frac{2\pi i}{6n}(n!)^{\frac 1n}N_{N_f}^{\frac1n}\frac{g_2(u_j^*)^{\frac 3n+\frac 12}}{g_3(u_j^*)^{\frac12}}\frac{1}{\Delta_{N_f}^{(n)}(u_j^*)^{\frac1n}}q_j^{\frac 1n}+\dots.
\end{equation}
Identifying both leading terms \eqref{dadtau1} and \eqref{dadtau2}, we find the interesting relation
\begin{equation}\label{phat_Deltan}
    \widehat P^{\text M}_{N_f}(u_j^*)=-\frac{2}{9(4-N_f)(n-1)!}\frac{g_2(u_j^*)}{g_3(u_j^*)}\Delta_{N_f}^{(n)}(u_j^*).
\end{equation}
We checked this relation for various mass configurations with $n=1,2,3,4$. 
It is important to stress that it only holds on the discriminant locus $\Delta_{N_f}(u_j^*)=0$. 

Using \eqref{dadtau2} and eliminating $g_2(u_j^*)$ as above, this gives for the local coordinate
\begin{equation}\begin{aligned}\label{a_j_local}
    a_j&=\frac{(3s_j)^{\frac3n+\frac12}}{6}(n!\, N_{N_f})^{\frac1n}\frac{g_3(u_j^*)^{\frac 2n-\frac 16}}{\Delta_{N_f}^{(n)}(u_j^*)^{\frac1n}}\, q^{\frac 1n}+\dots \\    
    &=\frac{(3s_j)^{\frac 12}}{6 g_3(u_j^*)^{\frac 16}}\left(n!\,N_{N_f}\frac{g_2(u_j^*)^{3}}{\Delta^{(n)}(u_j^*)}\right)^{\frac 1n}q_j^{\frac1n}+\dots
\end{aligned}\end{equation}

This agrees precisely with \eqref{aj_I1}  for $n=1$.
It also agrees with an explicit calculation of the asymptotics at the $I_2$ singularity in $N_f=2$ with  $\bfm=(m,m)$, again keeping in mind the $n$'th root of unity ambiguity. 
The $n$-dependence of the leading term in \eqref{a_j_local} is in fact the same as that of $u-u_j^*$, \eqref{u_j_exp}, as can be seen from the second line:  Up to the $g_3$ prefactor, $a_j\sim u-u_j^*$. 

This concludes our analysis of Coulomb branch functions near strong coupling singularities. Such expansions are relevant for the contributions of the  singularities to the $u$-plane integral as well as the SW functions. For the latter,  in some cases subleading corrections are required, for instance for the SW contributions of $I_n$ singularities, as we discuss in Subsection \ref{sec:SWI2}. These corrections can in principle be determined by a perturbative analysis similar to the above. In some examples, exact expressions of CB functions are available, and we can use the previous calculation for consistency checks.

\subsection{Behaviour near branch points} \label{sec:behaviour_bp}

The fundamental domain for $N_f$ generic masses contains $N_f$ pairs of branch points, connected by branch cuts \cite{Aspman:2021vhs}. In section \ref{sec:integrationFD} of part  I, we demonstrated that  branch points do not contribute to the $u$-plane integral, based on the assumption \eqref{branch_point_singularity} that the integrand is sufficiently regular near a given branch point.

In this subsection, we provide explicit evidence for this assumption, in the rather generic example of $N_f=2$ with equal masses $m$. For this configuration, the full integrand (without the couplings to the $\spinc$ structure) can be expressed as a modular form, which facilitates the analysis. We assume in the following that  $m>0$ with $m\neq \mad=\frac{\Lambda_2}{2}$, such that we are strictly away from the AD locus where the branch points collide and annihilate each other. 

After the exact analysis of the equal mass case, we then formulate the asymptotics of the general integrand near a generic branch point, and prove that the assumption \eqref{branch_point_singularity} is always satisfied and thus branch points never contribute to the integrals.

\subsubsection*{Branch points of $u$}
Consider the equal mass case in $N_f=2$. We list the relevant modular forms in Appendix \ref{sec:equal_mass_expansion}.
In \cite{Aspman:2021vhs} we found that the effective coupling of the branch point $\ubp=2m^2-\frac{\Lambda_2^2}{8}$ is determined by $f_{2\text{B}}(\tbp)=0$, such that \footnote{The definition of $f_{2\text{B}}$ differs slightly from \cite{Aspman:2021vhs}.}
\begin{equation}
\frac{u(\tbp)}{\Lambda_2^2}=-\frac{f_{2\text{B}}(\tbp)+16}{128}, \qquad f_{2\text{B}}(\tau)=\frac{16\jt_4(\tau)^8}{\jt_2(\tau)^4\jt_3(\tau)^4}.
\end{equation}
Then $u(\tbp)=\ubp$ is solved by 
\begin{equation}\label{f2btbp}
\tbp\in f_{2\text{B}}^{-1}\left(-2^8\tfrac {m^2}{\Lambda_2^2}\right)\cap \CF_2(m,m),
\end{equation}
that is, we find all preimages of the branch point on the fundamental domain $\CF_2(m,m)$.
Since $f_{2\text{B}}$ is a Hauptmodul for the index 3 group $\Gamma_0(2)$, inside the index 6 domain $\CF_2(m,m)$ there are consequently two distinct points $\tbp$. Using \eqref{f2btbp}, we can eliminate all but one of the Jacobi theta functions $\jt_j$ from \eqref{nf2mmquantities} to find
\begin{equation}\label{dadutbp}
\frac{da}{du}(\tbp)=-\frac{\im}{4}\frac{\jt_4(\tbp)^2}{\sqrt m\,(m^2-\mad^2)^{\frac14}}.
\end{equation}
Since $\jt_4$ is holomorphic and nowhere vanishing on $\mathbb H$, $\frac{da}{du}(\tbp)$ is never zero or infinite. 
   
From Matone's relation  \eqref{Matone}
we see that $\frac{du}{d\tau}$ diverges as $\CO\left((u-\ubp)^{-1}\right)$, since $\widehat\Delta(\tbp)=4m^2(m^2-\mad^2)^2$ remains finite. For $\tau$ near $\tbp$, we can integrate this equation to find
\begin{equation}
\frac{du}{d\tau}(\tau)=e^{\frac{\pi\im}{4}}\sqrt {\pi m}(m^2-\mad^2)^{\frac34}\frac{\jt_4(\tbp)^2}{\sqrt{\tau-\tbp}}+\CO(1)
\end{equation}
for $\tau\to\tbp$.
This is sufficient to study the $u$-plane integrand  near $\tbp$.
From $\frac{da}{d\tau}=\frac{da}{du}\frac{du}{d\tau}$ and $\sigma+\chi=4$, we have that $\nu=\frac{du}{d\tau}\left(\frac{da}{du}\right)^{-1+\frac\sigma2}\Delta^{\frac\sigma8}$. The discriminant $\Delta(\tbp)=4m^2(m^2-\mad^2)^3$ is regular and nonzero. Thus the power series of $\nu$ at $\tbp$ reads\footnote{We ignore the constants $\alpha_{N_f}$ and $\beta_{N_f}$  here since they are irrelevant to the analysis.}
\begin{equation}
\nu(\tau)=(-1)^{\frac{3+7\sigma}{4}}2^{2-\frac{3\sigma}{4}}\sqrt\pi \,m(m^2-\mad^2)^{1+\frac\sigma4}\frac{\jt_4(\tbp)^\sigma}{\sqrt{\tau-\tbp}}+\CO(1).
\end{equation}

Regarding the photon path integral, let us assume that we can express $ \Psi(\tau,\bar\tau,\bfz,\bar\bfz)=\partial_{\bar\tau}\widehat G(\tau,\bar\tau,\bfz,\bar\bfz)$, 
then the function 
\begin{equation}
\widehat h(\tau,\bar\tau)=\nu(\tau) \widehat G(\tau,\bar\tau,\bfz,\bar\bfz) \,e^{2pu(\tau)/\Lambda_{N_f}^2+\bfx^2 G_{N_f}(\tau)}
\end{equation}
provides an anti-derivative of the integrand, as required in Section \ref{sec:integrationFD}. 
The other factors of the integrand are regular: Due to \eqref{dadutbp}, $\bfz(\tbp)$ and thus $\widehat G(\tbp,\bfz(\tbp))$ are regular. The contact term \eqref{contactterm} becomes a constant  as well, for the same reason. Thus, up to constants, we find
\begin{equation}\label{tbp_expansion_Imu}
\widehat h(\tau)\sim  m(m^2-\mad^2)^{1+\frac\sigma4}\widehat G(\tbp,\bfz(\tbp))e^{2p\ubp+\bfx^2G(\tbp)}\frac{\jt_4(\tbp)^\sigma}{\sqrt{\tau-\tbp}}+\CO(1),
\end{equation}
in the notation \eqref{generaluplaneintegral}. This shows that the integrand $\widehat h$ diverges at $\tbp$, however in a subcritical fashion  $\sim (\tau-\tbp)^{-\frac12}$. Equation \eqref{tbp_expansion_Imu} also suggests that the integrand is not single-valued at $\tbp$. However, a small circular path around $\tbp$ in $\CF_2(m,m)$ describes a curve of angle $4\pi$ or winding number $2$, as is clear for example from Fig. \ref{fig:nf2domain}. Since $\widehat h(\tau)$ has a Laurent series in $\sqrt{\tau-\tbp}$, it is single-valued around such a path. Section \ref{sec:integrationFD} in part I then guarantees that the branch point does not contribute to the $u$-plane integral.

\subsubsection*{Branch points of $\frac{da}{du}$}

Another potential source of branch points is the period $\frac{da}{du}$.
Even if the masses $\bfm$ are such that $u(\tau)$ is modular for a congruence subgroup, $\frac{da}{du}$ is  in general not modular. This is due to the square root $\frac{da}{du}\sim \sqrt{\frac{g_2}{g_3}\frac{E_6}{E_4}}$ and the possible roots in $u$. Let us study if the square root in $\frac{da}{du}$ introduces another branch point, in the example of $\bfm=(m,m)$.  From \eqref{nf2mmquantities} we find that any solution to $\jt_2^4+\jt_3^4+\sqrt{f_2}=0$ is a branch point of $\frac{da}{du}$. Necessarily but not sufficiently, $(\jt_2^4+\jt_3^4)^2=f_2$, whose only solution is in fact independent on $\tau$, it is $m=\mad$. Since we exclude the case $m=\mad$ to study the branch points, the denominator of $\frac{da}{du}$ is never zero in $\mathbb H$. This agrees with the observation \cite{Aspman:2021vhs} that zeros of $\frac{du}{da}$ in $\mathbb H$ are AD points, since there are none on $(0,\infty)\backslash \{\mad\}$. The other radicand in $\frac{da}{du}$ is $f_2$, whose zeros are studied above. From \eqref{dadutbp} we know that $\jt_2^4(\tbp)+\jt_3^4(\tbp)$ is nonzero, otherwise $\frac{da}{du}$ would have a pole at $\tbp$. We have shown that $\tbp$ is the only branch point (of a square root) for $\frac{da}{du}$, i.e. $\frac{da}{du}$  has a regular series in $\sqrt{\tau-\tbp}$.

\subsubsection*{Branch points of the integrand}
With the intuition from the equal mass case, we can formulate the behaviour of the general $u$-plane integrand around a branch point. As pointed out in \cite[Section 3.3]{Aspman:2021vhs}, for $N_f$ generic masses there are $N_f$ pairs of branch points connected by  branch cuts. The branch points correspond in all cases to a \emph{square root} of $u(\tau)$. Let us assume that $\ubp=u(\tbp)$ is a branch point that is not simultaneously an AD point.\footnote{This only excludes the case in $N_f=3$ where the branch point is also an AD point of type $IV$. Our argument works regardless, as the integrand becomes less singular in that case.}
The expansion of $u(\tau)$ at $\tau=\tbp$ thus reads
\begin{equation}
    u(\tau)=\ubp+ c_{\text{bp}}(\tau-\tbp)^{\frac12}+\dots.
\end{equation}
Then it is clear that 
\begin{equation}
    u'(\tau)\sim(\tau-\tbp)^{-\frac12}.
\end{equation}
On the other hand, from $\eta^{24}\sim \left(\frac{da}{du}\right)^{12}\Delta_{N_f}$ \cite{Aspman:2021vhs} it is clear from $\eta(\tbp)\neq 0$ and $\Delta_{N_f}(\ubp)\neq 0$ that $\frac{da}{d\tau}(\tbp)\neq 0$ is a nonzero constant. Thus $\frac{da}{d\tau}=\frac{da}{du}\frac{du}{d\tau}$ has the same asymptotics at $\tbp$ as $\frac{du}{d\tau}$. 
Excluding the couplings to the background fluxes, from \eqref{measurefactornf} it is then clear that
\begin{equation}
    \nu(\tau)\sim (\tau-\tbp)^{-\frac12}.
\end{equation}
Since $\frac{du}{da}(\tbp)\neq 0$ we have $\bfz(\tbp)\neq0$ and we thus expand the non-holomorphic modular form  $ \widehat G$ at a regular point. It can accidentally vanish, but by varying $\tbp$ slightly the value is generically nonzero. In either case, we have 
\begin{equation}\label{bp_bound}
    \widehat h(\tau,\bar\tau)\sim (\tau-\tbp)^n, \qquad n\geq -\frac12.
\end{equation}
This justifies the assumption made in Section \ref{sec:integrationFD} and demonstrates that branch points never contribute to $u$-plane integrals. 

The bound $n\geq -\frac12$ is not sharp, indeed, as long as $n>-1$ the branch point will not contribute. Consider for instance a theory which includes  branch points of a $k$-th root of $u$, with $k\in\mathbb N$. In that case, there is no contribution either, since $n=\frac1k-1>-1$. 

Finally, since we lack modular expressions for the extra couplings  $v_j$ and $w_{jk}$, we leave it for future work to determine whether those couplings have branch points or singularities.

\section{\texorpdfstring{$u$}{u}-plane integrals and (mock) modular forms}\label{sec:uplane_int_mock}
We proceed by discussing the evaluation of the $u$-plane integrals near the different special points. 
Below we consider  $N_f$ generic masses, and in particular $\bfm \not\in \CL_{N_f}^{\text{AD}}$. To explicitly evaluate the $u$-plane integral we make use of the theory of (mock) modular forms \cite{Donaldson90,Malmendier:2012zz,Griffin:2012kw,Korpas:2017qdo, Korpas:2019cwg}, as discussed in Section \ref{sec:integrationFD}. As before, we specialise to four-manifolds with $(b_1,b_2^+)=(0,1)$.

\subsection{Fundamental domains}
\label{sec:funddom}
In the vicinity of special points, the fundamental domains simplify. We will consider here two cases, namely the large mass expansion and the small mass expansion. 

Verifying the IR decoupling limit \eqref{IR-decoupling} through the $u$-plane integral requires  a  precise definition of the integral \eqref{P2:generaluplaneintegral}. Specifically, the integrand as well as the integration domain must be determined in a region which is compatible with the decoupling limit. 
As found in \cite{Aspman:2021vhs}, when $m_{N_f}\to \infty$ there is always a branch point $\tbp$ whose imaginary part $\ybp=\ima(\tbp)$ grows as a function of $m_{N_f}$.  If $m_{N_f}$ is large, we can take expansions of the Coulomb branch parameters in two regions. For $\CN=2$ SQCD, implicit, yet exact, expressions for $\tbp$ have been determined in \cite{Aspman:2021vhs} (see \eqref{f2btbp} for an example). In the region with $\ima(\tau)>\ybp$, the order parameter $u(\tau)$ has periodicity $4-N_f$. For $\ima(\tau)<\ybp$ rather, the periodicity is that of the decoupled theory, which is $4-(N_f-1)$. Since in the limit $m_{N_f}\to\infty $ the periodicity at $u=\infty$ is $4-(N_f-1)$, we need to choose a cutoff $Y_-<\ybp$ for the fundamental domain $\CF(\bfm_{N_f})$ in order to find the consistent limit. 

In order to integrate over the whole fundamental domain, we must take the cutoff $Y\to\infty$. If we choose $\ybp> Y\to\infty$, then necessarily $m_{N_f}\to\infty$. In other words, choosing the cutoff $Y<\ybp$ is only a consistent choice in the decoupling limit. For a finite mass $m_{N_f}$, we can choose $Y>\ybp$, such that for $Y\to\infty$ we do not cross the branch point(s). This is illustrated in Fig. \ref{fig:nf1cutoff} for the example of the decoupling in $N_f=1$. In making a large mass expansion, we will assume that $1/m$ is infinitesimally small, such that $\ybp\to \infty$, and disappears from the fundamental domain. 

\begin{figure}[ht]\centering
	\includegraphics[scale=1.3]{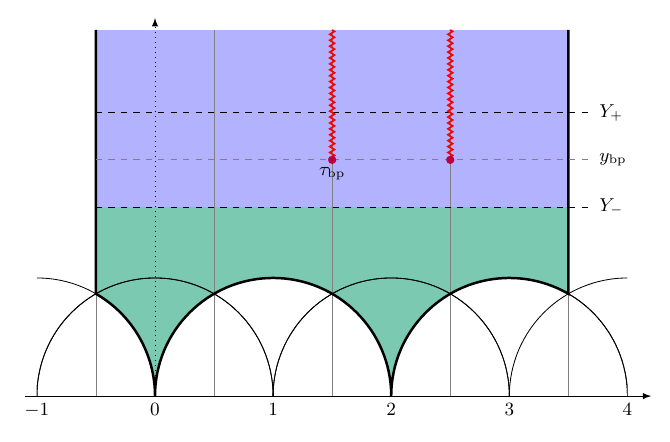}
	\caption{For the definition of the $u$-plane integral over the fundamental domain $\CF(m)$ in $N_f=1$ with a large mass $m$, there are two different choices $Y_\pm$ for the cutoff, with $Y_+>\ybp$ or $Y_-<\ybp$ and $\ybp=\ima(\tbp)$ the imaginary part of the branch point. The integration requires taking the limit $Y\to\infty$. The green region is the one we choose for the decoupling limit. }\label{fig:nf1cutoff}
\end{figure}

Let us denote the regulated fundamental region by
\begin{equation}
    \CF_Y(\bfm)=\{\tau\in\CF(\bfm)\, |\, \ima(\tau)<Y\}.
\end{equation}
We can choose two different cutoffs $Y_\pm$, with $Y_+>\ybp$ or $Y_-<\ybp$, which serve two different purposes. 
For any finite $\bfm$, we define the integral \eqref{P2:generaluplaneintegral} as  (we suppress most variables for clarity) 
\begin{equation}
\Phi_\bfmu(\bfm)=\lim_{Y_+\to\infty}\int_{F_{Y_+}(\bfm)}I_\bfmu(\tau,\bar\tau),
\end{equation}
and renormalise it as described in \cite{Korpas:2019ava}. As reviewed in Section \ref{sec:integrationFD}, the contribution from the arc at ${\rm Im}(\tau)=Y\to \infty$ is the constant term of the holomorphic part of the anti-derivative $h(\tau)$,\footnote{We use the notation also developed in \cite{kim2023}.} 
\be 
\Phi^{\infty}_\bfmu(\bfm)={\rm Coeff}_{q^0}[h(\tau,\bfm)].
\ee 

In the decoupling limit rather, we make an expansion in $1/m_{N_f}$ of the integrand, and integrate over $\CF_{Y_-(\bfm)}$ with $Y_-\to \infty$ term by term in the expansion. This results in 
\be 
\label{PhiLargeM}
\Phi^{\infty}_\bfmu(\bfm)={\rm Coeff}_{q^0}\,{\rm Ser}_{{m_{N_f}^{-1}}}[h(\tau,\bfm)],\qquad \bfm=(m_{1},\dots,m_{N_f-1},m_{N_f}).
\ee 
One can similarly make mass expansions near other special points in mass space, such as distinct small masses, 
\be 
\label{Phi_specificM}
\Phi^{\infty}_\bfmu(\bfm)={\rm Coeff}_{q^0}\,{\rm Ser}_{\bfm}[h(\tau,\bfm)]
\ee 
or AD mass $\bfm\in\CL^{\text{AD}}_{N_f}$. We will find that these prescriptions agree in many examples. However, in some cases they also lead to different results. To avoid cluttering, we have chosen not to add additional labels to $\Phi^{\infty}_\bfmu(\bfm)$ to specify the evaluation prescription. We will rather specify this when we present the results.

 For $N_f>1$, there are of course generally $N_f>1$ branch points $\tau_{\text{bp},j}$, $j=1,\dots, N_f$. The above analysis then proceeds with $Y_+ > \max_j\text{Im}\, \tau_{\text{bp},j}$ and $Y_- < \min_j\text{Im}\, \tau_{\text{bp},j}$.

\subsection{Factorisation of \texorpdfstring{$\Psi_{\bfmu}^J$}{Psi}}\label{factorisation_psi}
One can split the study of $u$-plane integrals into two classes of four-dimensional manifolds, depending on their intersection form being  even or odd (see Section \ref{4manifolds} for relevant aspects of four-manifolds). We can use the analysis of \cite[Sec. 5]{Korpas:2019cwg}  without much alteration and we simply outline the rough ideas. For simplicity, we only consider the odd lattices, and refer to \cite{Korpas:2019cwg} for the case of even lattices. The first important step is to factorise the indefinite theta function appearing in the $u$-plane integrand.  For odd intersection form we can diagonalise the  quadratic form to 
\begin{equation}\label{oddsignature}
Q=\langle 1\rangle \oplus (b_2-1)\langle-1\rangle.
\end{equation} 
This implies that the components $K_j$ of a characteristic vector $K$ are odd for all $j=1,\dots,b_2$.\footnote{Proof: We have that $\bfk^2+B(K,\bfk)=\sum_{j=1}^{b_2}(2\delta_{j,1}-1)k_j(k_j+K_j)$. Let $k_j=0$ for $j\neq m$. Then $k_m(k_m+K_m)\in2\mathbb Z$. If $k_m$ is odd, then $k_m+K_m$ must be even and therefore $K_m$  is odd.}
The lattice $L$ can be factorised as $L=L_+\oplus L_-$, where $L_+$ is a one-dimensional positive definite lattice and $L_-$ is a $(b_2-1)$-dimensional negative definite lattice. The polarisation corresponding to this decomposition is $J=(1,\boldsymbol 0)$, where $\boldsymbol{0}$ is the $(b_2-1)$-dimensional zero-vector. We will also employ the notation $\bfk=(k_1,\bfk_-)\in L$ where $k_1\in\BZ+\mu_1$, $\bfk_-\in L_-+\bfmu_-$ and $\bfmu=(\mu_1,\bfmu_-)$.

The sum over fluxes (\ref{psi}) now factorises as \cite[Eq.(5.45)]{Korpas:2019cwg}
\begin{equation}
\Psi_\bfmu^J(\tau,\bar\tau,\bfrho,\bar\bfrho)=-\im (-1)^{\mu_1(K_1-1)}f_{\mu_1}(\tau,\bar\tau,\rho_+,\bar\rho_+)\Theta_{L_-,\bfmu_-}(\tau,\bfrho_-),
\end{equation}
where
\begin{equation}
\begin{aligned}\label{ThetaL-}
f_\mu(\tau,\bar\tau,\rho,\bar\rho)=& \, \im e^{\pi \im \mu} e^{-2\pi yb^2}\sum_{k\in\BZ+\mu}\partial_{\bar{\tau}}\left(\sqrt{2y}(k+b)\right)(-1)^{k-\mu}\bar{q}^{k^2/2}e^{-2\pi i\bar\rho k}, \\
\Theta_{L_-,\bfmu_-}(\tau,\bfrho_-)=&\sum_{\bfk_-\in L_-+\bfmu_-}(-1)^{B(\bfk_-,K_-)}q^{-\bfk_-^2/2}e^{-2\pi i B(\bfrho_-,\bfk_-)}.
\end{aligned}
\end{equation}
If the elliptic variable $\rho$ is zero, $\Psi_\bfmu$ vanishes unless $\bfmu=(\tfrac12,\bfnull)\mod \mathbb Z^{b_2}$ \cite{Korpas:2019cwg}. In that case, it evaluates to
\begin{equation}\label{psi120odd}
\Psi_{(\frac12,\bfnull)}(\tau,\bar\tau)=\frac{\im^{K_1+1}}{2\sqrt{2y}} \overline{\eta(\tau)^3} \jt_4(\tau)^{b_2^-}.
\end{equation}
We will also need the dual theta series
\begin{equation}
\begin{aligned}
\Theta_{D,L_-,\bfmu_-}(\tau,\bfrho_{D,-})=\sum_{\bfk_-\in L_-+K_-/2}(-1)^{B(\bfk_-,\bfmu_-)}q^{-\bfk_-^2/2}e^{-2\pi iB(\bfrho_{D,-},\bfk_-)}.
\end{aligned}
\end{equation}

In order to write the integrand as a total anti-holomorphic derivative  one can use the theory of mock modular forms and Appell--Lerch sums \cite{Moore:1997pc, Malmendier:2012zz, Korpas:2017qdo, Korpas:2019cwg, Manschot:2021qqe}. 
An important constraint is that the anti-derivative must be a well-defined function on the fundamental domain for $\tau$, and thus transform appropriately under duality transformations of the theory.

As discussed in Section \ref{weakcoupexp}, in the large mass limit the duality group is $\Gamma^0(4)$, such that we can use results for the $N_f=0$ theory \cite{Korpas:2019cwg}. We write $f_\mu(\tau,\bar\tau,\rho,\bar\rho)$ as 
\be 
f_\mu(\tau,\bar\tau,\rho,\bar\rho)=\frac{\partial \widehat F_{\mu}(\tau,\bar \tau,\rho,\bar \rho)}{\partial \bar \tau},
\ee 
with $\widehat F_\mu$ a specialisation of the Appell--Lerch sum $M$ and its completion, which we define in Appendix \ref{app:mock}.
The holomorphic parts of $\widehat{F}_{\mu}$ are given by \cite[Eqs.(5.51) and (5.53)]{Korpas:2019cwg}
\begin{equation}
\begin{aligned}\label{f12function}
F_{\frac{1}{2}}(\tau,\rho)=&-\frac{w^{\frac{1}{2}}}{\vartheta_4(\tau)}\sum_{n\in\BZ}\frac{(-1)^nq^{n^2/2-\frac{1}{8}}}{1-wq^{n-\frac{1}{2}}}, \\
F_{0}(\tau,\rho)=&\frac{\im}{2}-\frac{\im}{\vartheta_4(\tau)}\sum_{n\in\BZ}\frac{(-1)^nq^{n^2/2}}{1-wq^n},
\end{aligned}
\end{equation}
where $w=e^{2\pi i \rho}$.

To evaluate the contributions from the strong coupling cusps, we introduce furthermore the ``dual" functions $F_{D,\mu}$ \cite[Equations (5.63) and (5.64)]{Korpas:2019cwg},
\begin{equation}\label{f12functionD}
F_{D,\mu}(\tau,\rho)=-\frac{w^{1/2}}{\vartheta_2(\tau)}\sum_{n\in\BZ}\frac{q^{n(n+1)/2}}{1-(-1)^{2\mu}w q^n}.
\end{equation}
We note that $F_{\frac{1}{2}}$ has a finite limit for $\rho \to 0$, 
\be 
F_{\frac{1}{2}}(\tau)=\lim_{\rho\to 0}F_{\frac{1}{2}}(\tau,\rho).
\ee
If the subscript $\mu$ is clear from the context, we will occasionally drop it and denote $F=F_{\frac12}$.
The first terms of the $q$-series are
\be 
F_{\frac{1}{2}}(\tau)=2q^{\frac38}\left(1+3q^{\frac12}+7q+14q^{\frac32}+\CO(q^4)\right).
\ee
This $q$-series is proportional to the McKay-Thompson series $H^{(2)}_{1A,2}$ \cite{Cheng:2012tq}. See also the OEIS sequence A256209. 

The duality groups are different for small masses, or other special points in the mass space. For such cases, other anti-derivatives are required. The most widely applicable anti-derivative will transform under $\slz$. As we review in detail in Appendix \ref{app:mock}, anti-derivatives are not unique one can add an integration constant, i.e. a weakly holomorphic function of $\tau$. There are in fact three well-known mock modular forms with precisely the same shadow $\sim \eta^3$, namely $F$, $\frac{1}{24}H$ and $\frac12 Q^+$.\footnote{We may view the relation between $H$, $F$ and $Q^+$ as follows. Following \cite[Section 5]{zagier2009}, there is a short exact sequence
\begin{equation}
0\longrightarrow M_k^! \longrightarrow\mathbb M_k\overset{\mathscr S}{\longrightarrow}  M_{2-k},
\end{equation}
where $M_k$ is the space of all classical modular forms of weight $k$, $M_k^!$ is the larger space of weakly holomorphic modular forms of weight $k$, and $\mathbb M_k$ is the space of mock modular forms of weight $k$. The `shadow map' $\mathscr S$ associates the shadow $\mathscr S[h]$ to each mock modular form $h\in \mathbb M_k$. This shadow map is surjective, but clearly not injective, since $\mathscr S[F_{\frac12}]=\mathscr S[\frac{1}{24}H]=\mathscr S[\frac12Q^+]\propto\eta^3$. We are thus considering the preimage $\mathscr S^{-1}[\eta^3]$, which contains $F_{\frac12}$, $H$ and $Q^+$. As we will discuss in more detail in Section \ref{sec:AD_photon_path_integral}, their differences can be understood as integration `constants'. Roughly speaking, we can understand $\mathscr S^{-1}[\Psi]$ as the space of anti-derivatives of $\Psi$.}

 Their completions are non-holomorphic modular functions for $\Gamma^0(2)$, $\slz$ and $\Gamma^0(2)$, respectively.  In \cite{Korpas:2019cwg} it was shown that for $N_f=0$ either of these three functions can be used for the evaluation of $u$-plane integrals, and they give the same result. This is possible because all three functions transform well under the monodromies on the $u$-plane. For $N_f=1$ and $N_f=3$ on the other hand, $F$ and $Q^+$ do not have the right monodromy properties, since they do not transform under $T^3$ or $T$. This singles out the function $H$, which transforms under all possible monodromies for all $N_f$. 

The function $H$ is related to $F_{\frac{1}{2}}$ as
\be 
\label{HtauDef}
H(\tau)=24\,F_\frac{1}{2}(\tau)-2\frac{\vartheta_2(\tau)^4+\vartheta_3(\tau)^4}{\eta(\tau)^3},
\ee 
with $F_{\frac{1}{2}}$ as above. 
This function is well known as the generating function of dimensions of representations of the Mathieu group \cite{Eguchi:2010ej, Dabholkar:2012nd},
\be
H(\tau)=2\,q^{-1/8}\left(-1+45\,q+231\,q^2+770\,q^3+\CO(q^4) \right),
\ee 
and transforms under $\slz$.

Including either surface observables or the coupling to the background fluxes requires an elliptic generalisation, which has to transform under $\slz$ in order to be applicable to $u$-plane integrals with small masses. In Appendix \ref{app:mockjacobi}, we construct such an $\slz$ mock Jacobi form $H(\tau,\rho)$, and discuss the relation to $F_{\frac12}(\tau,\rho)$. 
 Deriving a similar expression related to $F_0$ which transforms under $\slz$ is more involved, since $F_0(\tau,\rho)$ has a pole at $\rho=0$ (due to  the term $n=0$ in the sum \eqref{f12function}). We leave it for future work to find such an elliptic generalisation of $F_0$. In Appendix \ref{app:mock}, we study further properties of  the above mock modular forms in great detail, while their elliptic generalisations including  zeros and poles are discussed in Appendix \ref{app:mockjacobi}.

\subsection{Constraints for contributions from the cusps}
\label{sec:ContCusps}

In this subsection, we consider the $u$-plane integrals with vanishing external fluxes, $\bfk_j=0$.  We consider the leading behaviour of the integrand near cusps, and determine selection rules for the cusps to have potential nonzero contributions. 

\subsubsection*{Point observables}
Let us first assume that the intersection form of $X$ is odd. If we turn off the surface observable $\bfx$, we can evaluate the $u$-plane integral \eqref{integrationresult} for generic masses. As found above, if  $J=(1,\boldsymbol 0_{b_2^-})$  then $\Psi_{\frac K2}^J(\tau,0)$ vanishes whenever $b_2^->0$. This gives the result 
\begin{equation}\label{pointobservableb2m>0}
\Phi_{\frac K2}^{J,N_f}[e^{2pu}]=0, \qquad b_2^->0.
\end{equation}
Let us therefore proceed with $b_2^-=0$, such that $\sigma=1$ and $\chi=3$. We calculate $u$-plane integrals for such manifolds in great detail in Section \ref{sec:p2}.

In this case, the Siegel-Narain theta function \eqref{psi120odd} becomes 
\begin{equation}\label{PsiK2P2}
\Psi_{\frac K2}^J(\tau,\bar\tau)=-\im^{K_1}\frac{-\im}{2\sqrt{2y}}\overline{\eta(\tau)^3}.
\end{equation}
We can use the fact that $\frac{-\im}{2\sqrt{2y}}\overline{\eta(\tau)}^3=\partial_{\bar\tau}\widehat F_{\frac12}(\tau,\bar\tau)$ is the shadow of the mock modular form $F\coloneqq F_{\frac12}$, defined in \eqref{f12function}.

As discussed in Section \ref{sec:integrationFD}, the $u$-plane integral can then be expressed as a sum over $q^0$-coefficients of the integrand evaluated near the cusps, labelled by $j$,
\begin{equation}\label{pointobservableodd}
\Phi_{\frac K2}^{J,N_f}[e^{2 pu}]=-\frac{\im^{K_1}}{24}\sum_{j} w_j {\rm Coeff}_{q_j^0}\left[ \nu(\alpha_j\tau)H(\alpha_j\tau)\,e^{2pu(\alpha_j\tau)}\right].
\end{equation}
Here, $\alpha_j\in \slz$ give the cosets in the fundamental domain \eqref{fundamental_domain}, and   $w_j$ is the width of the cusp $j$.

Let us study which cusps contribute to the sum \eqref{pointobservableodd}. We have that $H(\tau)=\CO(q^{-\frac18})$ for $\tau\to\im\infty$. One furthermore finds
\begin{equation}\begin{aligned}
u(\tau)&=\CO(q^{-\frac{1}{4-N_f}}), \quad \frac{du}{da}(\tau)=\CO(q^{-\frac{1}{2(4-N_f)}}), \\
\frac{da}{d\tau}&=\CO(q^{-\frac{1}{2(4-N_f)}}), \quad \Delta(\tau)=\CO(q^{-\frac{2+N_f}{4-N_f}}).
\end{aligned}\end{equation}
in the weak coupling frame. 
Then the measure factor goes as 
\begin{equation}
\nu(\tau)=\CO(q^{-\frac{12+\sigma N_f }{8(4-N_f)}}),
\end{equation}
which holds for generic masses and $b_2^+=1$. For $\sigma=1$ the exponent is $\frac18-\frac{2}{4-N_f}\leq -\frac38$, where equality holds for $N_f=0$. Combining with the exponent $-\frac18$ of $H$, the exponent of the leading term in the $q$-expansion of $\nu H u^\ell$ \eqref{pointobservableodd} is
\begin{equation}
    \nu(\tau)\, H(\tau)\, u(\tau)^\ell =\CO(q^{-\frac{2+\ell}{4-N_f}}).
\end{equation}
Since both $4-N_f>0$ and $2+\ell>0$, this exponent is strictly negative.  We confirm through explicit calculations in Section \ref{sec:p2} that indeed for generic masses also the $q^0$ term is present. 
This shows that the cusp $\im\infty$ generally contributes to $\Phi_{\frac K2}^{J,N_f}[e^{2 pu}]$ to all orders in $p$, for all $N_f\leq 3$.

This is not true for the strong coupling cusps, $j\in\mathbb Q$. These cusps are in fact simpler to analyse, since the measure $\nu$ at strong coupling becomes a constant. In order to see this, recall that $u_D(\tau)=\CO(1)$ and   $\left(\frac{du}{da}\right)_D=\CO(1)$ (see Section \ref{sec:behaviour_strong_c_s} for more details). We also have $\Delta_D(\tau)=\CO(q)$, such that we are left with studying $\frac{da}{d\tau}=\frac{da}{du}\frac{du}{d\tau}$. Near a singularity $u_j^*$, the local coordinate reads $u_D(\tau)-u_j^*=\CO(q^{\frac{1}{n}})$, where $n$ is the width of the cusp $j$ (corresponding to an $I_n$ singularity). For asymptotically free SQCD, the possibilities are $n=1,2,3,4$. 
Therefore, we have that $\left(\frac{du}{d\tau}\right)_D=\CO(q^{\frac{1}{n}})$ and thus
\begin{equation}
\nu_D(\tau)=\CO(q^{\frac{1}{n}+\frac \sigma8}).
\end{equation}
Since $H$ is mock modular for $\slz$, also $H_D(\tau)=\CO(q^{-\frac 18})$. Thus we find that the lowest $q$-exponent of the contribution to an $I_n$ cusp is
\begin{equation}\label{strong_coupling_exponent}
    \nu_D(\tau)\, H_D(\tau)\, u_D(\tau)^\ell =\CO(q^{\frac{1}{n}+\frac\sigma 8-\frac 18}).
\end{equation}
For our choice of period point $J=(1,\boldsymbol 0_{b_2^-})$, we can set $\sigma=1$. Then the leading exponent is $\frac 1n>0$, such that the $q^0$ coefficient vanishes.  Thus the for manifolds with $\sigma=1$ the strong coupling cusps never contribute to correlation functions $\Phi_{\frac K2}^{J,N_f}[e^{2 pu}]$ of the point observable.

The correlation functions for the point observable on manifolds with odd intersection form then  receives contributions only from weak coupling. Since the width of the cusp at infinity is $w_{\im\infty}=4-N_f$, we can simplify \eqref{pointobservableodd} substantially,
\begin{equation}\label{pointobservableoddResult}
\Phi_{\frac K2}^{J,N_f}[e^{2 pu}]=-\frac{\im^{K_1}(4-N_f)}{24}\,{\rm Coeff}_{q^0}\left[ \nu(\tau)\,H(\tau)\,e^{2pu(\tau)}\right].
\end{equation}
In \cite{Korpas:2019cwg} it is observed that for $N_f=0$, correlation functions of point observables are (up to an overall dependence on the canonical class) universal for any four-manifold with odd intersection form and given period point $J$. The reason for this is that the topological dependence of the measure factor $\nu \sim\jt_4^{-b_2}$ cancels precisely with the holomorphic part  of the Siegel-Narain theta function $\Psi_\bfmu^J\supset \jt_4^{b_2}$. This is not true for $N_f>0$, which one may also  see by comparing \eqref{pointobservableb2m>0} with \eqref{pointobservableoddResult}.

\subsubsection*{Surface observables}
 We can also consider correlation functions of surface observables $\bfx\in H_2(M)$ supported on the compact four-manifold $X$. Following  Section \ref{factorisation_psi}, $\Psi_{\frac K2}^J$ for the choice $J=(1,\boldsymbol 0_{b_2^-})$ factorises as
 \begin{equation}\label{Psi_surface}
 \Psi_{\frac K2}^J(\tau,\bfrho)=-\im^{K_1}f_{\frac12}(\tau,\rho_1)\Theta_{L_-,\bfmu_-}(\tau,\bfrho_-),
 \end{equation}
 with $\bfmu=(\mu_+,\bfmu_-)\equiv(\frac12,\frac12,\dots,\frac12) \mod \mathbb Z^{b_2}$. Due to \eqref{ThetaL-}, we have that 
 \begin{equation}\label{Psi_surface_L-}
 \Theta_{L_-,\bfmu_-}(\tau,\bfrho_-)=\prod_{k=2}^{b_2}-\im^{K_k+1}\jt_1(\tau,\rho_k),
 \end{equation}
 where $\bfrho_-=(\rho_2,\cdots,\rho_{b_2})$. The function $f_{\frac12}(\tau,\rho)$ is the shadow of the mock modular form $\tfrac{1}{24}H(\tau,\rho)$, as in \eqref{H12Def}. This allows to evaluate \eqref{integrationresult}, where we also include the point observable,
\begin{equation}\label{PhiODD}
\begin{split}
& \Phi_{\frac K2}^{J,N_f}[e^{I_-(\bfx)}]=\frac{-\im^{K_1}}{24}\!\!\sum_{j} w_j \\
 &\qquad \times {\rm Coeff}_{q_j^0}\left[\nu(\tau_j)e^{2pu(\tau_j)/\Lambda_{N_f}^2+\bfx^2 G(\tau_j)}H(\tau,\rho_{1,j})\prod_{k=2}^{b_2}(-\im^{K_k+1})\jt_1(\tau_j,\rho_{k,j})\right],
 \end{split}
 \end{equation}
where we calculate the local $q_j$ series around the cusps $j$ and extract the constant term.
 Let us check that the case $\bfx=0$ is consistent with the previous result. Consider thus that $\bfrho=0$ in above formula. If $b_2^->0$ and consequently $b_2\geq2$, then all factors in the product vanish, since $\jt_1(\tau,0)\equiv0$. This reproduces \eqref{pointobservableb2m>0}. If $b_2^-=0$ on the other hand, then the product is over an empty set and therefore equal to $1$. By construction $H(\tau,0)=H(\tau)$, and the limit to \eqref{pointobservableodd} is obvious.

\subsubsection*{When do strong coupling cusps contribute?}

In \eqref{strong_coupling_exponent} it was found that the contribution of strong coupling cusps to the $u$-plane integral depends on an intricate way on the four-manifold $X$ and on the type of cusp. For instance, let $u_j^*$ be an $I_n$ singularity, such that the local expansion reads $u_D(\tau)=u_j^*+\CO(q^{\frac1n})$. Then the smallest exponent in the $q$-series of the measure factor $\nu_D$ whose coefficient is strictly non-zero is $\tfrac1n+\tfrac\sigma 8$, independent of the mass configuration $\bfm$ giving rise to that $I_n$ singularity.
Consider now  an SQCD mass configuration containing singularities of type $I_n$ and $I_m$ which can be merged by colliding  some masses. If the signature $\sigma$ is such that the smallest exponent of the $q$-series of the integrand is positive, then their  individual contributions  vanish. However, if the $I_n$ and $I_m$ singularities merge to an $I_{n+m}$ singularity, the lowest exponent can become non-positive and there can be a contribution to the $u$-plane integral. The simplest example would be two $I_1$ singularities colliding to an $I_2$ singularity. 

For the complex projective plane $X=\pt$ this does not occur, since for any $I_n$ singularity and any $N_f$ the smallest exponent is strictly positive, $\frac 1n$. This is in agreement with the theorem that for $N_f=0$ and for four-manifolds with $b_2^+(X)>0$ that admit a Riemannian metric of positive scalar curvature the Seiberg-Witten invariants vanish \cite{Witten:1994cg,Taubes1994}.\footnote{This is a consequence of the Bochner-Lichnerowicz-Weitzenb{\"o}ck formula, which relates the Dirac operator to the scalar curvature via the connection Laplacian. If a given metric has positive curvature, the kernel of the Dirac operator is empty. See \cite{yanez2023} for an overview.} The theorem has been shown to generalise also to $N_f>0$ \cite{bryan1996, Dedushenko:2017tdw}. See \cite{MANTIONE2021566} for a survey on four-manifolds with positive scalar curvature.

To test whether this vanishing theorem also holds for the multi-monopole SW equations, we can calculate $u$-plane integrals for manifolds $X$ of small signature  that admit metrics with positive scalar curvature. 
Such a class of four-manifolds are the del Pezzo surfaces $dP_n$. They are blow-ups of the complex projective plane at $n$ points, where $n=1,\dots, 8$. For $n=9$, it is known as $\frac12\text{K3}$. These surfaces have $b_2^+(dP_n)=1$ and signature $\sigma(dP_n)=1-n$. The canonical class of $dP_n$ is $K=-3H+E_1+\dots+E_n$, with $E_i$ the exceptional divisors of the blow-up, and $H$ the pullback of the hyperplane class from $\pt$. The intersection form can be brought to the form  
\begin{equation}
Q_{dP_n}=\begin{pmatrix}1&0\\0&-\mathbbm 1_n
\end{pmatrix},
\end{equation}
with $\mathbbm 1_n$ the $n\times n$ identity matrix. From this it follows that $K_{dP_n}^2=9-n$, which is the degree of $dP_n$.
As explained in Section \ref{factorisation_psi}, for manifolds with odd intersection form, the components $K_j$ of the characteristic vector $K$ are odd for all $j=1, \dots, b_2$. Without external fluxes $\bfk_j=0$, the $u$-plane integrals are well-defined if $\bfmu=\frac 12 K\mod L$. On the other hand, the Siegel-Narain theta function for $J=(1,0,\dots,0)$ vanishes identically unless $\bfmu=(\tfrac12,0,\dots, 0) \mod \mathbb Z^{b_2}$. This shows that without surface observables and without external fluxes, the $u$-plane integrals necessarily vanish for the del Pezzo surfaces $dP_n$ with $n\geq 1$.

If we include surface observables, the $\theta_1(\tau,\rho_k)$ in Eq. (\ref{PhiODD}) transform to $\theta_1(\tau_j,\rho_{k,j})$ with the leading term in the $q_j$ expansion of $\rho_{k,j}$ a non-vanishing constant (where we use $\left(\frac{du}{da}\right)_D=\CO(1)$ and the $S$-transformation \eqref{jt_S-transformation}). The leading term of $\jt_1$ is $\theta_1(\tau_j,\rho_{k,j})\sim q_j^{1/8}$, such that the product over $b_2-1$ of these gives $q_j^{-(\sigma-1)/8}$. As a result, the $\sigma$ dependence of the measure is cancelled by $\Psi_\bfmu^J$, and the local asymptotics is $\CO(q^{\frac1n})$ for any $\sigma$. We conclude that the strong coupling cusps do not contribute after inclusion of surface observables. It would be interesting to explore if non-vanishing background fluxes affect this conclusion.

\subsection{Wall-crossing}\label{sec:wall-crossing}
An intrinsic feature of $u$-plane integrals for $b_2^+=1$ is the metric dependence and the wall-crossing associated with it.  
The metric dependence of the Lagrangian is encoded in the period point $J\in H^2(X,\mathbb R)$, which generates the space $H^2(X,\mathbb R)^+$ of self-dual two-cohomology classes and is normalised as $Q(J)=1$. It depends on the metric through the self-duality condition $*J=J$. Using a period point $J$, we can project some vector $\bfk\in L$ to the positive and negative subspaces $H^2(X,\mathbb R)^\pm$ using $\bfk_+=B(\bfk,J)J$ and $\bfk_-=\bfk-\bfk_+$. 

Even when including the background fluxes, the dependence of the $u$-plane integrand on the metric is only through the Siegel-Narain theta function $\Psi_\bfmu^J$. The metric dependence is then captured through the difference $\Phi_{\bfmu}^J-\Phi_{\bfmu}^{J'}$ for  two period points $J$ and $J'$, which we aim to evaluate. To this end, we note that the difference 
\begin{equation}
    \Psi_\bfmu^J(\tau,\bfz) -\Psi_\bfmu^{J'}(\tau,\bfz)=\partial_{\bar\tau}\widehat\Theta_\bfmu^{J,J'}(\tau,\bfz)
\end{equation}
is a total derivative to $\bar\tau$, with 
\be
\label{hatTheta}
\begin{split}
\widehat \Theta^{JJ'}_{\bfmu}\!(\tau,\bar \tau,\bfz,\bar \bfz)&=\sum_{\bfk\in L+\bfmu}
\tfrac{1}{2}\left( E(\sqrt{2y}\,B(\bfk+\bfb,
  J))-\sgn(B(\bfk, J'))\right)\\
&\quad \times e^{\pi i B(\bfk,K)} q^{-\bfk^2/2}e^{-2\pi
  i B(\bfk,\bfz)}
\end{split}
\ee
and 
\begin{equation}
\label{Eerrorfunction}
E(u) = 2\int_0^u e^{-\pi t^2}dt = \text{Erf}(\sqrt{\pi}u)
\end{equation}
a rescaled error function $E:\mathbb{R}\to (-1,1)$. We have under the $S$- and $T$-transformations,
\be
\begin{split} 
&\widehat \Theta^{JJ'}_{\bfmu}(-1/\tau,-1/\bar \tau,\bfz/\tau,\bar\bfz/\bar \tau)= i (-i\tau)^{b_2/2}\exp(-\pi i \bfz^2/\tau)\\
&\qquad \qquad \times e^{\pi i B(\bfmu,K)}\,\widehat \Theta^{JJ'}_{K/2}(\tau,\bar \tau,\bfz-\bfmu+K/2,\bar\bfz-\bfmu+K/2),\\
& \widehat \Theta^{JJ'}_{\bfmu}(\tau+1,\bar \tau+1,\bfz,\bar \bfz)=e^{\pi i (\bfmu^2-B(K,\bfmu))}\widehat \Theta^{JJ'}_{\bfmu}(\tau,\bar \tau,\bfz+\bfmu-K/2,\bar \bfz+\bfmu-K/2).
\end{split}
\ee

Since the couplings to the background fluxes are holomorphic, the inclusion of the latter is not affected by a total $\bar\tau$ derivative. This allows to express $\Phi_{\bfmu}^J-\Phi_{\bfmu}^{J'}$ as an integral of the form $\CI_f$ (defined in \eqref{CI_f}) for some function $f$ satisfying $\partial_{\bar\tau}\widehat h=y^{-s}f$, where we can read off 
\begin{equation}
    \widehat h(\tau,\bar \tau)= \nu(\tau;\{\bfk_j\})\,\widehat\Theta_\bfmu^{JJ'}(\tau,\bar\tau,\bfz,\bar\bfz)\,e^{2pu/\Lambda_{N_f}^2+\bfx^2 G_{N_f}}.
\end{equation}
Then, according to \eqref{CIfStokes}, we can write 
\begin{equation}
    \Phi_{\bfmu}^J-\Phi_{\bfmu}^{J'}=-\int_{\partial \CF(\bfm)} d\tau\,\widehat h(\tau,\bar \tau),
\end{equation}
which  may be evaluated using the methods in Section \ref{sec:integrationFD}. In particular, it can be evaluated using the indefinite theta function $\Theta_\bfmu^{JJ'}$, which is the holomorphic part of $\widehat\Theta_\bfmu^{JJ'}$. The contribution from the singularity at infinity is 
\be 
\label{WCPhi}
\begin{split}
\left[ \Phi_{\bfmu}^{J}-\Phi_\bfmu^{J'} \right]_\infty&= \CK_{N_f} \int_{-1/2+iY}^{1/2+iY} d\tau\,\frac{da}{d\tau} A^\chi B^\sigma \prod_{j,k=1}^{N_f} C_{ij}^{B(\bfk_i,\bfk_j)}\\
&\quad \times \widehat \Theta_\bfmu^{JJ'}(\tau,\bar \tau, \sum_j \bfk_j v_j, \sum_{j} \bfk_j \bar v_j),
\end{split}
\ee 
where we left out the observables. The contributions from the strong coupling singularities follow from modular transformations and will be discussed in Section \ref{sec:SWcontributions}.

\section{Example: four-manifolds with \texorpdfstring{$b_2=1$}{b2=1}}\label{sec:p2}
Let us study in detail the $u$-plane integrals \eqref{pointobservableoddResult} for the point observable on four-manifolds with $b_2=1$. 

The complex projective plane $\mathbb P^2$ is the most well-known example. The complex projective plane has   $\sigma=1$, $\chi=3$ and thus $b_2=1$, furthermore  $K=K_1=3$. With the exact results from \cite{Aspman:2022sfj,Korpas:2019cwg}, condensed in Section \ref{sec:uplane_int_mock}, it is straightforward to evaluate \eqref{pointobservableoddResult} for arbitrary masses.

In this Section, we compute the vev $\Phi_{\frac12}^{N_f}[u^\ell]$ for the point observable, which in the notation of  \cite{Korpas:2019cwg} is related to the exponentiated observable as
\begin{equation}
\Phi_{\bfmu, \{\bfk_j\}}^{N_f}[e^{2p u/\Lambda_{N_f}^2}]\coloneqq \left\langle e^{2p u/\Lambda_{N_f}^2}\right\rangle_{\bfmu, \{\bfk_j\}}^{N_f}=\sum_{\ell=0}^\infty \frac{1}{\ell!}\left(\frac{2p}{\Lambda_{N_f}^2} \right)^\ell \Phi_{\bfmu, \{\bfk_j\}}^{N_f}[u^\ell].
\end{equation}

If the background fluxes $\bfk_j=0$ on $X=\mathbb P^2$ are vanishing, in part I we argued that the theory is consistent only if we restrict to $\bar w_2(X)\equiv \bar w_2(E) \mod 2L$. For $N_f=1$ for instance, we can choose $\mu=1/2$. For this flux, we can also turn on any integral background flux $k_1\in\mathbb Z$. If we turn off the 't Hooft flux $\mu=0$ rather, the consistent formulation on $\mathbb P^2$  requires half-integer background fluxes $k_1\in \mathbb Z+1/2$. 

In the following two subsections, we consider the large mass and small mass calculations for $\Phi_{1/2}^{N_f}$ with $k_j=0$, while in subsection \ref{sec:nonvanishing_bckg_flux} we turn on non-vanishing background fluxes $k_j$ for both $\mu=0$ and $\mu=1/2$.

\subsection{Large mass expansion with vanishing background fluxes}\label{sec:large_mass_P2}

We first consider the large mass expansion for equal masses $m_i\eqqcolon m$ for $N_f=1,2,3$, in the absence of background fluxes.
This allows to normalise the integral, by requiring that the decoupling limit $m\to \infty$ for $\Phi_{1/2}[1]$ reproduces the $N_f=0$ result. We will demonstrate that with this normalisation, the decoupling limit of other observables also matches with $N_f=0$ as expected. As shown in Section \ref{sec:ContCusps}, there are no contributions from the strong coupling cusps for $N_f\leq 3$ and $X=\mathbb{P}^2$.  
Since the holomorphic part of the integrand is a function of $u$,  large $m$ expansion of the latter can be determined as described in Section \ref{weakcoupexp}. 

From \eqref{decouplimit}, we find for the  decoupling formula for equal masses
\begin{equation}\label{decoupling_equal}
    \Lambda_0^4=m^{N_f}\Lambda_{N_f}^{4-N_f}.
\end{equation}
In the large $m$ limit, the domain is a truncated $\mathbb{H}/\Gamma^0(4)$ domain for all $N_f$, as discussed in Section \ref{sec:funddom}. Combining the measure factor \eqref{nutauNf} applied to $X=\mathbb{P}^2$ with \eqref{pointobservableoddResult} and \eqref{decoupling_equal}, we find in the notation of \eqref{PhiLargeM},
\be 
\begin{split}
&\Phi_{1/2}\big[e^{2pu/\Lambda_{N_f}^2}\big]=(-1)^{N_f}\frac{2}{4-N_f}\frac{m^{3N_f}}{\Lambda_{N_f}^{N_f}} \frac{1}{\Lambda_0^{12}}\\
&\qquad \times {\rm Coeff}_{q^0} {\rm Ser}_{m^{-1}} \left[ \left(\frac{du}{da}\right)^{12}\frac{\eta(\tau)^{27}}{P_{N_f}^{\text{M}}}\,F_{\frac12}(\tau) \, e^{2pu/\Lambda_{N_f}^2}\right].
\end{split}
\ee 
Here,  we have  used the holomorphic part $F_{\frac12}$ of the anti-derivative $\widehat F_{\frac12}$. It is straightforward to check that other choices of anti-derivative, such as $\frac{1}{24}\widehat H$ give the same result.

We first present the series in a form which makes the decoupling limit manifest. To this end, we list the coefficients of $(p/\Lambda_1^2)^\ell$ as function of $\Lambda_0$ and $m$, up to the overall prefactor $(m/\Lambda_{N_f})^{N_f}$. For $N_f=0$, we have
\be
\label{Nf0Eu}
\begin{split} 
\Phi_{1/2}\big[e^{2pu/\Lambda_0^2}\big]&=1 +\frac{19}{32}\Lambda_0^4 \frac{p^2}{\Lambda_0^4} +\frac{85}{768} \Lambda_0^8 \frac{p^4}{\Lambda_0^8}+\CO(p^5).
\end{split}
\ee 
For $N_f=1$, we then find
\be
\label{Nf1Eu}
\begin{split} 
\Phi_{1/2}\big[e^{2pu/\Lambda_1^2}\big]&=\frac{m}{\Lambda_1} \left(1-\frac{7}{32} \frac{\Lambda_0^{4}}{m^2}\,\frac{p}{\Lambda_1^2} +\frac{19}{32}\Lambda_0^4 \frac{p^2}{\Lambda_1^4} \right. \\
&\left.-\frac{7}{64} \frac{\Lambda_0^8}{m^2} \frac{p^3}{\Lambda_1^6}+\left(\frac{85}{768} \Lambda_0^8+\frac{1093}{393216}\frac{\Lambda_0^{12}}{m^4}\right) \frac{p^4}{\Lambda_1^8}+\CO(p^5,m^{-8})\right).
\end{split}
\ee 
For $N_f=2$, we find 
\be 
\label{Nf2Eu}
\begin{split} 
\Phi_{1/2}\big[e^{2pu/\Lambda_2^2}\big]&=\frac{m^2}{\Lambda_2^2}\left(1+\frac{3}{2^6} \frac{\Lambda_0^4}{m^4}+\left(-\frac{7}{2^4}\frac{\Lambda_0^4}{m^2} \right)\frac{p}{\Lambda_2^2}\right.\\
&+ \left(\frac{19}{32} \Lambda_0^4+\frac{23}{256}\frac{\Lambda_0^8}{m^4}+\frac{53}{2^{17}}\frac{\Lambda_0^{12}}{m^8}\right)\frac{p^2}{\Lambda_2^4} \\
&-\left(\frac{7}{32} \frac{\Lambda_0^8}{m^2} + \frac{421}{49152} \frac{\Lambda_0^{12}}{m^6}\right)\frac{p^3}{\Lambda_2^6}\\
& +\left( \frac{85}{768}\Lambda_0^8+\frac{2421}{65536}\frac{\Lambda_0^{12}}{m^4}+\frac{2161}{3145728}\frac{\Lambda_0^{16}}{m^8}\right)\frac{p^4}{\Lambda_2^8}\\
&\left.+\CO(p^5,m^{-8})\right).
\end{split}
\ee 
Finally, for  $N_f=3$ we find 
\be 
\label{Nf3Eu}
\begin{split} 
\Phi_{1/2}\big[e^{2pu/\Lambda_3^2}\big]&=\frac{m^3}{\Lambda_3^3}\left(1+\frac{9}{2^6} \frac{\Lambda_0^4}{m^4}+\frac{5}{2^{18}}\frac{\Lambda_0^{12}}{m^{12}}\right. \\
&-\left(\frac{21}{2^5}\frac{\Lambda_0^4}{m^2} +\frac{3}{2^{6}}\frac{\Lambda_0^8}{m^6}\right)\frac{p}{\Lambda_3^2}\\
&+ \left(\frac{19}{2^5} \Lambda_0^4+\frac{69}{2^8}\frac{\Lambda_0^8}{m^4}+\frac{1659}{2^{17}}\frac{\Lambda_0^{12}}{m^8}\right)\frac{p^2}{\Lambda_3^4} \\
&-\left(\frac{21}{2^6} \frac{\Lambda_0^8}{m^2}+\frac{3305}{3\times 2^{14}} \frac{\Lambda_0^{12}}{m^6} \right) \frac{p^3}{\Lambda_3^6}\\
& +\left(\frac{85}{768}\Lambda_0^{8}+\frac{13433}{2^{17}}\frac{\Lambda_0^{12}}{m^4}+\frac{13397}{2^{20}}\frac{\Lambda_0^{16}}{m^8}\right) \frac{p^4}{\Lambda_3^8}\\
&\left.+\CO(p^5,m^{-{9}})\right).
\end{split}
\ee 
For the $m\to \infty$ decoupling limit \eqref{IR-decoupling}, we multiply by the factor $(\Lambda_0/\Lambda_{N_f})^\alpha$ with $\alpha=-4$ for $\mathbb{P}^2$. Eliminating $\Lambda_0$ by (\ref{dscalinglimit}), this removes precisely the prefactors in the expressions \eqref{Nf1Eu}, \eqref{Nf2Eu} and \eqref{Nf3Eu}. We thus find a consistent  decoupling limit \eqref{IR-decoupling} to the $N_f=0$ result (\ref{Nf0Eu}) for all three cases. 
\begin{center}
\begin{table}[ht!]
\renewcommand{\arraystretch}{2}
\centering
\begin{tabular}{|Sl|Sr|}
\hline 
$\ell$ & $\displaystyle\Phi_{1/2}[u^\ell]$ for $N_f=1$\\ 
\hline 
0 & $\displaystyle\frac{m}{\Lambda_1}$\\ 
1 & $\displaystyle -\frac{7}{2^6}\frac{\Lambda_1^2m^2}{m^2}$ \\
2 & $\displaystyle \frac{19}{2^6}\frac{\Lambda_1^2m^2}{m^0}$\\
3 & $\displaystyle  -\frac{21}{2^8}\frac{\Lambda_1^5m^3}{m^2}$ \\
4 & $\displaystyle \frac{85}{2^9}\frac{\Lambda_1^5m^3}{m^0} +\frac{1093}{2^{18}}\frac{\Lambda_1^8m^4}{m^4}$ \\
\hline 
\end{tabular}
\caption{Table with values of $\Phi_{1/2}[u^\ell]$ for $N_f=1$. }
\label{PtObsNf1}
\end{table} 
\end{center}
\begin{center}
\begin{table}[ht!]
\centering
\renewcommand{\arraystretch}{2}
\begin{tabular}{|Sl|Sr|}
\hline 
$\ell$ & $\displaystyle\Phi_{1/2}[u^\ell]$ for $N_f=2$\\ 
\hline 
0 & $\displaystyle\frac{m^2}{\Lambda_2^2}+\frac{3}{64}\frac{m^4}{m^4}$\\ 
1 & $\displaystyle -\frac{7}{2^5}\frac{m^4}{m^2}$ \\
2 & $\displaystyle \frac{19}{2^6}m^4+\frac{23}{2^7}\frac{\Lambda_2^2 m^6}{m^4}+\frac{53}{2^{18}}\frac{\Lambda_2^4m^8}{m^8}$\\
3 & $\displaystyle  -\frac{21}{2^7} \Lambda_2^2\frac{m^6}{m^2}-\frac{421}{2^{16}}\frac{\Lambda_2^4m^8}{m^6}$ \\
4 & $\displaystyle \frac{85}{2^8}\Lambda_2^2m^6 +\frac{7263}{2^{17}}\frac{\Lambda_2^4m^8}{m^4}+\frac{2161}{2^{21}}\frac{\Lambda_2^{6}m^{10}}{m^8}+\frac{1811}{2^{30}}\frac{\Lambda_2^{8}m^{12}}{m^{12}}$ \\
\hline 
\end{tabular}
\caption{Table with values of $\Phi_{1/2}[u^\ell]$ for $N_f=2$. }
\label{PtObsNf2}
\end{table} 
\end{center}
\begin{center}
\begin{table}[ht!]
\centering
\renewcommand{\arraystretch}{2}
\begin{tabular}{|Sl|Sr|}
\hline 
$\ell$ & $\displaystyle\Phi_{1/2}[u^\ell]$ for $N_f=3$\\ 
\hline 
0 & $\displaystyle\frac{m^3}{\Lambda_3^3\,m^0}+\frac{9}{2^6}\frac{m^6}{\Lambda_3^2\,m^4}+\frac{5}{2^{18}}\frac{m^{12}}{m^{12}}$\\ 
1 & $\displaystyle -\frac{21}{2^{6}}\frac{m^6}{\Lambda_3^{2}\,m^2} -\frac{3}{2^7}\frac{m^9}{\Lambda_3\,m^6}-\frac{15}{2^{17}}\frac{m^{12}}{m^{10}}+\CO(m^0)$ \\
2 & $\displaystyle \frac{19}{2^6}\frac{m^6}{\Lambda_3^{2}}+\frac{69}{2^9}\frac{m^9}{\Lambda_3\,m^4}+\frac{1659}{2^{18}}m^4+\frac{55}{2^{21}}\frac{\Lambda_3m^{15}}{m^{12}}+\CO(m^2)$\\
3 & $\displaystyle  -\frac{63}{2^{8}} \frac{m^9}{\Lambda_3\,m^2}-\frac{3305}{2^{16}} \frac{m^{12}}{m^6}-\frac{477}{2^{18}}\frac{\Lambda_3 m^{15}}{m^{10}}+\CO(m^4)$ \\
4 & $\displaystyle \frac{85}{2^9} \frac{m^9}{\Lambda_3} +\frac{40299}{2^{18}}\frac{m^{12}}{m^4}+\frac{40191}{2^{21}}\frac{\Lambda_3 m^{15}}{m^8}+\frac{150471}{2^{28}}\frac{\Lambda_3^{2}m^{18}}{m^{12}}+\CO(m^5)$ \\
\hline 
\end{tabular}
\caption{Table with values of $\Phi_{1/2}[u^\ell]$ for $N_f=3$ and equal masses $m_j=m$. We leave the $\CO(\dots)$ terms undetermined due to large running time of the Mathematica notebook. As explained in the main text, we expect these terms to match with the expressions in Table \ref{PtObsNf3G} determined using the small mass expansion.}
\label{PtObsNf3}
\end{table} 
\end{center}
To facilitate comparison of these results with the UV expression \eqref{UVpartfunc}, we have presented the data in an alternative form in the Tables \ref{PtObsNf1}, \ref{PtObsNf2} and \ref{PtObsNf3}. Here, correlators $\Phi_{1/2}[u^\ell]$ are listed as functions of $m$ and $\Lambda_{N_f}$. The monomial in $m$ and $\Lambda_{N_f}$ is expressed as
\be\label{UV_scale_ratio}
\frac{\Lambda_{N_f}^{{\rm vdim}(\CM^{Q,N_f}_{k})}m^{-\sum_j {\rm rk}(W^j_k)}}{m^r}.
\ee
The exponent of the mass $m$ is the virtual rank \eqref{rkW_k^j} of the matter bundle, that is, the rank of the obstruction bundle $W_k^j$. The exponent of the scale $\Lambda_{N_f}$ can be identified with the (complex) virtual dimension \eqref{vdimML} of the moduli space (see also \eqref{UVpartfunc}). The exponent $r$ is the degree of the Chern class of the matter bundle. For $N_f=1$, we have
\be 
\begin{split}
&{\rm dim}(\CM^Q_k)=3k-4,\\
&{\rm dim}(\CM^{\rm i}_k)=4k-3,\\
&{\rm rk}(W_k)=-k-1.\\
\end{split} 
\ee 
With $s$ the exponent of $p$, the data in the Tables \ref{PtObsNf1}, \ref{PtObsNf2} and \ref{PtObsNf3} satisfy the selection rule (\ref{selrule}). Moreover, since the integration is over the instanton moduli space, we have the selection rule $4p+2r={\rm dim}_{\mathbb{R}}(\CM_k^{\rm i})$.  
If the obstruction bundle is a proper bundle rather than a sheaf, we have $r\leq -\sum_j {\rm rk}(W^j_k)$. We find that this is the case for all data in these tables. Thus even though the evaluation of the $u$-plane integral was performed in terms of a $1/m$ expansion, the results in the tables  have a good $m\to 0$ limit. We will discuss this in more detail in the following Section \ref{smallmassP2} on the small mass expansion. It is furthermore noteworthy that for fixed $\ell$ the coefficients have the same sign. The results of the next few sections show only a few exceptions to this.

From the large powers of 2 in the denominators, we deduce that the normalisation is not precisely consistent with integral Chern classes in the large $m$ expansion. We discuss this in more detail in the next subsection, where we give results for generic masses $m_j$. Mathematically, these invariants are known as (virtual) Segre numbers of $X$ \cite{Gottsche:2020ass}. We comment more on this connection in the section \ref{sec:SWcontributions} on four-manifolds with $b_2^+>1$.

\subsection{Small mass expansion with vanishing background fluxes}
\label{smallmassP2}
For small masses, the integration domains are now naturally the small mass perturbations of the domains for the massless theories. See e.g. Fig. \ref{fig:fundnf1} and \ref{fig:nf2massless} for the massless $N_f=1$ and $N_f=2$ domains. 
The regularised integration domains suitable for the integration are described in Section \ref{sec:funddom}, and the weak coupling cusps have width $4-N_f$. For the anti-derivative of $\Psi_{1/2}$, we take the $\slz$ mock modular  form $H(\tau)$ \eqref{HtauDef}, which transforms consistently on any of these domains.

With the normalisation determined above, we have for the small mass result
\be 
\Phi_{1/2}\big[e^{2pu/\Lambda_{N_f}^2}\big]=\frac{(-1)^{N_f}}{2\,\Lambda_{N_f}^{12-2N_f}} {\rm Coeff}_{q^0} {\rm Ser}_m\left[ \left(\frac{du}{da}\right)^{12} \frac{\eta(\tau)^{27}}{P_{N_f}^{\text{M}}}\frac{H(\tau) }{24} e^{2pu/\Lambda_{N_f}^2}\right].
\ee 

For the massless case, Table  \ref{tableulP2} gives the first 8 nonzero intersection numbers for massless $N_f=0,1,2,3$. For $N_f=0$, the results match precisely with \cite{ellingsrud1995wall}. The results for $N_f=2,3$ are in agreement with the results for $\Phi^{N_f,H,1}_{k,m,n}$ in Ref. \cite{Malmendier:2008db}.\footnote{The characteristic classes of the matter bundle $\omega_{j,\text{MO}}$  in \cite{Malmendier:2008db} are normalised such that $2\omega_{j,\text{MO}}$ is an integral class. Moreover, the Tables in \cite{Malmendier:2008db} are for point class insertions which correspond to $2u$ in our notation.} We observe that the intersection numbers grow quickly as function of $\ell$. It would be interesting to study the asymptotic behaviour of these series similar to the case $N_f=0$ in \cite{Korpas:2019cwg}, and leave this for future work.

\begin{table}[ht]\small\begin{center}
\renewcommand{\arraystretch}{1.2}
		\begin{tabular}{|>{$\displaystyle}Sl<{$}|>{$\displaystyle}Sr<{$}|>{$\displaystyle}Sr<{$}|>{$\displaystyle}Sr<{$}|>{$\displaystyle}Sr<{$}|}
		\hline
			\ell&    2^{6\ell}
   \,\Phi_{\frac12}^0[(\tfrac{u}{\Lambda_0^2})^{2\ell}]& 2^{4(3\ell+1)+2}\,\Phi_{\frac12}^1[(\tfrac{u}{\Lambda_1^2})^{3\ell+1}]&
			2^{12\ell+6}\,\Phi_{\frac12}^2[(\tfrac{u}{\Lambda_2^2})^{2\ell}] &2^{12\ell+18}\,\Phi_{\frac12}^3[(\tfrac{u}{\Lambda_3^2})^{\ell}]\\
			\hline
 0 & 1 & -7 & 3 & 5 \\
 1 & 19 & 1093 & 53 & 45 \\
 2 & 680 & -276947 & 1811 & 489 \\
 3 & 29557 & 83085225 & 76089 & 5843 \\
 4 & 1414696 & -27198001581 & 3545549 & 73944 \\
 5 & 71741715 & 9392089537432 & 175864168 & 972531 \\
 6 & 3781276680 & -3362682187522317 & 9096101133 & 13149681 \\
 7 & 204876497145 & 1235776027296239921 & 484817292641 & 181538893 \\
 8 & 11333072742376 & -463189292656525103411 & 26431317935635 &
   2547290047 \\
 \hline
		\end{tabular}
		\caption{List of non-vanishing $\Phi_{\frac 12}^{N_f}[u^{\ell}]$ for the massless $N_f=0,1,2,3$ theories on $\mathbb P^2$. The first column for $N_f=0$ reproduces the results of \cite{ellingsrud1995wall}, while the columns for $N_f=2,3$  
  reproduce results of \cite[p. 35/36]{Malmendier:2008db}. The column for $N_f=1$ is new to our knowledge.}\label{tableulP2}\end{center}
\end{table}
We notice that in the massless $N_f=0,1,2$ theories  there are constraints for  $\Phi_{\frac 12}^{N_f}[u^{\ell}]$ to be nonzero:
\begin{equation}\begin{aligned}\label{constraints_point_obs_p2}
N_f=0:& \quad \ell \equiv 0 \mod 2, \\
N_f=1:& \quad \ell \equiv 1 \mod 3, \\
N_f=2: &\quad \ell \equiv0 \mod 2.
\end{aligned}\end{equation}
This matches with the virtual dimensions of the moduli space for $\mathbb{P}^2$,
\be
{\rm vdim}(\CM^Q_{k})=(4-N_f)k-3-N_f/4.
\ee 
If ${\rm vdim}(\CM^Q_{k})$ is even for $\mu=1/2$, is precisely of the form in \eqref{constraints_point_obs_p2}.

To treat generic masses, we introduce for $N_f=2,3$ the mass combinations,
\be 
M_{2l}=\sum_{j=1}^{N_f} m_j^{2l},\qquad M_4'=\sum_{i<j}m_i^2m_j^2, \qquad P_{N_f}= \prod_{j=1}^{N_f} m_j.
\ee 
For $N_f=2,3$, we then find Tables \ref{PtObsNf2G} and  \ref{PtObsNf3G}, which agree with the large mass calculation in Tables \ref{PtObsNf2} and \ref{PtObsNf3} by setting $m_i=m$.
\begin{center}
\begin{table}[ht!]
\centering
\renewcommand{\arraystretch}{2}
\begin{tabular}{|Sl|Sr|}
\hline 
$\ell$ & $\displaystyle\Phi_{1/2}[u^\ell]$ for $N_f=2$\\ 
\hline 
0 & $\displaystyle\frac{P_2}{\Lambda_2^2}+\frac{3}{2^6}$\\ 
1 & $\displaystyle -\frac{7}{2^{6}} M_2$ \\
2 & $\displaystyle \frac{19}{2^6} N_2^2+\frac{23}{2^{9}}\Lambda_2^2N_2+\frac{53}{2^{18}}\Lambda_2^4$\\
3 & $\displaystyle  -\frac{421}{2^{17}}\Lambda_2^2M_2-\frac{21}{2^8}\Lambda_2^2 M_2N_2$ \\
4 & $\displaystyle \frac{85}{2^9} \Lambda_2^2 N_2^3 +\frac{1093}{2^{18}}\Lambda_2^4 M_4 +\frac{3085}{2^{16}}\Lambda_2^4N_2^2+\frac{2161}{2^{21}}\Lambda_2^6 N_2+\frac{1811}{2^{30}}\Lambda_2^8$ \\
\hline 
\end{tabular}
\caption{Table with values of $\Phi_{1/2}[u^\ell]$ for $N_f=2$ and generic masses $m_j$.}
\label{PtObsNf2G}
\end{table} 
\end{center}
\begin{center}
\begin{table}[ht!]
\centering
\renewcommand{\arraystretch}{2}
\begin{tabular}{|Sl|Sr|}
\hline 
$\ell$ & $\displaystyle\Phi_{1/2}[u^\ell]$ for $N_f=3$\\ 
\hline 
0 & $\displaystyle\frac{P_3}{\Lambda_3^3}+\frac{3}{2^6}\frac{M_2}{\Lambda_3^2}+\frac{5}{2^{18}}$\\ 
1 & $\displaystyle -\frac{7}{2^{6}}\frac{M_4'}{\Lambda_3^{2}} -\frac{3}{2^7}\frac{P_3}{\Lambda_3}-\frac{5}{2^{17}}M_2+\frac{45}{2^{30}}\Lambda_3^2$ \\
2 & $\displaystyle \frac{19}{2^6}\frac{P_3^2}{\Lambda_3^{2}}+\frac{23}{2^9}\frac{M_2P_3}{\Lambda_3}+\frac{53}{2^{18}}M_4+\frac{125 M_4'}{2^{16}}+\frac{55}{2^{21}}\Lambda_3 M_3+\frac{25}{2^{30}}\Lambda_3^2M_2+\frac{489}{2^{42}}\Lambda_3^4$\\
\hline 
\end{tabular}
\caption{Table with values of $\Phi_{1/2}[u^\ell]$ for $N_f=3$ and generic masses $m_j$. }
\label{PtObsNf3G}
\end{table} 
\end{center}

The negative powers of 2 can also be understood as follows. An insertion of $u^\ell$, gives rise to a factor $2^{3\ell}$ in the denominator since $2u$ corresponds to a 2nd Chern character. Then, the factors of 2 in the tables suggest that the class $c_{l,j}$ in (\ref{UVpartfunc}) is not an integral class, but that $2^{3l/2} c_{l,j}$ is. Thus each power of $u$ gives rise to a factor of $2^{-3}$, while that of the matter bundle is $2^{-{3\, \rm rk}(W_k^{N_f})}$,

As a result, we find that in the massless case
\begin{equation}
    2^{d_{N_f}(\ell)}\Phi_{\frac12}^{N_f}\big[(\tfrac{u}{\Lambda_{N_f}^2})^{\ell}\big]\in \mathbb Z,
\end{equation}
where $d_{N_f}$ is given by\footnote{It turns out that $d_{N_f}$ is not linear or quadratic in $N_f$. One formula that fits all four values of $N_f$ is \begin{equation}
    d_{N_f}(\ell)=3\ell+\frac{N_f}{2}(2+\ell)(N_f^2-2N_f+3).
\end{equation}}
\begin{equation}\label{denominators_p}
    d_{N_f}(\ell)=\begin{cases} 3\ell, \quad &N_f=0 \\
    4\ell+2, \quad &N_f=1, \\
    6\ell+6, \quad &N_f=2, \\
    12\ell+18, \quad &N_f=3. \end{cases}
\end{equation}
In the massive cases, the vevs $\Phi_{\frac12}^{N_f}[(\tfrac{u}{\Lambda_{N_f}^2})^{\ell}]$ are dimensionless, and so can only depend on the dimensionless ratios $\mu_j=m_j/\Lambda_{N_f}$. By the above argument, the negative powers of 2 are maximal for the top Chern class. Therefore,
\begin{equation}
    2^{d_{N_f}(\ell)}\Phi_{\frac12}^{N_f}\big[(\tfrac{u}{\Lambda_{N_f}^2})^{\ell}\big]\in \mathbb Z[\mu_1,\dots,\mu_{N_f}]
\end{equation}
are valued in the polynomial ring in the masses $\mu_j$ over the integers, with the same denominators $d_{N_f}(\ell)$  as in the massless case.

\subsection{Non-vanishing background fluxes}\label{sec:nonvanishing_bckg_flux}
As described in Part I and above, we can introduce non-trivial background fluxes $\bfk_1,\dots, \bfk_{N_f}\in H^2(X,\mathbb{Z}/2)$.

For $X=\mathbb{P}^2$ and $N_f=1$, the consistent formulation of the theory on $\mathbb{P}^2$ requires that $k_1\in \mathbb{Z}+1/2$ for $\mu=0$, while $k_1\in \mathbb{Z}$ for $\mu=1/2$. We first determine the series for the large mass expansion, and using the mock Jacobi form $F_\mu(\tau,\rho)$ (\ref{f12function}). The exponentiated point correlator then reads
\be 
\label{PhimuBG}
\begin{split}
&\Phi_{\mu}[e^{2pu/\Lambda_{1}^2}]=-\frac{2}{3}\frac{m^{3}}{\Lambda_{1}} \frac{1}{\Lambda_0^{12}}\\
&\qquad \times {\rm Coeff}_{q^0} {\rm Ser}_{m^{-1}} \left[ \left(\frac{du}{da}\right)^{12}\frac{\eta(\tau)^{27}}{P_{1}^{\text{M}}}\,C^{k_1^2}\,F_{\mu}(\tau, k_1 v) \, e^{2pu/\Lambda_{1}^2}\right].
\end{split}
\ee 
For the couplings $C\coloneqq C_{11}$ \eqref{cij} and $v\coloneqq v_1$ \eqref{Defvw}, we substitute the expansions (\ref{Cseries}) and (\ref{vseries}) to sufficiently high order. The result for $\Phi_0[u^\ell]$ is listed in Table \ref{nf1backgroundul0} and for $\Phi_{1/2}[u^\ell]$ in Table \ref{nf1backgroundul1}. The results are consistent with the decoupling limit \eqref{IR-decoupling}.

\begin{table}[ht]\small\begin{center}
\renewcommand{\arraystretch}{2}
\begin{tabular}{|>{$\displaystyle}Sl<{$}| >{$\displaystyle}Sr<{$}|>{$\displaystyle}Sr<{$}|>{$\displaystyle}Sr<{$}|}
\hline
\ell& k_1=\frac12 & k_1=\frac32 & k_1=\frac52\\
\hline
  0& -\frac{3}{4\sqrt{2}} \frac{m^1}{m^1} & -\frac{9}{4\sqrt{2}} \frac{\Lambda_1^2}{m^2} & -\frac{15}{4\sqrt{2}} \frac{\Lambda_1^6}{m^6} \\
  1& 0 & -\frac{31}{64\sqrt{2}} \frac{\Lambda_1^5}{m^3}& -\frac{155}{64\sqrt{2}} \frac{\Lambda_1^9}{m^5}  \\
 2& -\frac{13}{64\sqrt{2}} \Lambda_1^3m & -\frac{39}{64\sqrt{2}}\frac{\Lambda_1^5}{m}- \frac{567}{2^{12}\sqrt{2}} \frac{\Lambda_1^8}{m^4} & -\frac{65}{64\sqrt{2}} \frac{\Lambda_1^9}{m^5}  \\
 3& \frac{113}{2^{13}\sqrt{2}} \Lambda_1^6+\frac{50175}{2^{23}\sqrt{2}}\frac{\Lambda_1^9}{m^3} & -\frac{867}{2^{12}\sqrt{2}}\frac{\Lambda_1^8}{m^2} & -\frac{1225}{2^{10}\sqrt{2}}\frac{\Lambda_1^{12}}{m^6}\\
 4 & -\frac{879}{2^{13}\sqrt{2}} \Lambda_1^6m^2 & -\frac{2637}{2^{13}\sqrt{2}} \Lambda_1^8-\frac{7305}{2^{17}\sqrt{2}} \frac{\Lambda_1^{11}}{m^3} & -\frac{4395}{2^{13}\sqrt{2}} \frac{\Lambda_1^{12}}{m^4}\\
 \hline
\end{tabular}
\caption{List of the first few $\Phi_{0}[u^{\ell}]$ for  $N_f=1$ on $\mathbb P^2$ with non-vanishing background fluxes, $k_1=\frac12, \frac{3}{2}, \frac{5}{2}$. The expansion is determined up to $\CO(m^{-6
})$, $\CO(m^{-6
})$ and $\CO(m^{-{10}
})$, respectively.} \label{nf1backgroundul0}\end{center}
\end{table}

\begin{table}[ht]\small\begin{center}
\renewcommand{\arraystretch}{2}
\begin{tabular}{|>{$\displaystyle}Sl<{$}|  >{$\displaystyle}Sr<{$}| >{$\displaystyle}Sr<{$}| >{$\displaystyle}Sr<{$}| }
\hline
\ell& k_1=1 & k_1=2 & k_1=3 \\
\hline
  0& 1 & \frac{\Lambda_1^3}{m^3}+\frac{15}{64}\frac{\Lambda_1^6}{m^6} & \frac{\Lambda_1^8}{m^8}+\frac{45}{32}\frac{\Lambda_1^{11}}{m^{11}}\\
  1& 0 & \frac{21}{64}\frac{\Lambda_1^6}{m^4} & \frac{7}{8}\frac{\Lambda_1^{11}}{m^9} \\
 2& \frac{19}{64}\frac{\Lambda_1^3 m}{m^0} & \frac{19}{64}\frac{\Lambda_1^6}{m^2} & \frac{19}{64}\frac{\Lambda_1^{11}}{m^7}+\frac{201}{256}\frac{\Lambda_1^{14}}{m^{10}} \\
 3& -\frac{11}{2^9} \Lambda_1^6 & 0 & \frac{237}{2^9} \frac{\Lambda_1^{14}}{m^8}\\
 4& \frac{85}{2^9} \Lambda_1^6 m^2 & \frac{85}{2^9} \frac{\Lambda_1^9}{m} & \frac{85}{2^9   } \frac{\Lambda_1^{14}}{m^6}+\frac{64775}{2^{17}} \frac{\Lambda_1^9}{m}\\
 \hline
\end{tabular}
\caption{List of the first few $\Phi_{1/2}[u^{\ell}]$ for large mass $N_f=1$ on $\mathbb P^2$ with background flux, $k_1=1,2,3$. The expansion is determined up to $\CO(m^{-6})$, $\CO(m^{-10})$ and $\CO(m^{-16})$ respectively. } \label{nf1backgroundul1}\end{center}
\end{table}

As discussed in the previous section, a small or generic mass calculation requires an anti-derivative that transforms under $\slz$ rather than under a subgroup. In Appendix \ref{sec:bkgFluxes}, we perform the calculation for $N_f=1$ and $N_f=2$ with generic masses,  using the mock Jacobi form $H(\tau,\rho)$ \eqref{H12Def} for $\mu=\frac12$ rather than $F_{\frac12}(\tau,\rho)$. For specific choices of the background fluxes, the naive evaluation using this function gives different results depending on the evaluation point: For large masses, the point correlators have a well-defined $m\to\infty$ limit, for small masses they have a good massless limit $m\to0$, while for generic masses, i.e. without expanding in the masses at all, the correlators do not have either of these limits. Possible obstructions for this involve a pole $\tn$ of the anti-derivative $H(\tau,\rho(\tau))$, which is not related to the branch point but rather due to solutions of $v_j(\tn)=\frac12$ inside the fundamental domain. In Appendix \ref{sec:bkgFluxes}, we analyse this issue in some detail and discuss possible resolutions.

\section{Contributions from strong coupling singularities}\label{sec:SWcontributions}
In this section, we analyse the mass dependence of the contributions from strong coupling singularities, or SW contributions, by analysing the wall-crossing of the $u$-plane integral. The general form of these contributions was determined in \cite{Moore:1997pc}, and studied in various cases, for instance for the massless theories on specific manifolds in  \cite{Kanno:1998qj, Labastida:1995gp} and for generic masses in \cite{Marino:1998uy,moore_talk2018,Dedushenko:2017tdw}.

The type of a strong coupling singularity is determined by the Kodaira classification of singular fibres (see Table \ref{tab:kodaira}). The  monopole and dyon singularities of the pure $N_f=0$ theory are examples of $I_1$ singularities.\footnote{This viewpoint depends on the global form of the theory, in this case that of pure $\CN=2$ SYM with gauge algebra $\mathfrak{su}(2)$ \cite{Argyres:2022kon}. In particular, the SW curves for the various global forms are related by compositions of isogenies, which do not leave the type of $I_n$ singular fibres invariant \cite{Argyres:2015ffa,Closset:2023pmc}. We will comment on this further in the conclusions (Section \ref{sec:discussion}).} The collision of $n$ mutually local $I_1$ singularities gives rise to an $I_n$ singularity. If rather $n=2,3,4$ mutually non-local singularities collide, we get type $II$, $III$ or $IV$ Argyres-Douglas points. Both types of collisions can have nontrivial consequences for the partition functions in the limit.

Most of this section will deal with the $I_1$ SW contributions for generic four-manifolds.
Under various assumptions, we also generalise the arguments to $I_2$ contributions. For $K3$, we calculate the SW contributions in various examples in $N_f=1,2,3$, and study the limit to the AD mass locus in some detail. In Section \ref{sec:segrenumbers}, we relate the contributions from the instanton component in some examples to Segre numbers \cite{Gottsche:2020ass,Gottsche:2021ihz}. In Section \ref{gen_sim_type}, we discuss the general structure of SW partition functions for arbitrary configurations, extending the notion of \emph{generalised simple type} conditions familiar from the pure $N_f=0$ case.

\subsection{SW contributions of \texorpdfstring{$I_1$}{I1} singularities}\label{sec:SWI1}
Let us first study $I_1$ singularities. This is the case for generic masses with  arbitrary number $N_f$ of flavours.
The Seiberg-Witten contribution from the strong coupling cusp $j=1,\dots,N_f+2$ reads \cite{Moore:1997pc},
\be 
\label{SWcont}
Z^J_{\text{SW},j,\bfmu}=\sum_c \Lambda_{N_f}^{n(c)}\,\SW(c;J)\,  \mathop{\mathrm{Res}}_{a_j=0}\,\left[ \frac{e^{-S_{SW,j}}}{a_j^{1+n(c)}}\right]
\ee 
where $n(c)$ is the complex dimension of the monopole moduli space
\be 
n(c)=\frac{1}{8}(c^2-2\chi-3\sigma),
\ee
and the exponentiated action takes the form 
\be 
\label{SWuniversal}
e^{-S_{SW,j}}=\kappa_j\, \CA_j^\chi \,\CB_j^\sigma\, \CC_j^{\bfk^2_1}\,\cD_j^{B(\bfk_1,c)}\,\CE_j^{c^2}\, \CF_{\bfmu,j},
\ee 
where we specialise to $N_f=1$ for simplicity. The general form for the contribution of the $j$-th singularity with general $N_f$ contains a product $\prod_{i,k=1}^{N_f}\CC_{ik,j}^{B(\bfk_i,\bfk_k)}$ and similar for $\mathcal D_j$, as explained in part I. 
Here, $\CF_{\bfmu,j}$ is short hand for the four couplings,
\be 
\CF_{\bfmu,j}=f_{1,j}^{\bfmu^2} f_{2,j}^{B(K,\bfmu)}f_{3,j}^{B(c,\bfmu)}f_{4,j}^{B(\bfk_1,\bfmu)}.
\ee
For four-manifolds with $b_2^+=1$, SW-invariants are metric dependent \cite{Witten:1994cg,Kronheimer1994}. If $B(c,J^+)>0$ and $B(c,J^-)<0$, then \cite{li_liu_1995,Morgan1995,PARK2004}
\be
\label{SWwall}
\SW(c;J^+)-\SW(c;J^-)=-(-1)^{n(c)},
\ee 
The SW invariant furthermore satisfies \cite{Witten:1994cg,Morgan1995}
\be 
\label{SWsign}
\SW(-c)=(-1)^{\chi_{\rm h}} \SW(c),
\ee 
with $\chi_{\rm h}=\frac{\chi+\sigma}{4}$ the holomorphic Euler characteristic.

Changing the known result \cite{Witten:1994cg}  to our convention using (\ref{changeconvs}), the result for $N_f=0$ is 
\be
\begin{split}
Z_\bfmu[e^{2pu}]&=-2^{1-\chi_{\rm h}+\lambda}e^{\pi i \lambda/2} \sum_{c} \SW(c)\,(-1)^{B(c-K,\bfmu)}\\
&\quad \times \left(e^{\pi i \chi_{\rm h}/2} e^{2pu_1^*}+e^{2\pi i  \bfmu^2}e^{2pu_2^*} \right)
\end{split}
\ee
with $\lambda=2\chi+3\sigma$.

For instance for $X=K3$, we have $\sigma=-16$, $\chi=24$, $\chi_h=2$, $K=0$ and thus $\lambda=2\chi+3\sigma=0$. Furthermore, ${\rm SW}(c)=1$ for $c=0$, and vanishes otherwise.
The result for $K3$  and $N_f=0$ is then \cite{OGrady1992DonaldsonsPF,Witten:1994ev}
\be 
\label{SWK3}
Z_{\bfmu}[e^{2pu}]=\frac{1}{2} \left( e^{2pu_1^*}-e^{2\pi i \bfmu^2} e^{2pu_2^*}\right).
\ee 

From the relation in Section \ref{sec:wall-crossing} between the wall-crossing of the $u$-plane integral and SW contributions, we have for the $j$-th singularity, 
\be 
\begin{split}
\left[\Phi_\bfmu^{J^+} -\Phi^{J^-}_\bfmu \right]_j&=\kappa_j \sum_c (-1)^{n(c)} \tfrac{1}{2} (\sgn(B(c,J^+))-\sgn(B(c,J^-)))\\
&\times \Lambda_{N_f}^{n(c)}\,\mathop{\mathrm{Res}}_{a_j=0}\left[ a_j^{-1-n(c)} \CA_j^\chi \,\CB_j^\sigma\, \CC_j^{\bfk^2_1}\,\cD_j^{B(\bfk_1,c)}\,\CE_j^{c^2}\, \CF_{\bfmu,j}\right].
\end{split}
\ee  
We substitute $\chi=4-\sigma$ in the right hand side, such that 
\be
\label{SWWC}
\begin{split} 
&\left[\Phi_\bfmu^{J^+} -\Phi^{J^-}_\bfmu \right]_j=-\kappa_j \sum_c (-1)^{(c^2-\sigma)/8}\tfrac{1}{2} (\sgn(B(c,J^+))-\sgn(B(c,J^-)))\\
&\times \Lambda_{N_f}^{(c^2-\sigma)/8-1} \mathop{\mathrm{Res}}_{a_j=0}\left[ a_j^{-(c^2-\sigma)/8} \CA_j^4 \,(\CB_j/\CA_j)^\sigma\, \CC_j^{\bfk^2_1}\,\cD_j^{B(\bfk_1,c)}\,\CE_j^{c^2}\, \CF_{\bfmu,j}\right].
\end{split}
\ee 
\subsubsection*{Contribution from monopole cusp $u_1^*$}
The left hand side is the contribution from the $u$-plane integral for $j=1$, which is obtained using the substitution $\tau=-1/\tau_1$. We introduce
\be 
\begin{split}
&A_1(\tau_1)=\tau_1^{1/2}\, A(-1/\tau_1),\\
&\quad \quad = 2^{3/4}\Lambda_{N_f}^{-1/2}s_1^{-1/4}\,(27g_3(u_1^*))^{1/12}+\dots,\\
&B_1(\tau_1)=B(-1/\tau_1),\\
&C_1(\tau_1)=e^{-\pi i v_1^2/\tau_1}\, C(-1/\tau_1),\\
&\left(\frac{da}{d\tau}\right)_1(\tau_1)=\tau_1^{-3} \left(\frac{da}{d\tau} \right)(-1/\tau_1),\\
&v_1(\tau_1)=\tau_1\, v(-1/\tau_1).
\end{split}
\ee 
The substitution brings the analogue of the integral (\ref{WCPhi}) near singularity $j=1$ to the form 
\be 
\label{uWC}
\begin{split} 
\left[\Phi_\bfmu^{J^+} -\Phi^{J^-}_\bfmu \right]_1&= - e^{\pi i \sigma/4}\CK_{N_f} \lim_{Y\to \infty}\int_{-1/2+iY}^{1/2+iY} d\tau_j \left(\frac{da}{d\tau}\right)_1 A_1^\chi B_1^\sigma C_1^{\bfk_1^2}\, e^{\pi i B(\bfmu,K)}\\
& \times \sum_{\bfk\in L+K/2} \tfrac{1}{2} (\sgn(B(c+\bfb_1,J^+))-\sgn(B(c+2\bfb_1,J^-))) \\
&\times e^{\pi i B(\bfk,K)}\,q_1^{-\bfk^2/2}e^{-2\pi i B(\bfk,\bfk_1 v_1)} e^{-2\pi i B(\bfk,K/2-\bfmu)}\\
&=-2\pi i\, e^{\pi i \sigma/4}\,\CK_{N_f} \\
&\times \sum_c  \tfrac{1}{2} (\sgn(B(c+2\bfb_1,J^+))-\sgn(B(c+2\bfb_1,J^-)))\\
& \times \mathop{\mathrm{Res}}_{a_1=0}\left[ A_1^4 (B_1/A_1)^\sigma C_1^{\bfk_1^2}\, q_1^{-c^2/8}e^{-\pi i B(c,\bfk_1 v_1)}e^{\pi i B(c+K,\bfmu)}\right]
\end{split}
\ee
where we substituted $c=2\bfk$ and $\bfb_1=\bfk_1\,{\rm Im}(v_1)$. We assume that $v_1=\CO(q_1)$ for $\tau_1\to 0$. We hope to check this in the future.

Comparing \eqref{SWWC} and \eqref{uWC}, we can read off the couplings for singularity $j=1$,
\be 
\begin{split} 
&\kappa_1 \CA_1^4/\Lambda_{N_f}=2\pi i\,\CK_{N_f} A_1^4,\\ 
& e^{-\pi i/8} (a_1/\Lambda_{N_f})^{1/8} \CB_1/\CA_1=e^{\pi i/4} B_1/A_1,\\
& \CC_1=C_1,\\
& \cD_1=e^{-\pi i v_1},\\
& (-a_1/\Lambda_{N_f})^{-1/8}\CE_1=q_1^{-1/8},\\
&\CF_{\bfmu,1}=e^{\pi i B(c+K,\bfmu)}.
\end{split}
\ee 

In the large mass expansion, the leading behavior for $\CC_1$ and $\cD_1$ for $N_f=1$ can be determined from \eqref{Cseries} and \eqref{wseries} by letting $\tau\to 0$. For the other couplings, we substitute
 \eqref{CNNf},
\be 
\label{CABE}
\begin{split}
&\CA_1= 2^{1/2} e^{\pi i/8} \kappa_1^{-1/4} \Lambda_{N_f}^{-1/2} \left( \frac{du}{da} \right)_1^{1/2} \\
&\quad = 2\, e^{\pi i/8} \kappa_1^{-1/4} s_1^{-1/4}\Lambda_{N_f}^{-1/2} (27g_3(u_1^*))^{1/12},\\
&\CB_1=2^{1/4}e^{\pi i/2}\kappa_1^{-1/4}\Lambda_{N_f}^{1/8} \frac{B_1}{a_1^{1/8}}\\
& \quad = 2^{5/4} e^{\pi i(N_f+5)/8} \kappa_1^{-1/4} 
s_1^{-3/16}\Lambda_{N_f}^{-(2N_f+3)/8} \\
&\qquad \times \left( (4-N_f) P_{N_f}^{\rm M}(u_1^*)\,(27g_3(u_1^*))^{1/2} \right)^{1/8},\\
&\CE_1=\left( (-1)^{N_f}s_1^{-3/2} \Lambda_{N_f}^{2N_f-9}\frac{64}{4-N_f} \frac{\left(27g_3(u_1^*)\right)^{3/2}}{P_{N_f}^{\rm M}(u_1^*)}  \right)^{1/8}.
\end{split}
\ee 
Combining all the factors for $\bfk_1=0$, we find assuming the simple type condition,
\be 
\label{ZSW1gen}
\begin{split}
Z_{\text{SW},1,\bfmu}&=\kappa_1^{1-\chi_{\rm h}}2^{2\chi_h+\lambda} e^{\frac{\pi i}{2}(\chih+\lambda)}(-1)^{N_f\chih}s_1^{-3\chih/2-\lambda/8} \Lambda_{N_f}^{(2N_f-3)\chi_{\rm h}-\lambda} \\
&\quad \times \left((4-N_f)P_{N_f}^{\rm M}(u_1^*) \right)^{-\chih}(27g_3(u_1^*))^{\chih/2+\lambda/6}\\
&\quad \times  \sum_c \SW(c)\,e^{\pi i B(c+K,\bfmu)}.
\end{split}
\ee 
In order to  fix $\kappa_1$, we apply this to $X=K3$ with $\bfk_1=0$, 
\be \label{Z1K3}
Z_{\text{SW},1,\bfmu}=-\frac{2^4}{\kappa_1} \Lambda_{N_f}^{-6+4N_f} (4-N_f)^{-2}
\, (P^{\rm M}_{N_f}(u_1^*))^{-2}\,27g_3(u_1^*), 
\ee 
where we used that $s_1$ is a third root of unity. Comparing with \eqref{SWK3} for $N_f=0$, we then find $\kappa_1=-2$. This holds for generic $N_f$ since this number is independent of the masses.
Substitution in (\ref{ZSW1gen}) and using (\ref{SWsign}) and $s_1=1$ gives
\be 
\label{ZSW1gen2}
\boxed{
\begin{aligned}
Z_{\text{SW},1,\bfmu}&=-2^{1+\chi_h+\lambda} e^{\frac{\pi i}{2}(\chih+\lambda)}(-1)^{N_f\chih}\Lambda_{N_f}^{(2N_f-3)\chi_{\rm h}-\lambda} C_1^{\bfk_1^2}\\
&\quad \times \left((4-N_f)P_{N_f}^{\rm M}(u_1^*) \right)^{-\chih}(27g_3(u_1^*))^{\chih/2+\lambda/6}\\
&\quad \times  \sum_c \SW(c)\,e^{\pi i B(c-K,\bfmu)}e^{-\pi i v_1 B(\bfk_1,c)}.
\end{aligned}}
\ee 
This agrees for $N_f=0$, with the well-known SW contribution at the monopole cusp \cite[eq. (2.17)]{Witten:1994cg}. Using \eqref{phat_Deltan} for $n=1$, this matches for any $N_f$  with \cite{Moore:1997pc,Marino:1998uy}.

\subsubsection*{Contribution from dyon cusp $u_2^*$}
Let us also consider the contributions from the other cusps.
The local coupling $\tau_2$ reads in terms of the effective coupling $\tau=2-1/\tau_2$, such that $\tau\to 2$ for $\tau_2\to i\infty$. We introduce the following couplings
\be 
\begin{split}
&A_2(\tau_2)=\tau_2^{1/2}\, A(2-1/\tau_2),\\
&\quad = 2^{3/4}\Lambda_{N_f}^{-1/2}\,s_2^{-1/4}\,(27g_3(u_2^*))^{1/12}+\dots,\\
&B_2(\tau_2)=B(2-1/\tau_2),\\
&C_2(\tau_2)=e^{-\pi i v_2^2/\tau_2}\, C(2-1/\tau_2),\\
&\left(\frac{da}{d\tau}\right)_2(\tau_2)=\tau_2^{-3} \left(\frac{da}{d\tau} \right)(2-1/\tau_2),\\
&v_2(\tau_2)=\tau_2\, v(-1/\tau_2).
\end{split}
\ee 
Then we find from the $u$-plane for wall-crossing
\be 
\label{uWC2}
\begin{split} 
\left[\Phi_\bfmu^{J^+} -\Phi^{J^-}_\bfmu \right]_2
&=-2\pi i\,e^{-2\pi i \bfmu^2}\, e^{\pi i \sigma/4}\,\CK_{N_f}\\
&\times \sum_c  \tfrac{1}{2} (\sgn(B(c+2\bfb_1,J^+))-\sgn(B(c+2\bfb_1,J^-)))\\
& \times \mathop{\mathrm{Res}}_{a_2=0}\left[ A_2^4 (B_2/A_2)^\sigma C_2^{\bfk_1^2}\, q_2^{-c^2/8}e^{-\pi i B(c,\bfk_1 v_2)}e^{\pi i B(c+K,\bfmu)}\right].
\end{split}
\ee
Following the same steps as for the monopole cusp, we find for the various couplings
\be
\begin{split}
&\CA_2 = 2\, e^{\pi i/8} \kappa_2^{-1/4} s_2^{-1/4}\Lambda_{N_f}^{-1/2} (27g_3(u_2^*))^{1/12},\\
&\CB_2=  2^{5/4} e^{\pi i(N_f+5)/8} \kappa_2^{-1/4} 
s_2^{-3/16}\Lambda_{N_f}^{-(2N_f+3)/8}\\
&\qquad \times \left( (4-N_f) P_{N_f}^{\rm M}(u_2^*)\,(27g_3(u_2^*))^{1/2} \right)^{1/8},\\
&\CC_2=C_2,\\
&\cD_2=e^{-\pi i v_2}\\
&\CE_2=\left( (-1)^{N_f}s_2^{-3/2} \Lambda_{N_f}^{2N_f-9}\frac{64}{4-N_f} \frac{\left(27g_3(u_2^*)\right)^{3/2}}{P_{N_f}^{\rm M}(u_2^*)}  \right)^{1/8},\\
&\CF_{\bfmu,2}=e^{\pi i B(c+K,\bfmu)-2\pi i \bfmu^2}.
\end{split}
\ee
This then gives for $\bfk_1=0$ (which requires $\bfmu=K/2$ for $N_f=1$),
\be 
\begin{split}
Z_{\text{SW},2,\bfmu}&=\kappa_2^{1-\chi_{\rm h}}2^{2\chi_h+\lambda} e^{\frac{\pi i}{2}(\chih+\lambda)}(-1)^{N_f\chih}s_2^{-3\chih/2-\lambda/8} \Lambda_{N_f}^{(2N_f-3)\chi_{\rm h}-\lambda}\\
&\quad \times \left((4-N_f)P_{N_f}^{\rm M}(u_2^*) \right)^{-\chih}(27g_3(u_2^*))^{\chih/2+\lambda/6}\\
&\quad \times  \sum_c \SW(c)\,e^{\pi i B(c+K,\bfmu)-2\pi i \bfmu^2}.
\end{split}
\ee 
Applying this to $X=K3$, we find
\be
Z_{\text{SW},2,\bfmu}=-\frac{2^4}{\kappa_2}e^{-2\pi i\bfmu^2}\Lambda_{N_f}^{-6+4N_f} (4-N_f)^{-2}
\, (P^{\rm M}_{N_f}(u_2^*))^{-2}\,27g_3(u_2^*).
\ee
For $N_f=0$, we have $27g_3(u_2^*)=-1=e^{\pi i}$. With $s_2=e^{4\pi i/3}$ (see Tables \ref{sj_phasesLight} and \ref{sj_phasesHeavy}),
we thus arrive at $\kappa_2=-2$. The general contribution from this cusp is
\be \label{ZSW1gen2j=2}
\boxed{
\begin{aligned}
Z_{\text{SW},2,\bfmu}&=-2^{1+\chi_h+\lambda} e^{\frac{\pi i}{2}(\chih+\lambda)}(-1)^{N_f\chih}s_2^{-3\chih/2-\lambda/8} \Lambda_{N_f}^{(2N_f-3)\chi_{\rm h}-\lambda}C_2^{\bfk_1^2}\\
&\quad \times \left((4-N_f)P_{N_f}^{\rm M}(u_2^*) \right)^{-\chih}(27g_3(u_2^*))^{\chih/2+\lambda/6}\\
&\quad \times  \sum_c \SW(c)\,e^{\pi i B(c-K,\bfmu)-2\pi i \bfmu^2}e^{-\pi i v_2 B(\bfk_1,c)}.
\end{aligned}}
\ee 
\subsubsection*{Contribution from hypermultiplet cusp $u_3^*$}
We will continue with the contribution from a singularity associated to one of the fundamental hypermultiplets. We label this singularity as  $u_3^*$, and we assume that this singularity is approached for $\tau\to 1\in \mathbb{Q}$. This is case for the hypermultiplet singularity in the $N_f=1$ theory. The contribution from hypermultiplet singularities in other theories, possibly at other points in $\mathbb{Q}$, can be determined similarly. 

The coupling $\tau$  then  reads $\tau=1-1/\tau_3$ in terms of the local coupling $\tau_3$ near the singularity $j=3$. For simplicity, we fix $\bfmu=K/2$ such that we can choose vanishing background fluxes for the flavour symmetry. Then the wall-crossing of the $u$-plane integral at $u_3$ reads
\be 
\label{uWC3}
\begin{split} 
\left[\Phi_{K/2}^{J^+} -\Phi^{J^-}_{K/2} \right]_3
&=-2\pi i\, e^{-\pi i \sigma/2}\,\CK_{N_f}\\
&\times \sum_c  \tfrac{1}{2} (\sgn(B(c+2\bfb_1,J^+))-\sgn(B(c+2\bfb_1,J^-)))\\
& \times \mathop{\mathrm{Res}}_{a_3=0}\left[ A_3^4 (B_3/A_3)^\sigma C_3^{\bfk_1^2}\, q_3^{-c^2/8}e^{-\pi i B(c,\bfk_1 v_3)}e^{\pi i B(c+K,K/2)}\right],
\end{split}
\ee
where we used that $K^2=\sigma \mod 8$.

We then find for the various couplings
\be
\begin{split}
&\CA_3 = 2\, e^{\pi i/8} \kappa_3^{-1/4} s_3^{-1/4}\Lambda_{N_f}^{-1/2} (27g_3(u_3^*))^{1/12},\\
&\CB_3=  2^{5/4} e^{\pi i(N_f+3)/8} \kappa_3^{-1/4} 
s_3^{-3/16}\Lambda_{N_f}^{-(2N_f+3)/8}\\
&\qquad \times \left( (4-N_f) P_{N_f}^{\rm M}(u_3^*)\,(27g_3(u_3^*))^{1/2} \right)^{1/8},\\
&\CC_3=C_3,\\
&\cD_3=e^{-\pi i v_3},\\
&\CE_3=\left( (-1)^{N_f}s_3^{-3/2} \Lambda_{N_f}^{2N_f-9}\frac{64}{4-N_f} \frac{\left(27g_3(u_3^*)\right)^{3/2}}{P_{N_f}^{\rm M}(u_3^*)}  \right)^{1/8},\\
&\CF_{K/2,3}=e^{\pi i B(c+K,K/2)}.
\end{split}
\ee

We assume that the normalization $\kappa_3$ equals $\kappa_1=\kappa_2=-2$. We then arrive for the general expression at 
\be \label{ZSW1gen2j=3}
\boxed{
\begin{aligned}
Z_{\text{SW},3,K/2}&=-2^{1+\chi_h+\lambda} e^{\frac{\pi i}{2}(\chih-\lambda)}(-1)^{N_f\chih}s_3^{-3\chih/2-\lambda/8} \Lambda_{N_f}^{(2N_f-3)\chi_{\rm h}-\lambda}C_3^{\bfk_1^2}\\
&\quad \times \left((4-N_f)P_{N_f}^{\rm M}(u_3^*) \right)^{-\chih}(27g_3(u_3^*))^{\chih/2+\lambda/6}\\
&\quad \times  \sum_c \SW(c)\,e^{\pi i B(c,K/2)}e^{-\pi i v_3B(\bfk_1,c)}.
\end{aligned}}
\ee 
We have thus derived the SW contributions \eqref{ZSW1gen2}, \eqref{ZSW1gen2j=2}, \eqref{ZSW1gen2j=3} for singularities at the cusps $\tau=0,2,1$, as is the case for instance in $N_f=1$ with a generic mass. 

\subsection{SW contributions of \texorpdfstring{$I_2$}{I2} singularities}\label{sec:SWI2}
In the SW curves for $N_f\geq 2$, $I_r$ singularities with $1<r\leq 2N_f-2$ can occur for special values of the masses. Indeed, as shown in \eqref{dDelta}, such higher $I_r$ singularities occur precisely if some of the $N_f$ masses (anti-)align. In such cases, there is a singular point on the Coulomb branch where $r$ mutually local dyons become massless. The `maximal' $I_{2N_f-2}$ singularity occurs precisely in the massless case.\footnote{For $N_f=4$, in the massless limit the six singularities collide in an $I_0^*$ singularity, rather than an $I_6$.}

Let us consider such a configuration with an $I_{r> 1}$ singularity. The (complex) virtual dimension of the $r$-monopole equation reads \cite{bryan1996}
\be \label{vdim_r}
n_r(c)=\frac{r(c^2-\sigma)-2(\chi+\sigma)}{8}.
\ee 
For $r>1$, the moduli spaces $\CM_r$ are non-compact. There is a global symmetry  $SU(r+1)$ which acts on the moduli space of $r$-monopole equations. The non-compactness can be mitigated by studying the equivariant cohomology of $\CM_r$ with respect to the $SU(r+1)$ action \cite{LoNeSha}. This corresponds to deforming the $I_r$ singularities to $I_1$ singularities by making the masses generic, for which the spaces $\CM_1$ are compact.

We will proceed quite heuristically in the following by {\it assuming} the existence of well-defined numerical Seiberg-Witten invariants $\SW_r$ for $I_r$ singularities, and a wall-crossing formula for these invariants similar to (\ref{SWwall}). We will then show for K3 that smoothness of the partition function as function of the mass implies sum rules for $\SW_2$ invariants in terms of ordinary $\SW_1$ invariants. The existence of a smooth limit is non-trivial, even with the freedom to fix $\SW_2(c)$ with $n_2(c)\leq 0$. We therefore consider it worthwhile to include it here.

We will consider $r=2$ in what follows. We let the singularity $u_3^*$ be of type $I_2$. This makes the analysis suited for the case of $N_f=2$ with $m_1=m_2$, as can be seen from Fig. \ref{fig:nf2domain}. 
Analogously to \eqref{SWcont}, the contribution from this singularity then reads
\be 
\label{SWcont_I2}
Z^J_{\text{SW},3,\bfmu}=\sum_c \SW_2(c;J) \Lambda_{N_f}^{n_2(c)}\,\mathop{\mathrm{Res}}_{a_3=0}\,\left[ a_3^{-1-n_2(c)}e^{-S_{SW,3}}\right],
\ee 
where $\SW_2(c;J)$ is a SW invariant for the $r$-monopole equations. These invariants are not expected to depend on the same four-manifold data as the $I_1$ SW-invariants. An alternative way to express this contribution is as the $q_3^0$ term,
\be \label{I2_contr_SW}
Z^J_{\text{SW},3,\bfmu}=\sum_c \SW_2(c;J)\,\Lambda_{N_f}^{n_2(c)}\, {\rm Coeff}_{q_3^0}\left[ \frac{1}{\pi i}\frac{da_3}{d\tau_3} a_3^{-1-n_2(c)}e^{-S_{SW,3}}\right].
\ee 
Moreover, in the absence of background fluxes, the effective action takes the form
\be 
e^{-S_{SW,3}}=\kappa_3\, \CA_3^\chi \,\CB_3^\sigma\, \CE_3^{c^2}\, \CF_{\bfmu,3}.
\ee 
In analogy with (\ref{SWwall}), we {\it assume} that these invariants change under wall-crossing as
\be
\SW_r(c;J^+)-\SW_r(c;J^-)=-(-1)^{n_r(c)}.
\ee 
The relation between the $u$-plane wall-crossing and the SW-contributions then gives
\be
\begin{split}
\left[ \Phi_\bfmu^{J^+} - \Phi_{\bfmu}^{J^-}\right]_3&=\kappa_3 \sum_c (-1)^{n_r(c)} \tfrac{1}{2} (\sgn(B(c,J^+))-\sgn(B(c,J^-)))\\
&\quad \times \Lambda_{N_f}^{n_2(c)}\,\mathop{\mathrm{Res}}_{a_3=0}\left[ a_3^{-1-n_2(c)}\CA_3^\chi \CB_3^\sigma \CE_3^{c^2} \CF_{\bfmu,3}\right].
\end{split} 
\ee 
By substitution of $\chi+\sigma=4$, this becomes
\be
\label{r2SWcont}
\begin{split}
\left[ \Phi_\bfmu^{J^+} - \Phi_{\bfmu}^{J^-}\right]_3&=-\kappa_3 \sum_c (-1)^{(c^2-\sigma)/4} \tfrac{1}{2} (\sgn(B(c,J^+))-\sgn(B(c,J^-)))\\
&\quad \times \Lambda_{N_f}^{(c^2-\sigma)/4-1}\,\mathop{\mathrm{Res}}_{a_3=0}\left[ a_3^{-(c^2-\sigma)/4}\CA_3^4 (\CB_3/\CA_3)^\sigma \CE_3^{c^2} \CF_{\bfmu,3}\right].
\end{split} 
\ee

To proceed, we assume that the $I_2$ singularity is at $\tau\to 1$ in the fundamental domain $\CF(\bfm)$, and choose vanishing background fluxes for simplicity. As a result, the 't Hooft flux is fixed $\bfmu=K/2$. 
We can determine the couplings $\CA_3,\dots , \CF_{\bfmu,3}$ by comparing (\ref{uWC3}) and (\ref{r2SWcont}).
This gives 
\be \label{ABE3}
\begin{split} 
\CA_3&=2^{1/2} e^{\pi i/8}\kappa_3^{-1/4}\Lambda_{N_f}^{-1/2}\left( \frac{du}{da}\right)_3^{1/2},\\
\CB_3&=2^{1/4}\kappa_3^{-1/4}e^{-\pi i/8}\Lambda_{N_f}^{1/4}\frac{B_3}{a_3^{1/4}},\\
\CE_3&=e^{-\pi i/4}\Lambda_{N_f}^{-1/4}\,a_3^{1/4}q_3^{-1/8},\\
\CF_{K/2,3}&=e^{\pi i B(c+K,K/2)},
\end{split}
\ee 
where 
\be 
\left(\frac{da}{du} \right)_3(\tau_3)=\tau_3^{-1} \left(\frac{da}{du} \right)(1-1/\tau_3).
\ee 

\subsubsection*{$\boldsymbol{I_2}$ singularity in $\boldsymbol{N_f=2}$}
We continue by sketching the determination of these couplings for the $I_2$ singularity in the $N_f=2$ theory with equal masses. The functions $u$ and $da/du$ are known exactly for this theory \cite[Section 5.1]{Aspman:2021vhs}, we recall them in  \eqref{nf2mmquantities}.\footnote{\label{ftn:sign}We have corrected the sign of $da/du$ in \eqref{nf2mmquantities} compared to \cite[Eq. (5.4)]{Aspman:2021vhs}, such that $da/du\to i 2^{3/2}\,q^{1/4}$ per our convention \eqref{decouplingConv} for the weak coupling limit.} It is then straightforward to determine couplings near the $I_2$ singularity $\tau=1-1/\tau_3\to 1$. For example, $(da/du)_3$ is given by (\ref{daduI2}), with leading terms in (\ref{dauq3}). 
 
One similarly determines the first terms in the $q_3-$expansions of $u_3$ (\ref{uq3}),  and $a_3$ (\ref{aq3}) at this singularity. Substitution in (\ref{ABE3}) then allows to the determine the contribution from the $I_2$ singularity for the equal mass theory.

\subsection{Results for \texorpdfstring{$K3$}{K3}}\label{sec:results_K3}

Using the above formulas, it is straightforward to explicitly calculate SW contributions for specific four-manifolds, for generic as well as specific masses. We focus in this subsection on evaluations for $X=K3$, and discuss $N_f=1,2,3$ separately.
\subsubsection*{$\boldsymbol{N_f=1}$}
For simplicity, we consider a vanishing background flux, $\bfk_1=0$,  such that $\bfmu=0$ necessarily. The full correlation function (\ref{Zcorr}) then reads,
\be 
Z_0(m)[u^\ell]=\sum_{j=1}^3 Z_{SW,j,0}(m)[u^\ell].
\ee 
We can evaluate this for generic mass $m$ using the topological data mentioned above Equation (\ref{SWK3}). This gives the results listed in Table \ref{PtObsNf1K3}. 

Similarly to the case for the $u$-plane integral for $\mathbb P^2$ without background fluxes,  large mass and small mass expansions of the singularities $u_j^*$ and their contributions give identical results, which demonstrates that these correlation functions are smooth functions of $m$. Given the intricate cubic roots (see discussion around \eqref{nf1_sing_exp}), it is remarkable that the mass expansions terminate at finite order.\footnote{This is possibly a consequence of Vieta's formula \eqref{Vieta_sing}: Both small and large mass expansions of the singularities $u_j^*$ are infinite series, while their sum is $\sum_{j=1}^3u_j^*=m^2$. By the fundamental theorem of symmetric polynomials, a similar statement can be made about sums $\sum_{j=1}^{2+N_f}P(u_j^*)$ with $P$ a polynomial. However, as is clear from the explicit expressions \eqref{ZSW1gen2}, \eqref{ZSW1gen2j=2}, \eqref{ZSW1gen2j=3}, the correlators $Z_0[u^\ell]$ are sums of such kind with $P$ a \emph{rational} function. }
\begin{center}
\begin{table}[ht!]
\centering
\renewcommand{\arraystretch}{2}
\begin{tabular}{|Sl|Sr|}
\hline 
$\ell$ & $\displaystyle Z_{0}(m)[u^\ell]$ for $N_f=1$\\ 
\hline 
0 & $\displaystyle 64\frac{m^2}{\Lambda_1^2}$\\ 
1 & $\displaystyle 64\frac{m^4}{\Lambda_1^2}+8\Lambda_1m$ \\
2 & $\displaystyle 64\frac{m^6}{\Lambda_1^2}+16\Lambda_1m^3+\frac{9}{4}\Lambda_1^4$\\
3 & $\displaystyle 64\frac{m^8}{\Lambda_1^2}+24\Lambda_1m^5+\frac{9}{2}\Lambda_1^4m^2$ \\
4 & $\displaystyle 64\frac{m^{10}}{\Lambda_1^2}+32\Lambda_1m^7+\frac{31}{4}\Lambda_1^4m^4+\frac{27}{16}\Lambda_1^7m$ \\
5 & $\displaystyle 64\frac{m^{12}}{\Lambda_1^2}+40\Lambda_1m^9+12\Lambda_1^4m^6+\frac{45}{16}\Lambda_1^7m^3-\frac{243}{1024}\Lambda_1^{10}$\\
6 & $\displaystyle 64\frac{m^{14}}{\Lambda_1^2}+48\Lambda_1m^{11}+\frac{69}{4}\Lambda_1^4m^8+\frac{9}{2}\Lambda_1^7m^5+\frac{1215}{1024}\Lambda_1^{10}m^2$\\
\hline 
\end{tabular}
\caption{Table with values of $Z_0[u^\ell]$ for $N_f=1$. }
\label{PtObsNf1K3}
\end{table} 
\end{center}
The dimension of the moduli space and the rank of the matter bundle are 
\be 
\begin{split}
&{\dim}(\CM^Q_k)=3k-2,\\
&{\dim}(\CM^{\rm i}_k)=4k-6,\\
&{\rm rk}(W_k)=-k+4.
\end{split}
\ee 
This is positive for small $k$, and as a result the dimension of the moduli space of non-Abelian monopoles is larger than the moduli space of instantons for these values. The point observable $u$ is a 4-form on the moduli space $\CM^Q_k$ of non-Abelian SW equations. 
\vspace{.3cm}\\
{\it Large mass limit and expansion}\\
The decoupling limit $m\to \infty$ (\ref{IR-decoupling}) applied to $Z_0[u^\ell]$ does not exist, since the contribution of $u_3^*$ diverges. There is thus no smooth decoupling limit to the $N_f=0$ result as before for $\mathbb{P}^2$. However motivated by the distinction between the monopole and instanton component, see Figure \ref{fig:branches} in Section \ref{Interlude}, we can consider the partition function for the instanton component as the sum of the contribution of $u_1^*$ and $u_2^*$,
\be 
\label{Zinstbr}
Z^{\rm i}_0(m)[u^\ell]=Z_{SW,1,0}(m)[u^\ell]+Z_{SW,2,0}(m)[u^\ell].
\ee 
Then the decoupling limit applied to $Z^{\rm i}_0(m)[u^\ell]$ is consistent with the $N_f=0$ result, that is to say, it vanishes in the limit. This is due to the dimension of the moduli space not being a multiple of 4 in this case, such that insertions of point observables give a vanishing answer.

 We can consider the {\it large} mass expansion of $Z_0^{\rm i}$. In contrast to $Z_0(m)[u^\ell]$, this is an infinite series. The first few terms are,
\be 
\label{K3instlargemass}
\begin{split}
&Z_0^{\rm i}(m)[u^0]= -\frac{3}{4}\frac{\Lambda_1^4}{m^4}-\frac{5}{16}\frac{\Lambda_1^7}{m^7}-\frac{63}{512}\frac{\Lambda_1^{10}}{m^{10}}-\frac{99}{2048}\frac{\Lambda_1^{13}}{m^{13}}+\dots, \\
&Z_0^{\rm i}(m)[u^1]= -\frac{\Lambda_1^4}{m^2}-\frac{7}{16}\frac{\Lambda_1^7}{m^5}-\frac{175}{1024}\frac{\Lambda_1^{10}}{m^{8}}-\frac{273}{4096}\frac{\Lambda_1^{13}}{m^{11}}+\dots, \\
&Z_0^{\rm i}(m)[u^3]= -\frac{5}{8}\frac{\Lambda_1^7}{m^3}-\frac{245}{1024}\frac{\Lambda_1^{10}}{m^6}-\frac{189}{2048}\frac{\Lambda_1^{13}}{m^{9}}-\frac{4719}{131072}\frac{\Lambda_1^{16}}{m^{12}}+\dots, 
\end{split}
\ee 
The infinite series demonstrate that infinitely many instanton sectors contribute. For example, we deduce from $Z^{\rm i}_0[u^0]$ that the top Chern class has maximal degree $2\dim(\CM^{\rm i}_k)$ and thus exceeds the rank of the matter bundle. In other words, it is actually a matter sheaf rather than a matter bundle. 
\vspace{.3cm}\\
{\it Small mass expansion}\\
Instead of a large mass expansion, one can also make a small mass expansion of the contributions of $u_1^*$ and $u_2^*$. The coefficients turn out to be complex numbers, obscuring their interpretation as intersection numbers. Moreover the expansion involves negative powers of $\Lambda_1$, which suggests that the terms arise from sectors with negative instanton numbers $k$. Altogether, the large $m$ expansion thus seems more physical.
\vspace{.3cm}\\
{\it Limit to AD mass }\\
In section \ref{sec:Contributions}, we argued that in an AD limit $m\to \mad$, we can split the singularities into two sets $S$ and $S'$, where the singularities in $S'$ collide in the AD point, while the ones in $S$ are not involved in the limit. The contribution from $S'$ we denote in \eqref{ZmAD} by $Z_{\widetilde{\text{AD}}}$. In our notation, the singularities $u_2^*$ and $u_3^*$ then give the contribution $Z_{\widetilde{\text{AD}}}$. 

In the limit $m\to m_{\text{AD}}=\tfrac{3}{4} \Lambda_1$ for $N_f=1$, the contributions from $u_2^*$ and $u_3^*$ diverge individually as $(m-\mad)^{-\frac12}$, where the coefficients for $u_2^*$ and $u_3^*$ differ by a minus sign. Their sum is correspondingly regular, and we find a finite limit $Z_{\widetilde{\text{AD}}}$. The conditions for a smooth limit were studied in \cite{Marino:1998uy}, and led to sum rules for SW invariants and the notion of superconformal simple type. For the type $II$ AD point on $K3$ there is no such sum rule, which is indeed not necessary because the partition function itself is regular without any constraints. 

More specifically, we calculate $Z_{\widetilde{\text{AD}}}[1]=\frac{320}{9}$, while the contribution from $u_1^*$ is $\frac{4}{9}$.  All together it reproduces the result  in the Table \ref{PtObsNf1K3} at $m=m_{\text{AD}}$. For $\ell=1$, the first cusp contributes $-\frac{5}{12} \Lambda_1^2$, and the sum of $u_2^*$ and $u_3^*$ is $Z_{\widetilde{\text{AD}}}[u^1]=\frac{80}{3} \Lambda_1^2$. Their sum again agrees with the table.
We can in fact calculate the AD contribution for any $\ell$,\footnote{This implies that $Z_{\widetilde{\text{AD}}}[e^{2pu}]=Z_{\widetilde{\text{AD}}}(1)e^{2p\uad}$. This generating function $Z_{\widetilde{\text{AD}}}^1(p)\coloneqq Z_{\widetilde{\text{AD}}}[e^{2pu}]$ then satisfies $\Delta(\frac12 \partial_p)Z_{\widetilde{\text{AD}}}^1(p)$ with $\Delta(u)=u-\uad$, similar to the structure result in section \ref{gen_sim_type}.}
\begin{equation}\label{ZtildeADNf1K3}
    Z_{\widetilde{\text{AD}}}[u^\ell]=\frac{320}{9}\left(\frac{3\Lambda_1^2}{4}\right)^\ell=Z_{\widetilde{\text{AD}}}[1]\, \uad^\ell,
\end{equation}
where $\uad=\frac34\Lambda_1^2$.
Clearly, all point correlators are nonzero. This is an indication that the selection rule associated with the SCFT is not valid in this case \cite{Shapere:2008zf,Moore:2017cmm,Gukov:2017}. 
This selection rule reads $3\ell=1$ for $K3$ in the absence of surface observables, and so it does not have any solutions for point observables only. We conclude that this selection rule is not valid for the AD point of the SQCD curve. Rather, it holds for the AD curve obtained by promoting the deformation parameters to operators and assigning scaling dimensions \cite{Argyres:1995xn, Moore:2017cmm, Marino:1998uy}. This explains our notation $Z_{\widetilde{\text{AD}}}$, where we keep $Z_{\AD}$ for the partition function calculated from the AD curve. 

Notably, from \eqref{ZtildeADNf1K3} it follows that
\begin{equation}
    Z_{\widetilde{\text{AD}}}[\ttu^\ell]=0, \qquad \ell\geq 1,
\end{equation}
where $\ttu=u-\uad$ is the Coulomb branch parameter with the AD point at its origin. That is, the 0-observable $\ttu$ is a null vector of the $\widetilde{\text{AD}}$ theory, which is in agreement with the selection rule (see also \cite{Moore:2017cmm}). We will discuss the relation to  SW contributions of AD curves in a future work \cite{AFMM:future}.
\subsubsection*{$\boldsymbol{N_f=2}$}
We continue with the $N_f=2$ theory. For distinct masses, $\bfm=(m_1,m_2)$ with $m_1\neq m_2$, this theory has 4 strong coupling singularities of type $I_1$. We find for the partition function,
\be 
\label{NF=2equalmass_exp}
\begin{split}
Z_0(\bfm)[u^0]&= 32\Lambda_2^2\left(\frac{1}{(m_1-m_2)^2} + \frac{1}{(m_1+m_2)^2}\right)\\
&= \frac{32\Lambda_2^2}{(m_1-m_2)^2}+\frac{8\Lambda_2^2}{m_2^2}+\CO(m_1-m_2).
\end{split}
\ee 
The divergence is due to the monopole contributions of the singularities $u_3^*$ and $u_4^*$.  

The singular behaviour near the equal mass limit, $\bfm=(m,m)$, has an interesting interpretation in terms of vertex algebras \cite{Dedushenko:2017tdw}. The singularity is expected to be cancelled by a contribution from the Higgs branch (we will comment on this further below).

The equal mass limit is non-singular for the contributions from $u_1^*$ and $u_2^*$. For the partition function of the instanton component (\ref{Zinstbr}) with equal masses, we find
\be 
\label{K3instlargemassNf2}
\begin{split}
& Z^{\rm i}_0(\bfm)[u^0]=\frac{4\Lambda_2^5}{m^2}\left(\frac{1}{(\Lambda_2+2m)^3}+\frac{1}{(\Lambda_2-2m)^3} \right)\\
&\quad \qquad \quad \,\,\, = -\frac{3}{2}\frac{\Lambda_2^6}{m^6}-\frac{5}{4}\frac{\Lambda_2^8}{m^8}-\frac{21}{32}\frac{\Lambda_2^{10}}{m^{10}}-\frac{55}{512}\frac{\Lambda_2^{14}}{m^{14}}+\dots, \\
&Z^{\rm i}_0(\bfm)[u^1]= -\frac{\Lambda_2^6}{m^4}-\frac{21}{16}\frac{\Lambda_2^8}{m^6}-\frac{25}{32}\frac{\Lambda_2^{10}}{m^{8}}-\frac{91}{256}\frac{\Lambda_2^{12}}{m^{10}}-\frac{9}{64}\frac{\Lambda_2^{14}}{m^{12}}+\dots, \\
&Z^{\rm i}_0(\bfm)[u^3]= -\frac{5}{4}\frac{\Lambda_2^8}{m^4}-\frac{115}{128}\frac{\Lambda_2^{10}}{m^6}-\frac{113}{256}\frac{\Lambda_2^{12}}{m^{8}}-\frac{373}{2048}\frac{\Lambda_2^{14}}{m^{10}}+\dots.
\end{split}
\ee 
\vspace{.3cm}\\
{\it The equal mass case and multi-monopole SW-invariants}\\
Instead of taking the equal mass limit of $Z_0(\bfm)$, one can evaluate the partition function directly at $m_1=m_2=m$. While this gives a finite answer, the theory now involves a strong coupling cusp of type $I_2$. The $Q$-fixed equation is a multi-monopole equation, and will thus involve a generalization of the SW-invariant for $I_1$ singularities as discussed in Section \ref{sec:SWI2}. Even though the limit $m_1\to m_2$ is singular for the sum of the contributions from $u_3^*$ and $u_4^*$, one obtains a finite answer if one works directly in the equal mass theory $m_1=m_2=m$. This was demonstrated for $m=0$ in \cite{Kanno:1998qj}.

To this end, let us recall that for a non-vanishing contribution, the virtual dimension of the monopole moduli space (\ref{vdim_r}) should be non-negative, while at the same time $c_+^2=0$. These two conditions have no solutions for a generic K\"ahler point $J$ except for $c=0$, but can have a finite number of solutions for special choices of $J$. 

For $K3$, the requirement that (\ref{vdim_r}) is non-negative is $c^2\geq 16(\frac 1r-1)$. Together with the condition that $c_+=0$ and that $c$ is a $\spinc$ structure, this gives the following possibilities 
\begin{equation}\begin{aligned}
       r=1&: \quad c_-^2=0  \qquad &&n_1(c)=0\\
    r=2&: \quad c_-^2=0,-8,  \qquad &&n_2(c)=2,0\\
    r=3&: \quad c_-^2=0,-8, \qquad &&n_3(c)=4,1, \\
    r=4&: \quad c_-^2=0,-8, \qquad &&n_4(c)=6,2.
\end{aligned}
\end{equation}

The contributions from different $c$ can have different signs, such that the total contribution can be metric independent, even if individual contributions are metric dependent.

Returning to the case $r=2$, using the exact expressions in Appendix \ref{sec:equal_mass_expansion}, for the contribution from $c=0$ we obtain
\be \label{SW20}
-\SW_2(0) \frac{16 \Lambda_2^2}{\kappa_3} \frac{(57\Lambda_2^4-48\Lambda_2^2m^2+64m^4)}{(\Lambda_2^2-4m^2)^3},
\ee 
and for the contribution from $c^2=-8$
\be \label{SW208}
-\sum_{c,c^2=-8} \SW_2(c) \frac{2\Lambda_2^6}{\kappa_3(\Lambda_2^2-4m^2)^3}.
\ee 
In the massless limit, these contributions \eqref{SW20} and \eqref{SW208} match with \cite[Eq. (4.6)]{Kanno:1998qj} (with $c=2\lambda_{\text{KY}}$). The total result then becomes
\be 
\label{Z0mEq}
\begin{split} 
Z_0(m)[u^0]&=\frac{2\Lambda_2^6}{(\Lambda_2^2-4m^2)^3}\left(\frac{4\Lambda_2^2+48m^2}{m^2}\,\SW_1(0)- \sum_{c,c^2=-8} \frac{\SW_2(c)}{\kappa_3}\right.\\
&\quad -\left. 8(57\Lambda_2^4-48\Lambda_2^2m^2+64m^4)\, \frac{\SW_2(0)}{\kappa_3}\right),
\end{split}
\ee 
where we substituted the contribution from $u_1^*$ plus $u_2^*$ given in Eq. (\ref{K3instlargemassNf2}), and made the dependence on $\SW_1(0)$ explicit.

    While this is finite for generic $m$, we notice a divergence for the individual contributions as the mass approaches the AD mass, $m\to \mad= \Lambda_2/2$. On the other hand, the previous discussion and results for $N_f=1$ give an indication that the combined contribution of the three cusps may become a (Laurent)  polynomial in $m$. In addition to (\ref{Z0mEq}), there is a contribution from the non-compact Higgs branch. It is unclear however, how the Higgs branch dynamics could cancel the divergence for $m\to \mad= \Lambda_2/2$. Although heuristically, we are then led to the idea that the 2-monopole invariants ${\rm SW}_2$ satisfy special relations to ensure smoothness at $\mad= \Lambda_2/2$. This is similar to the constraints on the SW invariants from the AD mass locus in other theories \cite{Marino:1998uy}.

To deduce the relations for ${\rm SW}_2$, note that both contributions from $c=0$ \eqref{SW20} and $c^2=-8$ \eqref{SW208} have a cubic singularity at $m=\mad$. Correspondingly, their sum is regular in the AD limit if three linear combinations of the invariants $\SW_2(c)$ vanish. Remarkably, there is  a linear dependence among these three combinations, and the AD limit is regular if\footnote{If we eliminate $\kappa_3$, we find $\sum_{c:c^2=-8}\SW_2(c)+264 \SW_2(0) =0$. This condition alone is not sufficient for the regularity of the limit.}
\begin{equation}
    \begin{aligned}
    \sum_{c:c^2=-8}\SW_2(c)+392\,\SW_2(0)-64\kappa_3\,\SW_1(0)&=0, \\
\sum_{c:c^2=-8}\SW_2(c)+520\,\SW_2(0)-128\kappa_3\,\SW_1(0)&=0.
    \end{aligned}
\end{equation}
These vanishing combinations are reminiscent of the sum rules for the type $II$ AD theory \cite{Marino:1998uy}. 
The unique solution to these equations, setting $\SW_1(0)=1$, is 
\be
\label{SWhyp}
\frac{\SW_2(0)}{\kappa_3}=\frac{1}{2},\qquad \frac{\sum_{c,\, c^2=-8} \SW_2(c)}{\kappa_3}=-132.
\ee 
We list correlation functions $\left< u^\ell \right>$ for small $\ell$ in Table \ref{PtObsNf2K3} for this solution.
Since the correlation functions are polynomials in $m$, the AD limit $m\to\mad$ is smooth. More generally, we expect that the multi-monopole invariants $\SW_{r>1}$ can be deduced in this way from the $I_1$ SW invariants. 

We do stress though that the analysis of the multi-monopole suggests that ${\rm SW}_2(c)=0$, $c\neq 0$, for generic metric since there are no solutions to $c_+=0$. It is thus an interesting question whether (\ref{SWhyp}) can be derived starting from the multi-monopole equations.
\begin{center}
\begin{table}[ht!]
\centering
\renewcommand{\arraystretch}{2}
\begin{tabular}{|>{$\displaystyle}Sl<{$}|  >{$\displaystyle}Sr<{$}|  }
\hline
\ell &  Z_{0}[u^\ell] \text{ for } N_f=2\\ 
\hline 
0 & \frac{8 \Lambda_2^2}{m^2} \\
 1 & 88 \Lambda_2^2-\frac{\Lambda_2^4}{m^2} \\
 2 & 2 \Lambda_2^4+\frac{\Lambda_2^6}{8 m^2}+232 \Lambda_2^2 m^2 \\
 3 & -\frac{27 \Lambda_2^6}{8}+440 \Lambda_2^2 m^4-\frac{\Lambda_2^8}{64 m^2}+33 \Lambda_2^4 m^2 \\
 4 & -\frac{7 \Lambda_2^8}{16}+712 \Lambda_2^2 m^6+116 \Lambda_2^4 m^4+\frac{\Lambda_2^{10}}{512 m^2}-\frac{23 \Lambda_2^6 m^2}{4} \\
 5 & -\frac{161 \Lambda_2^{10}}{512}+1048 \Lambda_2^2 m^8+275 \Lambda_2^4 m^6+\frac{11 \Lambda_2^6 m^4}{4}-\frac{\Lambda_2^{12}}{4096 m^2}-\frac{115
   \Lambda_2^8 m^2}{32} \\
\hline 
\end{tabular}
\caption{Table with values of $Z_0[u^\ell]$ for $N_f=2$ with equal mass $m$ and $\SW_2$ invariants as in \eqref{SWhyp}.}
\label{PtObsNf2K3}
\end{table} 
\end{center}
\vspace{.3cm}
{\it Comparing different limits} \\
In \cite{Marino:1998uy,Marino:1998uy_short}, the collision of two $I_1$ singularities to a type $II$ AD point has been considered, which is the only existing limit in $N_f=1$. For $N_f\geq 2$ on the other hand, higher type AD points appear, which allow for a larger variety of possible collisions. If we have two masses for instance, we can form a type $III$ AD point in three possible ways (see Fig. \ref{fig:Nf2ADlimits}). Carefully calculating the three possible limits of partition functions should allow us to study the precise form of the contribution from the Higgs branch (see also \cite{Moore:1997dj,losev1998,Dedushenko:2017tdw,Nakatsu:1997jw,Martens:2006hu,LoNeSha,Antoniadis:1996ra}).
\begin{figure}[ht]\centering
	\includegraphics[scale=0.7]{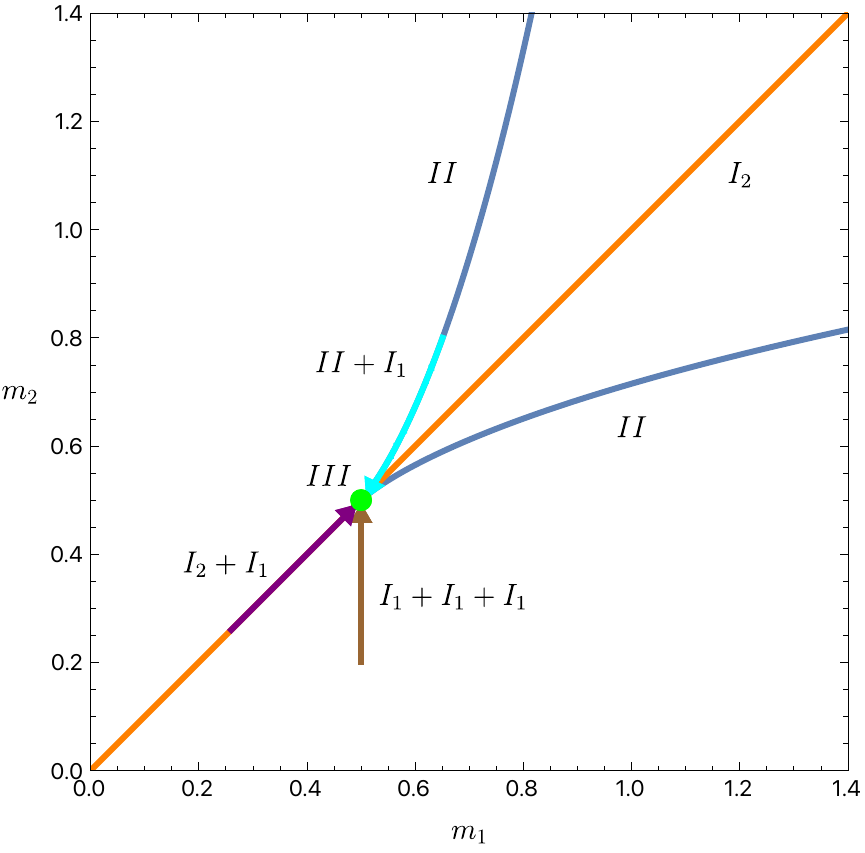}
	\caption{Plot of the mass space $(m_1,m_2)$ in $N_f=2$. The AD locus (blue) consists of two contours giving rise to a Coulomb branch (CB) with each one $II$ AD point. When the loci meet, the CB contains a type $III$ point (green). This point can be approached from three distinct configurations: Away from $m_1=m_2$, the limit is $I_1+I_1+I_1\to III$ (brown). On the line $m_1=m_2$ (orange), there is a $I_2$ singularity, with a limit $I_2+I_1\to III$ (purple). Finally, on the generic AD locus there is a $II$ singular point with a limit $II+I_1\to III$ (cyan). On any point on the $I_2$ line $m_1=m_2$, a Higgs branch with geometry $\mathbb C^2/\mathbb Z_2$ meets the Coulomb branch.}\label{fig:Nf2ADlimits}
\end{figure}
We can now compare the calculation of the SW contribution from the type $III$ Argyres-Douglas point, as illustrated in Fig. \ref{fig:Nf2ADlimits}. As required, the limit $m\to \mad$ for the $I_2+I_1$ contribution is smooth, and from Table \ref{PtObsNf2K3} we find that it equals 32. 
Calculating the limit of $I_1+I_1+I_1$ from any direction away from $m_1=m_2$, we can  use the expansion \eqref{NF=2equalmass_exp} 
for $N_f=2$, where we expand all singularities around $m_1=m_2$. This expansion agrees with setting say $m_1=\mad=\frac12\Lambda_2$ and considering the limit $m_2\to  \mad$, as in Fig. \ref{fig:Nf2ADlimits}. In this limit, the constant part in \eqref{NF=2equalmass_exp} evaluates to 32 as well, such that the two limits $I_1+I_1+I_1\to III$ and $I_2+I_1\to III$ agree precisely, up to the divergent term $32\Lambda_2^2/(m_1-m_2)^2$.  Thus the $\SW_2$ invariant for the $I_2$ contribution naturally regularises the singular limit of colliding two $I_1$ singularities. This is another nontrivial constraint on the relations \eqref{SWhyp}.

\subsubsection*{$\boldsymbol{N_f=3}$}
Another configuration suitable for studying $I_2$ singularities is $\bfm=(m,0,0)$ in $N_f=3$, which for generic mass $m$ is $(I_1^*,2I_2,I_1)$. See for example \cite{Aspman:2021vhs} for more details. The two $I_2$ singularities $u_\pm=\frac18 m\Lambda_3$ become the $I_4$ singularity in the massless limit $m\to 0$, while the $I_1$ singularity $m_*=m^2+\frac{1}{2^8}\Lambda_3^2$ collides with either $u_+$ or $u_-$ for $m\to \pm \mad\coloneqq \frac{1}{16}\Lambda_3$. 
We can calculate the sum over the two $I_2$ SW contributions and the $I_1$ SW contribution. For this, we solve the sextic equation associated with the SW curve \cite{Aspman:2021vhs} and calculate the local expansions $a$, $\frac{du}{da}$, $\frac{da}{d\tau}$ and $\Delta_3$ near the cusps $u_\pm$ using either of the two respective solutions near the cusps. We find again that the AD limit $m\to \mad$ is regular if we impose sum rules on the SW invariants. With the expressions \eqref{ABE3}, we find that the AD limit is regular if and only if \eqref{SWhyp} holds.
This may be viewed as a consequence of the fact that the SW curves for $\bfm=(m,m)$ and $\bfm=(m,0,0)$ are isogenous \cite{ACFMM23}.

We list some values of $Z_0[u^\ell]$ in Table \ref{PtObsNf3K3}. The order 6 pole comes from the mass singularity, since for $N_f=3$ we have $\left((m_1-m_2)(m_2-m_3)(m_3-m_1) \right)^{\chi_h}\sim m^6$ for all masses vanishing. One qualitative difference to $N_f=2$ is that the mass divergence for $m\to 0$ disappears for larger $\ell$, i.e. $Z_0[u^\ell]$ is regular as $m\to 0$ for $\ell\geq 5$. For the $I_2$ singularity in $N_f=2$, the leading term is always $m^{-2}+$ regular, for any $\ell$.
\begin{center}
\begin{table}[ht!]
\centering
\renewcommand{\arraystretch}{2}
\begin{tabular}{|Sl|Sr|}
\hline 
$\ell$ & $\displaystyle Z_{0}[u^\ell]$ for $N_f=3$\\ 
\hline 
0 & $\displaystyle 320\frac{\Lambda_3^6}{m^6}$\\ 
1 & $\displaystyle 192 \frac{\Lambda_3^6}{m^4}$ \\
2 & $\displaystyle 64\frac{\Lambda_3^6}{m^2}-\frac{\Lambda_3^8}{m^4}$\\
3 & $\displaystyle 64\Lambda_3^6-3 \frac{\Lambda_3^8}{m^2}$ \\
4 & $\displaystyle 64m^2\Lambda_3^6+\frac74 \Lambda_3^8+\frac{1}{64}\frac{\Lambda_3^{10}}{m^2}$ \\
5& $\displaystyle 64m^4\Lambda_3^6+2m^2\Lambda_3^8-\frac{401}{2^{10}}\Lambda_3^{10}$ \\
\hline 
\end{tabular}
\caption{Table with values of $Z_0[u^\ell]$ for $N_f=3$ with $\bfm=(m,0,0)$ and $\SW_2$-invariants as in \eqref{SWhyp}.}
\label{PtObsNf3K3}
\end{table} 
\end{center}
\subsection{Relation to results for algebraic surfaces and Segre numbers}\label{sec:segrenumbers}

The coefficients in the mass expansions such as \eqref{K3instlargemass} and \eqref{K3instlargemassNf2} are known in the mathematical literature  as Segre numbers \cite{Gottsche:2021ihz}. Yet, in the UV formulation (see Section \ref{corrfunctions}),  correlation functions are intersection numbers obtained by integrating \emph{Chern} classes on the moduli space. 
In this Section, we aim to explain the connection between Chern classes and Segre classes, and provide a relation between SW partition functions and these geometric invariants.

\subsubsection*{Brief review of Segre classes of moduli spaces}

The (total) Segre class $s(E)$ of a vector bundle is a characteristic class, which is suitable for the analysis of intersection theory in singular settings.\footnote{We refer the reader to \cite{fulton1984} for  definitions, in particular the Sections 3.1 and 3.2.} In this context, the more familiar Chern classes are defined in terms of the Segre classes. To this end, one defines for a vector bundle $E$ the following formal power series,
\be
\begin{split} 
s_t(E)=\sum_{j=0}^\infty s_j(E)\,t^j,\\
\end{split} 
\ee
in terms of the Segre classes $s_j(E)$. We refer for their definition to \cite[Section 3.1]{fulton1984}. The Chern classes $c_j(E)$ are then defined by the inverse of $s_t(E)$,
\be
c_t(E)=\sum_{j=0}^\infty c_j(E)\,t^j=s_t(E)^{-1}.
\ee 
It can be shown that $c_t(E)$ is a polynomial whose degree is bounded by ${\rm rk}(E)$.

Segre classes first appeared in the context of moduli of vector bundles in an article by Tyurin \cite{Tyurin1994}.
Lehn \cite{Lehn:1999} put forward an algorithmic description for the evaluation of top Segre classes for Hilbert schemes of points, which corresponds physically to $U(1)$ gauge theory. His conjecture was recently proven in \cite{Marian:2015}. Reference \cite{gottsche2022} gives proofs for various conjectures relating Segre and Verlinde numbers of Hilbert schemes. Segre classes were introduced in \cite{Gottsche:2020ass} for higher rank bundles over projective surfaces. Reference \cite{Oberdieck_2022} gives proofs for these higher rank conjectures for $K3$ surfaces as well as for the Segre-Verlinde correspondence.  Segre classes have also appeared in a series of work by Feehan and Leness on SW invariants for specific monopole moduli spaces \cite{Feehan:1997vp,Feehan:1997rt,Feehan:2001jc, feehan2001pu}.  See also \cite{Gottsche:2019meh} for a recent survey on recent conjectures on the various correspondences.

\subsubsection*{Comparison between Segre numbers and physical correlation functions}

In order to demonstrate the correspondence between Segre numbers and  physical partition functions, let us first give the mathematical presentation. For simplicity, we fix the gauge group to be SU(2), and will not consider surface observables. Then $\rho=2$, $s=N_f$ and $L=0$ in \cite{Gottsche:2021ihz,Gottsche:2020ass}). Conjecture 2.8 in \cite{Gottsche:2020ass} reads in terms of the universal functions $V,W,\dots$,
\be 
\label{GKSegre}
\begin{split} 
&\sum_k z^{\frac{1}{2}\dim(\CM^{\rm i}_k)} \int_{\CM^{\rm i}_k} c(\alpha_\CM)\,\mu_D(p)^\ell \\
&\qquad = 2^{2-\chi_{\rm h}+\lambda} V^{c_2(\alpha)} \,W^{c_1(\alpha)^2}\,X^{\chi_{\rm h}}\,Y^{c_1(\alpha)\cdot K_X}\,Z^{K_X^2}\,T^\ell\\
& \qquad \quad \times  \sum_c \varepsilon^{a\cdot \bfmu}\, {\rm SW}(K_X-2a)\,Y_1^{c_1(\alpha)\cdot a}\,Z_1^{a^2},
\end{split}
\ee 
where $\alpha_\CM$ is related to $\alpha$ using the universal bundle, and we used that the SW basic class $c$ is related to $a$ by $c=K_X-2a$.
The functions $V, W, X, T$ are explicitly determined in \cite{Gottsche:2021ihz,Gottsche:2020ass} for arbitrary rank. This leads typically to an infinite series in the parameter $z$, which suggests that $\alpha_\CM$ is a sheaf rather than a vector bundle. 

Clearly, the structure of (\ref{GKSegre}) is very similar to that of (\ref{SWuniversal}), and we can relate the  functions $\CA_1,\CB_1,\dots$ to combinations of the functions $V,W,\dots$.

We discuss in the following the agreement of some of the functions in both the physical and mathematical approach.

\subsubsection*{$\boldsymbol{N_f=1}$}
Let us begin with the case of only one mass.  The expansion in $z$ is expressed by introducing an auxiliary variable $t$ through $z=t(1+t/2)^{1/2}$. This relation can be  inverted, which we denote by $t(z)$.\footnote{\label{ftnt_zt}The square of the  relation $z=t(1+t/2)^{1/2}$ is cubic in $t$, such that we find the solution \begin{equation}\label{t(z)solution}
    t(z)=\frac{1}{3} \left(\sqrt[3]{27 z^2+3 \sqrt{3} \sqrt{27 z^4-16 z^2}-8}+\frac{4}{\sqrt[3]{27 z^2+3 \sqrt{3} \sqrt{27 z^4-16 z^2}-8}}-2\right),
\end{equation}
which can be expanded either at $z=0$ or $z=\infty$. 

For instance, we have $t(z)=z+\CO(z^2)$ as $z\to 0$, and $t(z)=2^{\frac13}z^{\frac23}+\CO(z^{-\frac23})$ as $z\to\infty$. In calculations, we must expand $t(z)$ in $z$ first  in order  to obtain the correct phases. }
The universal functions then read as function of $t$ as \cite{Gottsche:2020ass},
\be \label{universal_functions}
\begin{split} 
W(t)&=(1+\tfrac{1}{2}t)^{-1},\\
X(t)&=(1+\tfrac{1}{2}t)^{-3/4}\,(1+\tfrac{3}{4}t)^{-1/2},\\
T(t)&=2t(1+\tfrac{3}{8}t).
\end{split}
\ee 
By comparing (\ref{SWuniversal}) and (\ref{GKSegre}), we deduce that $X$ is proportional to $\CA_1^{12}/\CB_1^8$. Using (\ref{CABE}), we then find the relation
\be
 X\left(t(z)\right)^2=\frac{48\,m^{5/2}}{\Lambda_1^{9/2}} \frac{g_3(u_1^*)}{P^{\text M}_1(u_1^*)^2}.
\ee 
with the identification
\be\label{z_variable_nf1}
z=\frac12\left(\frac{\Lambda_1}{m}\right)^{\frac 32}=\frac{\Lambda_0^2}{2m^2}.
\ee 

Comparison of (\ref{SWuniversal}) and (\ref{GKSegre}) also identifies $T$ and $u_1^*$. Indeed with the relation (\ref{z_variable_nf1}), we arrive at
\be 
T(t(z))= -m^{-2}\, u_1^*. 
\ee 
This equivalence is due to the identity
\begin{equation}\begin{aligned}\label{segre_relations_Nf1}
\Delta_1(-m^2T(t(z)))&=0, \\
\end{aligned}\end{equation}
where $\Delta_1$ is the discriminant of $N_f=1$. These identities are consequences of the SW geometry alone.
For instance, we can prove the first identity using the definition \eqref{universal_functions} and reducing the sextic polynomial $\Delta_1(-m^2T(t(z)))$ in $t(z)$ using $t^3+2t^2-2z^2=0$. 

For $N_f=1$, there is a $\mathbb Z_3$ symmetry \eqref{Nf1_Z3sym} that relates all three singularities $u_j^*$ under $\mathbb Z_3$ rotations of $\Lambda_1$. This would suggest that if one singularity $u_1^*$ is related to universal functions generating the Segre invariants (as in \eqref{segre_relations_Nf1}), then it must be true also for the other singularities. This is however not the case, as can be seen by expanding $t(z)$ as a series in $z$, which is a (regular) Taylor series at $z=0$ (see footnote \ref{ftnt_zt}). If we rotate $\Lambda_1\mapsto \zeta_3\Lambda_1$, the variable $z$ \eqref{z_variable_nf1} is mapped to $-z$. Thus under the change of variables $t\leftrightarrow z$ relating the SW invariants to Segre invariants, the $\mathbb Z_3$ symmetry collapses to a $\mathbb Z_2$ symmetry, relating the contributions $u_1^*$, $u_2^*$ constituting the instanton component. We find, 
\begin{equation}
    u_2^*(\Lambda_1)=u_1^*(\zeta_3\Lambda_1)=-m^2T(t(-z)),
\end{equation}
while $u_3^*$ can not be expressed through $T$. 

If we expand at $m=0$ rather, we have the full $\mathbb Z_3$ symmetry $u_j^*(\zeta_3\Lambda_1)=u^*_{j+1}(\Lambda_1)$ with $j\mod 3$. Expanding $t(z)$ \eqref{t(z)solution} for large $z$ gives a Laurent series in $z^{\frac23}$ at $z=\infty$. Rotating $\Lambda_1\mapsto \zeta_3\Lambda_1$ gives $z\mapsto e^{\pi i}z$ and thus three different solutions $t(z)$, and therefore 
\begin{equation}\begin{aligned}\label{u1,2,3T}
    u_1^*&=-m^2 T(t(z)), \\
    u_2^*&=-m^2 T(t(e^{\pi i}z)), \\
    u_3^*&=-m^2 T(t(e^{-\pi i}z)).
\end{aligned}
\end{equation}
For large masses, the singularity $u_3^*$ corresponds to the hypermultiplet that decouples, and thus is associated with the monopole component. For small masses on the other hand, the three singularities $u_{1,2,3}^*$ are indistinguishable. This can also be seen from the fact that $u_2^*$ and $u_3^*$ merge as we increase the mass  from $m=0$ to $\mad=\frac34\Lambda_1$. Thus the labels of the singularities $u_j^*$ are not meaningful throughout the whole parameter space of the masses. As a consequence, it is not possible to attribute either $u_{1,2,3}^*$ to the monopole component, which enables the identification \eqref{u1,2,3T}. 

To include a background flux for $N_f=1$, we identify $c_1(\alpha)=\bfk_1$ and $c_2(\alpha)=0$. Comparison of (\ref{SWuniversal}) and (\ref{GKSegre}) relates $W$ and $\CC_1$. Indeed, using the large mass expansion (\ref{Cseries}) in terms of modular forms, we can compare the first terms near $u_1^*$, that is, in the limit $\tau_1=-1/\tau\to i\infty$. This matches indeed with the relation
\be 
W=\frac{m}{\Lambda_1}\,\CC_1.
\ee 

Returning to the specific $K3$ geometry, we have $\chi_h=2$, while $K_X=0$. The odd powers of $z$ of (\ref{GKSegre}) agree with the large mass expansions in Eq. (\ref{K3instlargemass}) up to an overall power in $z$. In fact, we can express the contribution from the $u_1^*$ singularity for $K3$ in terms of $X$, $W$ and $T$,
\begin{equation}\label{Z1_segre}
    Z_{0,1}[u_1^{*\ell}]=\frac12(-1)^\ell m^{2\ell}\left(\frac{\Lambda_1}{m}\right)^{\frac 52+\bfk_1^2}X\left(t(z)\right)^2\,W(t(z))^{\bfk_1^2}\,T\left(t(z)\right)^\ell.
\end{equation}

\subsubsection*{$\boldsymbol{N_f=2}$}

Using the equal mass $m_1=m_2=m$ theory for $N_f=2$, we can also check the matching with Segre invariants. For simplicity, we consider the theory only in absence of background fluxes. 

The auxiliary variable $t$ is equal to $z$ in this case. With $s=N_f=2$, one finds for $X$ from \cite{Gottsche:2020ass},
\be
\begin{split}
X(z)&=(1+z)^{-3/2}.
\end{split}
\ee 
Then with the identification
\be
z=\frac{\Lambda_2}{2m}=\frac{\Lambda_0^2}{2m^2},
\ee 
one find the same expansion as for the contribution of the monopole cusp $u_1^*$. On the other hand, we find a slightly different value  when point observables are inserted, $T(z)=2z$, whereas $u_1^*=-\Lambda_2 m -\frac{\Lambda_2^2}{8}$. Of course, both insertions capture the same information.

\subsection{Generalised  simple type}\label{gen_sim_type}
To conclude this Section on SW partition functions for SU(2) SQCD, we study a generalisation of the simple type condition familiar from the pure SU(2) theory.

Let us define the generating functional of Seiberg-Witten invariants ($b_2^+>1$)
\begin{equation}\label{ZSW_generating}
    Z_{\bfmu}^{N_f}[e^{2p u/\Lambda_{N_f}^2}]\coloneqq \sum_{\ell=0}^\infty \frac{1}{\ell!}\left(\tfrac{2p}{\Lambda_{N_f}^2} \right)^\ell Z_\bfmu^{N_f}\left[u^\ell \right],
\end{equation}
where $Z_\bfmu$ is the sum over the contributions of all the  $N_f+2$ strong coupling cusps. In the following, we drop labels such as `SW' and $\bfmu$ to avoid cluttering the notation, and denote the functions by $Z^{N_f}$.

General structure results about such generating functionals are known since the work of Kronheimer and Mrowka \cite{kronheimer1994recurrence,Kronheimer1995, Fintushel1995} (see also \cite{Donaldson1996,Morgan1995}). For instance, a four-manifold is of  \emph{generalised simple type} if there exists an integer $n\geq 0$ such that the generating function satisfies
\begin{equation}\label{wfinite}
   \left(\frac{\partial^2}{\partial p^2}-4\right)^n Z^{0}(p)=0,
\end{equation}
where we focus on the dependence of $Z^{N_f}(p)\coloneqq Z^{N_f}[e^{2p u/\Lambda_{N_f}^2}]$ on $p$ only.
The \emph{order} is the minimum of all $n$ such that \eqref{wfinite} holds. It has been shown that if $b_2^+(X)>1$, then the order of generalised simple type manifolds does not depend on the choice of 't Hooft flux $\bfmu$, that is, it only depends on the four-manifold. In fact, all four-manifolds with $b_2^+>1$ are of finite type. 
A manifold is said to be of Donaldson simple type if it is of simple type with $n=1$. While all four-manifolds with $b_2^+>1$ are of generalised simple type, it is not known if there are simply-connected manifolds with $b_2^+>1$ which are \emph{not} of Donaldson simple type. Donaldson invariants for manifolds which are not of simple type have been studied in \cite{Kronheimer1997,munoz1998basic,munoz2002donaldson}.

For instance, for $X$ a $K3$ surface the result is 
\begin{equation}\label{K3_ODE}
   \left(\frac{\partial^2}{\partial p^2}-4\right)Z^0(p)=0.
\end{equation}
In the pure SU(2) case, acting with the operator $\partial_p^2-4$ on correlation functions is analogous to an insertion of the discriminant $\Delta$ into the $u$-plane integral \cite{Moore:1997pc}. Since at the monopole and dyon cusps, $\Delta_D(\tau)= q+\CO(q^2)$, this increases the overall $q$-exponent of the integrand by one. For arbitrary signature, inserting $\Delta^n$ with large enough $n$ will annihilate the partition function. 

In \cite{Labastida:1995gp,Kanno:1998qj,Labastida:1998hf} it has been observed that  the SW contribution in massless SQCD satisfies a similar, but higher order  differential equation. In the massless $N_f=1$ theory, the generating functional satisfies \cite{Kanno:1998qj}
\begin{equation}\label{ZSW_Nf1_ODE}
   \left(\frac{\partial^3}{\partial p^3}+\frac{27}{32}\right) Z^1(p)=0, 
\end{equation}
which can be directly confirmed from the sum over \eqref{Z1K3}. This differential equation has a three-dimensional solution space which is spanned by the functions $e^{2p\ttu_j^*}$, where $\Delta_1(\ttu_j)=0$ are the three roots of the massless $N_f=1$ discriminant, and $\ttu\coloneqq u/\Lambda_1^2$. The three coordinates can in this case be easily determined, and are just proportional to $\ttu_j^*$. This gives \cite{Kanno:1998qj}
\begin{equation}\label{Nf1_K3_resum}
    Z^1(p)=-\frac{2^6}{3^2}\sum_{j=1}^3 \ttu_j^* e^{2 p\, \ttu_j^*}.
\end{equation}
In the massive case, to the best of our knowledge such a result is not known in the literature. By collecting the results for massive $N_f=1$ on $K3$ (such as in Table \ref{PtObsNf1K3}) into a generating function, we obtain the relation
\begin{equation}\label{ZSW_ODE_massiveNf1}
     \left(\frac{\partial^3}{\partial p^3}-2\mu^2\frac{\partial^2}{\partial p^2}-\frac92\mu\frac{\partial}{\partial p}+8\mu^3+\frac{27}{32}\right) Z^1(p)=0, 
\end{equation}
where $\mu=m/\Lambda_1$ is the dimensionless mass.
Solutions to this equation take a similar form to \eqref{Nf1_K3_resum}, with coefficients now more complicated functions of the dimensionless singular points $\ttu_j^*$.

\subsubsection*{$\boldsymbol{I_1}$ singularities}
We can in fact easily find the simple type condition for SQCD with $N_f$ generic masses. 
If we write the contribution from each cusp as (we again drop some dependence in the notation) 
\begin{equation}
    Z^{N_f}[u^\ell]=\sum_{j=1}^{2+N_f}Z_j^{N_f}(u_j^*) (u_j^*)^\ell,
\end{equation}
we resum it in \eqref{ZSW_generating} to find 
\begin{equation}\label{ZSW_sum_j}
    Z^{N_f}(p)=\sum_{j=1}^{2+N_f}Z_j^{N_f}(\tfrac{u_j^*}{\Lambda_{N_f}^2})e^{2p u_j^*/\Lambda_{N_f}^2}.
\end{equation}
Let us consider the function $f_j(p)=e^{p \, \ttu_j^*}$, where the $\ttu_j^*$ are roots of some polynomial, $\Delta_{N_f}(\ttu_j^*)=0$ (we consider $\Delta_{N_f}$ as a polynomial in $\ttu_j=u/\Lambda_{N_f}^2$). Then these $2+N_f$ functions $f_j$ form the basis of the $2+N_f$-dimensional $\mathbb C$-vector space of solutions to the $2+N_f$-order linear ODE $\Delta_{N_f}(\partial_p)f(p)=0$. This is clear from the fact that 
\begin{equation}
    \Delta_{N_f}(\partial_p)f_j(p)=\Delta(\ttu_j^*)f_j(p)=0.
\end{equation}
Indeed, $\Delta_{N_f}$ is the characteristic polynomial of this differential equation. 
In order to apply this to \eqref{ZSW_sum_j}, we need to multiply $p$ again by 2, due to conventions. It follows that 
\begin{equation}\label{SWST_Nf}\boxed{
    \Delta_{N_f}(\tfrac12\partial_p)Z^{N_f}(p)=0.}
\end{equation}
This relation includes \eqref{K3_ODE}, \eqref{ZSW_Nf1_ODE} and \eqref{ZSW_ODE_massiveNf1} as examples, and it is expected to hold for any configuration only involving $I_1$ singularities. 
Similar to the case $N_f=0$, we can understand \eqref{SWST_Nf} as inserting the discriminant $\Delta_{N_f}$ into the $u$-plane integral. Since $\Delta_{N_f}^D(\tau)=q+\dots$ for any cusp, this increases the exponent of the integrand $q$-series by one, and will thus eliminate the contribution from the strong coupling cusps to the partition function if applied sufficiently many times.

\subsubsection*{$\boldsymbol{I_2}$ singularities}
From \eqref{SWST_Nf} it is possible to extract the behaviour of the generating function $Z^{N_f}$ in the case where some of the singularities collide. Consider for example the case where all $I_1$ singularities collide in an $I_k$ singularity (this is not possible for the SW curves \eqref{eq:curves},  we are merely illustrating the structure). When zeros of the characteristic polynomials merge, the basis of the space of solutions to the ODE seemingly collapses. However, the dimension of the solution space is independent of the coefficients of the characteristic polynomial, as it only depends on its degree. Let for instance $\Delta(u)=(u-u_0)^k$ and we consider the equation $\Delta(\partial_p)f(p)=0$. As an elementary calculation shows\footnote{Let $f_0(p)=P(p)e^{p u_0}$, then $(\partial_p-u_0)f_0(p)=P'(p)e^{pu_0}$ and thus $0=\Delta(\partial_p)f_0(p)=P^{(k)}(p)e^{pu_0}$. Therefore, $P$ is a polynomial of degree $k-1$, and thus any solution to $0=\Delta(\partial_p)f(p)$ is a linear combination of the functions $f_0(p)=e^{pu_0}\sum_{n=0}^{k-1}a_np^n$.}, a basis of solutions is given by the functions $f(p)=P_{k-1}(p)e^{pu_0}$, with $P_{k-1}$ a polynomial of degree $k-1$. 
This structure agrees with the results (4.6) and (4.10) in \cite{Kanno:1998qj}. For the higher $I_k$ singularities, there are contributions from subleading terms, and as mentioned above, only applies to $\Delta_{N_f}^n$ for some power $n$. This follows necessarily also from the structure of the ODE, since the generating function  \eqref{ZSW_sum_j} is never of the form $f(p)=P_{k-1}(p)e^{pu_0}$ unless $k=1$.

For $I_n$ with $n>1$, in particular massless $N_f=2$ and $3$ where $n=2$ and $n=4$ occur, the degree of the differential operator is known to in general exceed the degree of the discriminant. For $X=K3$, it has been shown that in the massless case $D_{N_f}(\tfrac12\partial_p)Z^{N_f}(p)=0$, where $D_{2}$ is a polynomial factor of $\Delta_2^2$, and $D_3$ is a  factor of $\Delta_3^7$ \cite{Kanno:1998qj}.

We can confirm this by an explicit calculation. Let $m_1=m_2=m$ in $N_f=2$, then there are two $I_1$ singularities $u_\pm$ and an $I_2$ singularity $u_*$, in the same notation as  \cite{Aspman:2021vhs}. The physical discriminant accordingly reads $\Delta_2(u)=(u-u_*)^2(u-u_+)(u-u_-)$. 
Using the results from Section \ref{sec:SWI2} (in particular, resumming the values in Table \ref{PtObsNf2K3}), we can show that in this case the multiplicity of the $I_2$ singularity increases to $3$ in the differential equation, 
\begin{equation}
    (\tfrac12\partial_p-u_*)^3(\tfrac12\partial_p-u_+)(\tfrac12\partial_p-u_-)Z^{2}(p)=0.
\end{equation}
In the massless limit $m\to 0$, the singularities $u_\pm$ merge to another $I_2$ cusp $u_0$. Calculating the $I_2$ correlation functions for both $u_*$ and $u_0$, we find that 
\begin{equation}
    (\tfrac12\partial_p-u_*)^3(\tfrac12\partial_p-u_0)^3Z^{2}(p)=0.
\end{equation}
This matches precisely with \cite[(4.8)]{Kanno:1998qj}.
Thus whenever two $I_1$ singularities in the $N_f=2$ theory collide, in the ODE the linear factors are enhanced to cubic factors.

We can get another similar result for the massive configuration $\bfm=(m,0,0)$ in $N_f=3$, which has discriminant $\Delta_3(u)=(u-u_+)^2(u-u_-)^2(u-u_*)$. The point correlators are given explicitly in  Table \ref{PtObsNf3K3}. Resumming all point observable correlators, we find that the generating functions satisfies 
\begin{equation}
    (\tfrac12 \partial_p-u_+)^3(\tfrac12 \partial_p-u_-)^3(\tfrac12 \partial_p-u_*)Z^{3}(p)=0.
\end{equation}

\subsubsection*{$\boldsymbol{I_n}$ singularities}

Higher $I_n$ contributions have also been studied in \cite{Kanno:1998qj}. For the $I_4$ cusp of massless $N_f=3$, the multiplicity of the corresponding factor in the generalised simple type condition is $7$. This motivates the following conjecture: Let $\bfm\in \mathbb C^{N_f}$ be a mass configuration giving rise to a spectrum $\bfn(\bfm)=(n_1,n_2,\dots)$ of massless hypermultiplets at each singularity, i.e. the physical discriminant is $\Delta_{N_f}(u)=\prod_j (u-u_j^*)^{n_j}$. Then the generating function of SW invariants satisfies
\begin{equation}\label{generating_function_SW_ODE}
    \prod_j(\tfrac12 \partial_p-u_j^*)^{2n_j-1}Z^{N_f}(p)=0.
\end{equation}
This in particular implies that $\Delta_{N_f}(\tfrac12\partial_p)^7Z^{N_f}(p)=0$ for all $N_f=0,\dots, 4$ and all mass configurations, since $n_j\leq 4$ for $\CN=2$ SU(2) SQCD \cite{Aspman:2021evt,Aspman:2021vhs}.

The reason for the enhanced exponents is the fact that for $n_j\geq 1$, the residues receive contributions from subleading terms in the expansion. For $n_j=2$ in $N_f=2$ in particular, next-to-next-to-leading order terms of all quantities involved are required for the calculation. Analogous to the above comment, we can thus always express the generating function as 
\begin{equation}
    Z^{N_f}(p)=\sum_{j}P_{2n_j-2}(p)e^{2p u_j^*/\Lambda_{N_f}^2},
\end{equation}
where $u_j^*$ are the singularities with multiplicities $n_j$, and $P_{2n_j-2}(p)$ are polynomials of degree $2n_j-2$ in $p$ with coefficients depending on the masses $\bfm$.

\section{Contribution from AD points}\label{sec:ADcontribution}

As discussed in part I, the superconformal Argyres-Douglas theories present themselves in the fundamental domain as elliptic points and can contribute to $u$-plane integrals. Partition functions of Argyres-Douglas theories have been studied in various contexts, such as in the $\Omega$-background \cite{Nishinaka:2012kn,Kimura:2020krd,Fucito:2023txg,Fucito:2023plp} and in topological theories
\cite{Marino:1998uy,Marino:1998uy_short,Gukov:2017,Moore:2017cmm,Dedushenko:2018bpp,moore_talk2018,Marino:1998bm}. 

As demonstrated in e.g. 
\cite{Moore:1997pc, Moore:2017cmm, Manschot:2021qqe}, the contribution from `interior' points of the $u$-plane, i.e. AD points or the UV point in $N_f=4$ or $\CN=2^*$, exhibits \emph{continuous} metric dependence rather than discrete wall-crossing.\footnote{See for instance  \cite[Equation (7.2)]{Moore:2017cmm} for an example of an AD theory, and \cite[Section 11.3]{Moore:1997pc} for the $N_f=4$ and $\CN=2^*$ theory} Unlike the evaluation of the contribution from the cusps, the contribution from the AD points for a specific choice of period point $J$ of the metric will then not be  constant in the whole chamber of $J$. Rather, under small deformations of $J$ it is expected to vary continuously. Note that this can only occur for $b_2>1$, since for $b_2=1$ and $\mathbb P^2$ in particular, the period point is unique. 

In this section, we focus on the asymptotic behaviour of the $u$-plane integrand and derive conditions for the integral to be sufficiently singular near the AD point to pick up a contribution.  We find that for a large class of manifolds, the contribution from any AD point in any theory vanishes. 

While the following discussion can be applied to a large extent to general SW curves, we want to stress that in this part II we calculate only the contributions from the Coulomb branch of SU(2) SQCD, including the contribution from the AD points to SQCD partition functions. As discussed before in Section \ref{sec:Contributions}, this can differ from the Argyres-Douglas partition functions themselves, for which a similar analysis can be done using the curves in \cite{Argyres:1995xn}.

\subsection{Expansion at AD point}
For generic masses, the SW surface for $N_f$ flavours is a rational elliptic surface with an $I_{4-N_f}^*$ fibre at infinity corresponding to weak coupling,  and   $2+N_f$ isolated $I_1$ singularities at strong coupling.\footnote{See Table \ref{tab:kodaira} in Appendix \ref{sec:Kodaira_invariant} for an explanation of the Kodaira classification of singular fibres} For specific configurations of the masses however, the singularities can collide. For instance, when $m_{N_f}$ is large and the other masses are generic, one $I_1$ fibre merges with the fibre $I_{4-N_f}^*$ to form a new fibre $I_{4-(N_f-1)}^*$ at infinity \cite{Aspman:2021vhs}. When $m_i=\pm m_j$ for $i\neq j$, two $I_1$ singularities collide and become an $I_2$ fibre \cite{Aspman:2022sfj} (see also \eqref{dDelta}). Finally, there is a locus $D^{\text{AD}}_{N_f}=0$ in mass space where fibres of  the type $I_n$ collide and become an additive fibre $\mathscr A\in\{II, III, IV, IV^*, III^*, II^*\}$. For the SW surfaces, the only possibilities are $\mathscr A\in\{II, III, IV\}$.

On important distinction  has to be done between the SQCD curves \eqref{eq:curves} and the curves for the AD theories,  discussed in \cite{Argyres:1995xn} for instance. As recently reviewed in \cite{Closset:2021lhd}, the fibre at infinity essentially determines the field theory by giving it  a `UV definition': For SQCD with $N_f$ hypermultiplets it is $I_{4-N_f}^*$, while for the AD curves it is $II^*$, $III^*$ and $IV^*$ for the $II$, $III$ and $IV$ AD theories (see for instance \cite{Caorsi:2019vex}). On the other hand, the singular types $II^*$, $III^*$ and $IV^*$ themselves are associated with the Minahan--Nemeschansky (MN) SCFTs \cite{Minahan:1996fg,Minahan:1996cj}. These theories do not appear in SU(2) SQCD, however, the below analysis carries over to a large extent to such singularities.
In this section, we focus on the SQCD curves, with masses tuned such that an AD point appears. We plan to discuss the topological theory for the AD curves in a future work \cite{AFMM:future}.

When the SW surface contains such an additive fibre $II$, $III$, $IV$, the functional invariant $\CJ$ becomes either 0 (for $II$ and $IV$) or $12^3$ (for $III$). As we identify  it with the modular $j$-invariant, these correspond to points $\tn$ in the fundamental domain with $j(\tn)$ either $0$ or $12^3$, which means that $\tn=\gamma\alpha$  or $\tn=\gamma i$ for some $\gamma\in \psl$, where $\alpha=e^{2\pi i/3}$ is a cube root of unity. Those are the elliptic points of $\psl$. In general, the fundamental domain $\CF(\bfm)$ for some mass configuration $\bfm$ is not modular, however in some cases $\CF(\bfm)$ is the fundamental domain for a subgroup of $\psl$ and then $\tn$ is an elliptic point for that subgroup \cite{Magureanu:2022qym,Closset:2021lhd,Aspman:2021vhs}.\\

 Let us study the behaviour of the $u$-plane integrand\footnote{We omit the measure $d\tau\wedge d\bar\tau$ the discussion here} 
\begin{equation}
\CI(\tau,\bar\tau)=\nu(\tau) \Psi(\tau,\bar\tau,\bfz,\bar\bfz) e^{2pu(\tau)/\Lambda_{N_f}^2+\bfx^2 G_{N_f}(\tau)}
\end{equation}
of \eqref{generaluplaneintegral} near $\tn$, where we drop some of the dependence for brevity.
We also turn off the background fluxes for now. Recall from \eqref{bfz_def} that the elliptic variable then reads $    \bfz=\frac{\bfx}{2\pi\Lambda_{N_f}}\frac{du}{da}$. 
Let us assume that we can express $ \Psi(\tau,\bar\tau,\bfz,\bar\bfz)=\partial_{\bar\tau}\widehat G(\tau,\bar\tau,\bfz,\bar\bfz)$, 
then the function 
\begin{equation}\label{widehat_h}
\widehat h(\tau,\bar\tau)=\nu(\tau) \widehat G(\tau,\bar\tau,\bfz,\bar\bfz) \,e^{2pu(\tau)/\Lambda_{N_f}^2+\bfx^2 G_{N_f}(\tau)}
\end{equation}
is an anti-derivative of the integrand, as it satisfies  $\partial_{\bar \tau} \widehat h(\tau,\bar\tau)=\CI(\tau,\bar\tau)$. Thus $-\mathrm d(\widehat h \,d\tau)=d\tau\wedge d\bar\tau\, \CI$, such that we can apply Stokes' theorem and find that the integral of $d\tau\wedge d\bar\tau\, \CI(\tau,\bar\tau)$ over $\CF(\bfm)$ is given by $-\int_{\partial \CF(\bfm)}d\tau\,\widehat h(\tau,\bar\tau)$.

As discussed in Section \ref{sec:integrationFD}, near an elliptic point $\tn$ the integrand $\widehat h$ can be  singular, and there can be a contribution from a small contour integral with radius $\varepsilon$ as $\varepsilon\to 0$. 
In order to find the contribution, 
we expand $\widehat h$ around an elliptic point $\tn$ as  
  \be\label{hat_h_expansion}
\widehat h(\tau,\bar \tau)=\sum_{m\gg -\infty,n\geq 0} d_{0}(m,n)\,(\tau-\tau_{0})^m\,(\bar \tau-\bar \tau_{0})^n.
\ee
Recall from part I that the  contribution from the elliptic point $\tn$ is then 
\begin{equation}\label{tn_contr}
    [\CI]_{\tn}=\frac{n_{0}}{k_{0}}\, d_0(-1,0),
\end{equation}
where $n_0$ and $k_0$ are integers explained in Section   \ref{sec:integrationFD}.

It is crucial that $\widehat h$ is analytic in $\bar\tau$ near $\bar\tn$, that is, the non-holomorphic expansion \eqref{hat_h_expansion} only contains non-negative powers of $\bar\tau-\bar\tn$. As shown in \cite{Korpas:2019cwg},  anti-derivatives of $\Psi$ generally have poles in the elliptic variable. However, the poles can be avoided for a certain choice of anti-derivative.

Deriving the non-holomorphic expansion \eqref{hat_h_expansion} of the anti-derivative of $\widehat h$ of the $u$-plane integrand for a generic four-manifold with period point $J$, 't Hooft flux $\bfmu$ and arbitrary AD configuration is a challenging task. Due to the product structure \eqref{widehat_h}, we can study the expansion of the anti-derivative $\widehat G$ separately from the measure $\nu$. We discuss those expansions in the following subsections in as much generality as possible. Our primary interest is to determine in which cases the expansion \eqref{hat_h_expansion} has a nonzero coefficient $d_0(-1,0)$, which is why we mainly focus on the leading exponent rather than the precise (non-zero) leading coefficients.

\subsection{Measure factor}
In this subsection, we discuss the leading behaviour of the $u$-plane measure $\nu$ in the absence of background fluxes. Recall from \eqref{measurefactornf} that in this case the measure factor is proportional to 
\begin{equation}\label{measure_exp}
    \nu\propto\frac{da}{d\tau}\left(\frac{du}{da}\right)^{\frac\chi2}\Delta^{\frac\sigma8}.
\end{equation}
As noticed in \cite[Section 3.1]{Aspman:2021vhs}, AD points $\tn$ are characterised by the property that $\frac{du}{da}(\tn)=0$. In appendix \ref{sec:Kodaira_invariant}, we prove this rigorously for any elliptic surface containing a singular fibre of type $II$, $III$ and $IV$.\footnote{Many of the results easily carry over to singularities of type $IV^*$, $III^*$ and $II^*$ as well.} In order to work without square roots \eqref{dadu_def}, let us define 
\begin{equation}\label{omega_def}
    \omega\coloneqq \left(\frac{du}{da}\right)^2,
\end{equation}
which is proportional to $\frac{g_3}{g_2}\frac{E_4}{E_6}$. Then $\omega$ has an expansion 
\begin{equation}
    \omega(\tau)=d_\CT(\tau-\tn)^{\ord \omega}+\dots 
\end{equation}
at $\tau\to\tn$. The  value of $\ord \omega$ depends on the precise configuration of singular fibres in the elliptic surface,  however, for the SQCD curves it is 1 in most cases. The value of $d_\CT$ can be determined exactly as well (see appendix \ref{sec:Kodaira_invariant}). However, here we are primarily interested in the leading behaviour, for which it is enough that $d_\CT\neq 0$.

In order to determine the expansion of $\nu$ at $\tn$, we define  $u_0=u(\tau_0)$, such that we have an expansion
\begin{equation}\label{u_tn_exp}
    u(\tau)=u_0+c_\CT(\tau-\tn)^{\ord \ttu},
\end{equation}
where we have defined $\ttu=u-u_0$. This is because unless $u_0=0$, $u$ does not vanish at $\tn$ and so does not have a positive order of vanishing.\footnote{In the following, we slightly abuse notation, and define  $\ord f$ as a rational number such that $f(\tau)\sim (\tau-\tn)^{\ord f}+\dots$ for $\tau\to\tn$. This number can in particular also be negative. Furthermore, we distinguish this order of vanishing as a function of $\tau$ from the order of vanishing of say the \ws{} invariants $g_2$ and $g_3$ as a function of $u$, which we denote by $o_2$, $o_3$ etc., particularly in  Appendix \ref{sec:Kodaira_invariant}.}
Again, the order of vanishing $\ord \ttu$ as well as the nonzero coefficient $c_\CT$ can be determined exactly given a singular elliptic surface. The value $\ord \ttu$ turns out to partially characterise the \emph{type} of AD point, i.e. type $II$, $III$ or $IV$.

In the following, we focus only on the leading behaviour, such that it is enough to calculate the orders of vanishing of all holomorphic functions involved. From \eqref{u_tn_exp} we can determine the order of vanishing of $\frac{du}{d\tau}$, it is  
\begin{equation}
    \ord \frac{du}{d\tau}=\ord \ttu -1.
\end{equation}
On the other hand, from \eqref{omega_def} it is clear that 
\begin{equation}
    \ord \frac{da}{du}=-\frac12\ord\omega.
\end{equation}
Thus from $\frac{da}{d\tau}=\frac{da}{du}\frac{du}{d\tau}$ it is clear that 
\begin{equation}
    \ord \frac{da}{d\tau}=-\frac12\ord\omega+\ord \ttu-1.
\end{equation}
The order of vanishing of the physical discriminant depends on the type of AD point as well. We can use the identity $\eta^{24}\propto \Delta /\omega^{6}$ \cite[Equation (3.9)]{Aspman:2021vhs}, and using the fact that $\eta(\tn)\neq 0$, we have 
\begin{equation}
    \ord\Delta=6\, \ord\omega.
\end{equation}
We have thus calculated the order of vanishing of the $u$-plane measure $\nu$ \eqref{measure_exp} for any given elliptic point $\tn$, it is
\begin{equation}
    \ord \nu=\ord \ttu -1 + \frac14\left(3\sigma+\chi-2\right)\ord\omega.
\end{equation}
We can further use that $\chi+\sigma=4$ for manifolds with $b_2^+=1$, since otherwise the $u$-plane integral vanishes. Thus
\begin{equation}\label{ord_nu}
    \boxed{\ord\nu=\ord\ttu-1+\frac{\sigma+1}{2}\ord\omega.}
\end{equation}
It is important to notice that $\ord\ttu$ and $\ord \omega$ are not given uniquely for a type ($II$, $III$ and $IV$) of AD point. Rather, they can differ depending on the configuration containing a given Kodaira singularity. The `degree of freedom' is the undetermined order of vanishing of the \ws{} invariant $g_2$ or $g_3$ on the base $\mathbb P^1$. Since for the SQCD curves $g_2$ and $g_3$ are polynomials of degree $2$ and $3$ however, this order of vanishing is greatly restricted. Thus in practise, for the curves under consideration, both $\ord \ttu$ and $\ord \omega$ can  take at most two values for a given type $II$, $III$ or $IV$. We summarise these values in the following table \ref{orduordomega_values}.\footnote{We do not know if the second possibility for type $III$ is realised, but are also not able to rule it out.}
We stress again that this result \eqref{ord_nu}  is the order of vanishing for any measure factor \eqref{measure_exp} calculated from an arbitrary elliptic surface containing an additive fibre of type $II^{(*)}$, $III^{(*)}$ or $IV^{(*)}$ (including the Minahan--Nemeschansky theories), and will be useful for the study of other $u$-plane integrals. 

\begin{table}%
\renewcommand{\arraystretch}{1}
\centering 
\begin{tabular}{|Sc| Sc|Sc|Sc|Sc| Sc|} %
\hline
$\mathscr A$& $o_2$ & $o_3$ & $o_\Delta$ & $\ord \ttu$ & $\ord\omega$\\  \hline  
$II$ & 1, 2&1&2&3, $\frac34$&1, $\frac14$ \\  
$III$ &1&2, 3& 3& 2, $\frac23$& 1, $\frac13$  \\  
$IV$ &2&2&4&$\frac32$&1\\  \hline
\end{tabular}
\caption{The possible orders of vanishing of the SW curves at AD points $II$, $III$, $IV$. The first three columns are the order of vanishing at $u_0$ of the \ws{} invariants $g_2$, $g_3$ and the discriminant $\Delta$ as a function of $u$, while $\ord\ttu$ and $\ord \omega$ are the orders of vanishing of $\ttu$ and $\omega$ as a function of $\tau$. The values separated by commas give the two different possibilities. For arbitrary elliptic surfaces, the values of $\ord\ttu$ and $\ord\omega$ are given in Appendix \ref{sec:Kodaira_invariant}.}\label{orduordomega_values}
\end{table} 

For the manifolds under consideration, the signature is bounded as $\sigma\leq 1$. Since $\ord\omega>0$, this implies that the measure can become arbitrarily singular for manifolds with large $-\sigma$. Similarly, one can expect selection rules on the signature for the $u$-plane integrand to be sufficiently singular to have a nonzero residue. 

Finally, if we include background fluxes in the theory, the path integral acquires couplings $\prod_{i,j=1}^{N_f}C_{ij}^{B(\bfk_i,\bfk_j)}$ which enter the measure factor in \eqref{measurefactornf}. For $N_f=2$ with equal masses we have expressions of $C_{ij}$ in terms of modular forms, see for instance \eqref{CijNf2}. Since we do not know the coupling $v$ analytically however, it is difficult to study the behaviour of those couplings at a given elliptic point $\tn$.
We hope to come back to this point in future work.

\subsection{Photon path integral}\label{sec:AD_photon_path_integral}
Having discussed the singular behaviour of the measure factor at any given elliptic AD point, we discuss in this subsection the series expansion of the anti-derivative $\widehat G$ of the Siegel-Narain theta function. This function has been determined in \cite{Korpas:2019cwg} for a canonical choice of period point $J\in H^2(X,\mathbb R)$, i.e. either $J=(1,\boldsymbol{0})$ for odd intersection forms on $H^2(X,\mathbb Z)$, or $J=\frac{1}{\sqrt2}(1,1,\boldsymbol{0})$ for even intersection forms. Including surface observables (or elliptic arguments in general), $\widehat G$ takes the form of a non-holomorphic completion of an Appell--Lerch sum.

\subsubsection*{Odd intersection form}

Since it is rather involved to study the general expansion of $\widehat G$ for all 't Hooft fluxes $\bfmu$ and values or $\sigma$ of  $b_2$ simultaneously, let us focus on manifolds $X$ with odd intersection form and a fixed flux. 
Consider the Siegel-Narain theta function for $\bfmu=(\mu_+,\bfmu_-)\equiv(\frac12,\frac12,\dots,\frac12) \mod \mathbb Z^{b_2}$ with $K=(3,1,\dots, 1)$, such that $K^2=9-n$ with $n=b_2^-$ and $\bfmu\equiv K/2\mod L$, as is the case for the del Pezzo surfaces $dP_n$ for instance. From  \eqref{Psi_surface} and \eqref{Psi_surface_L-} we have 
 \begin{equation}
 \Psi_{\bfmu}^J(\tau,\bfrho)=i f(\tau,\rho_1)\prod_{k=2}^{b_2}\jt_1(\tau,\rho_k).
  \end{equation}
Then $\Psi_{\bfmu}^J(\tau,\bfrho)=\partial_{\bar\tau}\widehat G(\tau,\bar\tau,\rho,\bar\rho)$ with 
\begin{equation}\label{psi_H_hat}
\widehat G(\tau,\bar\tau,\rho,\bar\rho)=i \widehat H(\tau,\bar\tau,\rho,\bar\rho)\prod_{k=2}^{b_2}\jt_1(\tau,\rho_k),
\end{equation}
where $\widehat H(\tau,\bar\tau,\rho,\bar\rho)$ is defined in \eqref{hatHdef}.
Since $\rho \to 0$ for $\tau\to\tn$, we consider the limit of \eqref{psi_H_hat} for $\rho\to 0$.
As explained in detail in Appendix \ref{app:mockjacobi}, the function $\widehat H(\tau,\bar\tau,\rho,\bar\rho)$ has a well-defined Taylor series in $(\rho,\bar\rho)$, and for $\rho,\bar\rho\to 0$ we have $\widehat H(\tau,\bar\tau,\rho,\bar\rho)\to \widehat H(\tau,\bar\tau)$. On the level of the Siegel-Narain theta function $\Psi_\bfmu^J(\tau,\rho)$, the limit $\rho\to 0$ exists since $\Psi$ is a well-defined function on $\mathbb H\times \mathbb C^{b_2}$ \cite{ZwegersThesis,zagier2009,Korpas:2017qdo}: For $\rho=0$, $\Psi(\tau,0)$ is a non-holomorphic vector-valued modular form for $\slz$ of weight $(\frac n2,0)$.

We can in fact study the  Taylor series of the  whole anti-derivative \eqref{psi_H_hat} in $\rho$, while we focus on the leading contribution as $\rho\to 0$. For the Taylor series of the elliptic Jacobi theta series \eqref{jt_elliptic}, 
\begin{equation}
\jt_1(\tau,\rho)=\sum_{n=0}^\infty \frac{1}{n!}\partial_\rho^n \jt_1(\tau,0)\rho^n,
\end{equation}
we determine
\begin{equation}
\partial_\rho^n\jt_1(\tau,0)=i (2\pi i)^n\sum_{r\in \mathbb Z+\frac12}(-1)^{r-\frac12}r^n q^{r^2/2}.
\end{equation}
In the limit $\rho\to 0$, we are primarily interested in the first nonzero term,
\begin{equation}
\vartheta_1(\tau,\rho)=-2\pi \eta(\tau)^3 \rho +\CO(\rho^3).
\end{equation}

The Taylor series of $\widehat G$ at $\rho=0$ then reads
\bea\label{hbar_at_rho=0}
\widehat G(\tau,\bar\tau,\rho,\bar\rho)=(-2\pi)^{b_2^-}\widehat H(\tau,\bar\tau)\eta(\tau)^{3b_2^-}\rho_2\dots \rho_{b_2}+\CO(\rho_k^2).
\eea

\subsubsection*{Chowla--Selberg formula}
Let us determine the holomorphic expansion of this function at $\tau=\tn$. The function $\widehat H(\tau,\bar\tau)$ is a non-holomorphic modular form. If its value at the elliptic point $\tn$ is nonzero, it has a regular Taylor expansion in $\tau-\tn$ starting with a constant term. If it vanishes at $\tn$ on the other hand, it has a positive order of vanishing, shifting the holomorphic series of the remainder of the integrand. In order to find the leading term in the expansion, we thus have to evaluate it at $\tau=\tn$. 

In many cases, explicit values of holomorphic modular forms at the elliptic fixed points $i$ and $\alpha$ of $\slz$ are specific combinations of powers of $\pi$ and the Euler gamma function $\Gamma$ evaluated at rational numbers.  This is a consequence of the \emph{Chowla--Selberg formula}, which describes the value of modular forms at complex multiplication points in terms of products of gamma functions.\footnote{See Appendix \ref{chowla_selberg} for a survey on the Chowla--Selberg formula}  The Chowla--Selberg formula expresses the value of any modular form $f$ of weight $k$ at a \emph{complex multiplication} (CM) point $\fz\in\mathbb H$ as 
\begin{equation}
    f(\fz)\in\overline{\mathbb Q}\,\Omega_K^k,
\end{equation}
where $\overline{\mathbb Q}$ is the  field of  all algebraic numbers\footnote{An algebraic  number is a  number which is a root of a non-zero polynomial in one variable with integer coefficients}, and $\Omega_K\in\mathbb C^*$ is a complex number that depends only on the imaginary quadratic field $K$ containing $\fz$. Any  $\tn$ in the $\slz$-orbit of the $\slz$ elliptic points $i$ or $\alpha=e^{2\pi i/3}$ is a CM point, and the period $\Omega$ is computed from the Chowla--Selberg formula (Theorem \ref{thm_CS}) as
\begin{equation}\label{elliptic_periods}
    \begin{aligned}
        \Oi&=\frac{\Gamma(\tfrac14)^2}{4\pi^{\frac32}},\\
        \Oa&=\frac{3^{\frac14}\Gamma(\tfrac13)^3}{4\pi^2},
    \end{aligned}
\end{equation}
where $\mathbb Q(i)$ and $\mathbb Q(\alpha)$ are the quadratic fields generated by the CM points $i$ and $\alpha$.
Since the values \eqref{elliptic_periods} are known, the value $f(\fz)$ of some elliptic point $\fz$ of a modular form $f$ of weight $k$ can now more easily be determined, since $f(\fz)\Omega_{\mathbb Q(\fz)}^{-k}$ is an algebraic number. These can  often   be found using an integer relation algorithm such as the LLL-algorithm (see for instance \cite{Cohen1993}).

The Chowla--Selberg formula is generally stated to apply to holomorphic modular forms  with algebraic Fourier coefficients for finite index subgroups of $\slz$. 
Generalisations to non-holomorphic modular forms \cite{bruinier2013} and mock modular forms \cite{GUERZHOY2010,CHOI2018428}  have been explored in the literature, however, the full range of validity of this formula is possiblly not determined yet.

A somewhat trivial example of the Chowla--Selberg formula concerns the anti-derivative $\widehat H$. Its values at the elliptic points are
\begin{equation}\begin{aligned}\label{Hhat_elliptic_CS}
  \widehat H(i,\bar i)&=0, \\
 \widehat H(\alpha,\bar \alpha)&=0.
\end{aligned}\end{equation}
These can be found from the modular transformations of $\widehat H$ (see \eqref{Hhat_mod_transf}): As it is a non-holomorphic modular form for $\slz$ of weight $(\frac12,0)$, one evaluates the modular $S$ and $ST$ transformation of $\widehat H$ at their fixed points $i$ and $\alpha$, and the multipliers and weight factors differing from one results in $\widehat H(i,\bar i)$ and $\widehat H(\alpha,\bar \alpha)$ vanishing.
The coefficients of higher orders in the Taylor expansion can also be found, see Appendix \ref{app:mockjacobi}, and for instance \eqref{H2coef} and Table \ref{H2hcoefficients}. In particular, $\widehat H$ has a non-holomorphic Taylor series at $(\tau,\bar\tau)=(\tn,\bar\tn)$ with non-zero coefficient of $(\tau-\tn)^1(\bar\tau-\bar{\tau}_0)^0$. Thus the holomorphic order of vanishing is $\ord \widehat H=1$. 

The property that anti-derivatives of the photon path integral vanishes at elliptic points is certainly not true in general. Indeed, anti-derivatives are only defined up to integration constants, which can change properties of distinct anti-derivatives drastically.

\subsubsection*{Integration constant}
There are two obvious way in which an integration `constant' can be added. First, we can add to $\widehat G$ a weakly holomorphic modular function $g$ of weight $\frac12$. Indeed, for this new anti-derivative $\widehat h=f(\tau)(\widehat G(\tau,\bar\tau)+g(\tau))$, it remains true that  $\partial_{\bar \tau} \widehat h(\tau,\bar\tau)=\CI(\tau,\bar\tau)$. Regarding the functions $H$ and $F$ defined in Section \ref{factorisation_psi}, and their completions $\widehat H$ and $\widehat F$ defined in Appendix \ref{app:mock}, their difference is the weakly holomorphic modular form $(\jt_2^4+\jt_3^4)/\eta^3$, as stated precisely in \eqref{HtauDef} and \eqref{FhatHhat}. 
By numerically evaluating the Eichler integral and using the Chowla--Selberg formula, we find the following values\footnote{We omit the anti-holomorphic dependence in the notation.}, 
   \begin{equation}\label{Fhat_values}
    \begin{aligned}
        \widehat F(i)&=\frac12\Oi^{\frac12}, 
        &\widehat F(i+1)&=0, \\
        \widehat F(e^{\frac23\pi i})&=\frac{e^{-\frac{5}{24}\pi i}}{2^{\frac23}\sqrt3}\Oa^{\frac12},  
        &\widehat F(e^{\frac13\pi i})&=\frac{e^{\frac{5}{24}\pi i}}{2^{\frac23}\sqrt3}\Oa^{\frac12}, \\
        \widehat F(\tfrac12+\tfrac i2)&=\frac{e^{\frac38\pi i}}{2^{\frac34}}\Oi^{\frac12},\quad
        &\widehat F(\tfrac{1}{\sqrt3}e^{\frac16\pi i})&=\frac{e^{\frac{13}{24}\pi i}}{2^{\frac23}\sqrt[4]{3}}\Oa^{\frac12}.
    \end{aligned}
\end{equation}
As is apparent, all coefficients of the Chowla--Selberg periods are algebraic numbers. For example, the coefficient of $ \widehat F(e^{\frac23\pi i})$ has minimal polynomial $2^{16}3^{12}x^{24}+1$.
Since $\widehat F$ is a non-holomorphic modular form of weight $\left(\tfrac12,0\right)$ for $\Gamma^0(2)$ (which is of index 3 in $\slz$), these six values determine $\widehat F(\tn)$ for any elliptic point $\tn$ of $\slz$. In contrast to the modular form $\widehat H$, the holomorphic order of vanishing  of $\widehat F$ at any $\slz$ elliptic point is zero.

Alternatively, we can also add a weakly holomorphic function of weight 2 to $\widehat h$ directly. 
The difference between the two places where the holomorphic function is inserted is of course equivalent. 
The basis for the space of weakly holomorphic modular forms for $\psl$ is given by the derivatives of powers of $j$, $\partial_\tau j^l$, $j\in\mathbb N$. These are given by
\begin{equation}
    \partial_\tau j^l=-2\pi i\,  l \, j^l \frac{E_6}{E_4},
\end{equation}
with are holomorphic at $i$ and $\alpha$. In fact, they vanish at both $i$ and $\alpha$ (we have $\ord \partial_\tau j(\tau)^l=3l-1$ at $\tau=\alpha$ and $\ord \partial_\tau j(\tau)^l=1$ at $\tau=i$). 
Thus they do not alter the residue.

\subsection{AD contribution}
We can now combine the results from the above considerations regarding the leading behaviour of the anti-derivative $\widehat h$ \eqref{hbar_at_rho=0}. 
Recall from \eqref{bfz_def} that $\ord\rho=\frac12\ord \omega$. The Dedekind eta-function $\eta$ is nowhere vanishing in $\mathbb H$. Thus from \eqref{hbar_at_rho=0} we have
\bea\label{ordGhat}
    \ord \widehat G&=\ord \widehat H+  b_2^- \ord \rho \\
    &=1+\frac{1-\sigma}{2}\ord\omega.
\eea
From the anti-derivative \eqref{widehat_h} we are missing the exponentiated observables, $2pu(\tau)/\Lambda_{N_f}^2$ and $\bfx^2 G_{N_f}(\tau)$. Of course, for $\tau$ near $\tn$ the first exponential is regular, since $2pu(\tau)/\Lambda_{N_f}^2\to 2pu(\tn)/\Lambda_{N_f}^2$. For the surface observable, we need to study the contact term \eqref{contactterm}, which contains the three functions $E_2$, $\omega$ and $u$. In Lemma \ref{lemmaE2} in Appendix \ref{app:modularforms2} we prove that $E_2(\tn)\neq 0$ for any elliptic point $\tn$. Thus 
\begin{equation}
    \ord G_{N_f}=\ord\omega.
\end{equation} 
Since the contact term is exponentiated, $e^{\bfx^2 G_{N_f}(\tau)}$ however becomes a constant as $\tau\to\tn$.\footnote{This should be modified when the exponential is expanded and we study a particular coefficient in the series in $\bfx$. Here, we are interested mainly in the generating function.} This shows
\begin{equation}
   \ord e^{2pu(\tau)/\Lambda_{N_f}^2+\bfx^2 G_{N_f}(\tau)}=0.
\end{equation}
Combining now \eqref{ord_nu} and \eqref{ordGhat} in  \eqref{widehat_h} we have $\ord\widehat h=\ord\nu+\ord \widehat G$, which gives
\bea\label{ordhhat}\boxed{
    \ord\widehat h=\ord \ttu+\ord \omega.}
\eea
Strikingly, the dependence on the signature of the  manifold drops out. Moreover, this value is bounded from below,
\begin{equation}
\ord\widehat h \geq 1
\end{equation}
as is clear from Table \ref{orduordomega_values}.\footnote{If we allow for other elliptic surfaces with type $II$, $III$ and $IV$ singular fibres, the value of $\ord\ttu+\ord \omega$ is only bounded from below by 0, since it converges to 0 for either $o_2\to\infty$ or $o_3\to\infty$, whichever is undetermined by the Kodaira classification. For \emph{rational} elliptic surfaces, the \ws{} invariants $g_2$ and $g_3$ are polynomials of degree at most $4$ and $6$, and thus $\ord\ttu+\ord \omega$ is bounded from below by a positive rational number.} In fact, for \emph{any} elliptic surface containing a singular fibre of type $II$, $III$ or $IV$, the order of vanishing \eqref{ordhhat} is strictly positive. 

We have shown that in the absence of background fluxes, for a fixed 't Hooft flux on arbitrary four-manifolds with odd intersection form, the anti-derivative of the $u$-plane integrand has an expansion \eqref{hat_h_expansion} with 
\begin{equation}
    d_0(-1,0)=0.
\end{equation}
That is, according to \eqref{tn_contr} the $u$-plane integrals for these manifolds does not acquire a contribution from any possible AD point in $\CN=2$ SQCD,
\begin{equation}
    [\CI]_{\tn}=0.
\end{equation}
This result is not entirely unexpected, since this class of manifolds also includes the del Pezzo surfaces $dP_n$, for which there is no SW contribution and thus also no contribution from the strong coupling $I_r$ cusps of the $u$-plane.  Of course, AD points are collisions of $I_r$ and $I_{r'}$ cusps, and it is conceivable that if there is no contribution from those singularities if they are separated, their collision does not produce a contribution either.  See also Section \ref{sec:ContCusps} for the related discussion on the collision of mutually local $I_r$ and $I_r'$ singularities to an $I_{r+r'}$ singularity.

\section{Conclusion and Outlook}\label{sec:discussion}

In this part II, we formulated and evaluated topological partition functions for massive SQCD on compact four-manifolds.

As addressed in the introduction, topological correlators take the form of a continuous $u$-plane integral and a finite sum over the Coulomb branch singularities. In part I, we argued that the $u$-plane integral can itself be expressed as a sum over the CB singularities (see \eqref{integrationresult}). The SW curve fibered over the Coulomb branch constitutes an elliptic surface $\CS$, whose singular fibres correspond to the CB singularities. The possible singular fibres $F_i$ of an elliptic surface fall into Kodaira's classification. To any surface $\CS$, we can associate its configuration of singular fibres $\{F_1, \dots, F_k\}$. We may therefore express topological correlation functions $ Z_\bfmu^X$ for fixed fluxes $\bfmu$ on a given manifold $X$ as a sum
$Z_\bfmu^X[\CS]= \sum_{i=1}^k \Phi_{\bfmu,i}^J$, 
where $\Phi_{\bfmu,i}^J$ is the contribution from the singular fibre $F_i$. This proposal is valid for massive SU(2) SQCD, as well as the topological $\CN=2^*$ theory \cite{Manschot:2021qqe}.

An important consistency check of this formulation and of correlation functions in general is the invariance under reparametrisations of the theory. 
For a given 4d $\CN=2$ theory, there can be multiple SW curves describing the same local  low-energy physics \cite{Seiberg:1994aj,Moore:1997pc,Shapere:2008zf,Gaiotto:2010be,Aharony:2013hda,Manschot:2021qqe}.\footnote{This was first noticed in the pure 4d SU(2) theory, where the SW curves associated with $\Gamma(2)$ \cite{Seiberg:1994rs} and $\Gamma^0(4)$ \cite{Seiberg:1994aj} monodromy groups were understood to be related by a 2-isogeny and a change of normalisation of the periods.} Recently, the choice of SW curve has been identified with the choice of global structure of a rank one 4d $\CN=2$ theory,  i.e. the spectrum of line operators  \cite{Argyres:2022kon}. The global forms of a given local theory are related by gauging 1-form symmetries. On the level of the SW curve, the relation is provided by compositions of isogenies \cite{Closset:2023pmc}. An isogeny is given by a quotient $\CS/\alpha$ of the elliptic surface $\CS$ by an automorphism $\alpha$ induced by the 1-form symmetry group, which is a particular subgroup of the Mordell-Weil group $\text{MW}(\CS)$ \cite{Closset:2021lhd}. The Mordell-Weil group sits in a larger group $\text{Aut}(\CS)$ of automorphisms of the surface $\CS$ itself, which encodes the 1-form as well as the 0-form symmetries of the theory \cite{Caorsi:2018ahl,Caorsi:2019vex,Closset:2023pmc,ACFMM23}. Formulating $\CN=2$ theories on (especially non-spin) compact four-manifolds provides an ideal testing ground for the effects of such higher-form symmetries and their anomalies \cite{Aharony:2013hda,Gaiotto:2014kfa,Cordova:2018acb,Ang:2019txy,Brennan:2022tyl,Brennan:2023kpo,Brennan:2023vsa}. It would thus be interesting and important to explore the action of $\text{Aut}(\CS)$ on topological correlators.

Formulating superconformal field theories on four-manifolds has been proven to be a fruitful tool for understanding inherent features about the theories themselves, as well as the topological invariants of the underlying spacetime. An example of the former is the derivation of the conformal and flavour central charges of 4d $\CN=2$ SCFTs  from the topological twist, by comparing the $U(1)_R$ anomaly of the Coulomb branch measure with the conformal anomaly of the trace of the energy-momentum tensor generated by a background gravitational field
\cite{Shapere:2008zf}. Given these central charges of an SCFT, topological correlators satisfy a selection rule due to the $U(1)_R$ anomaly \cite{Aharony_2008,Xie:2013jc,Shapere:2008zf,Marino:1998bm,Moore:2017cmm,Gukov:2017,Witten:1995gf}. For rank 1 $\CN=2$ SCFTs, this selection rule identifies the anomalous $R$-charge of the vacuum with the $R$-charge of the operator whose correlation function is evaluated. 

Moreover, the existence of superconformal theories  on a four-manifold $X$ can give new structural
insights on the invariants of $X$. For the simplest Argyres-Douglas SCFT $H_0$, it was found that the regularity of topological partition functions in the critical limit imposes \emph{sum rules} on the classical Seiberg-Witten invariants \cite{Marino:1998uy}.
All other Argyres-Douglas theories, including the $H_1$ and $H_2$ theories appearing as critical points in SU(2) SQCD, are associated with noncompact Higgs branches. For those higher AD theories, postulating analogous sum rules is therefore more intricate. Nevertheless, identifying the hypermultiplet masses as equivariant deformation parameters, the Higgs branch singularities can be made precise \cite{LoNeSha,Dedushenko:2017tdw}. 
We plan to discuss the interplay between Higgs branch singularities, sum rules and selection rules in more detail in a future work \cite{AFMM:future}.

While coupling SQCD to background fluxes for the flavour group makes the theory well-defined on non-spin manifolds $X$ with arbitrary 't Hooft flux, it renders the evaluation of topological correlators more involved. One possible reason for this is that the structure of poles and zeros of the background couplings can interfere with the pole structure of the mock modular forms used in the evaluation. Another obstruction comes from the branch points of the integrand, which only in the absence of background fluxes are guaranteed to not affect the result. In Appendix \ref{sec:bkgFluxes}, we analyse the various issues and discuss   possible resolutions. It would be desirable to have a more systematic treatment of this aspect.

The methods we employ to evaluate the topological partition functions of $\CN=2$ SQCD are largely independent on the specific form of the underlying SW curves. The SW curves for other 4d rank 1 $\CN=2$ theories are (generally non-modular) elliptic surfaces, where the `remaining' Kodaira singularities  absent in SQCD are the   types $II^*$, $III^*$ and $IV^*$, which are associated with the Minahan--Nemeschansky (MN) SCFTs \cite{Minahan:1996fg,Minahan:1996cj}.  It would  therefore be interesting to formulate topological partition functions of other 4d rank 1 $\CN=2$ theories, including all rank 1 $\CN=2$ SCFTs, which are classified \cite{Argyres:2015ffa, Argyres:2015gha, Argyres:2016xmc, Argyres:2016xua}.

Another interesting avenue for future research is the 5-dimensional uplift of 4d $\CN=2$ theories, for instance the  5-dimensional $\CN=1$ SU(2) gauge theory compactified on a circle 
\cite{Nekrasov_1998,Ganor:1996pc}. More generally, a family of 4d $\CN=2$ theories is represented by the Kaluza-Klein theories obtained by a circle compactification of the 5d $\CN=1$ SU(2) theory with fundamental flavours, which have   $E_n$ flavour symmetry \cite{Seiberg:1996bd,Morrison:1996xf}. See 
\cite{Chang:2017cdx,Closset:2022vjj,kim2023,Closset:2021lhd} for related recent works on 5d topological  partition functions.

\section*{Acknowledgments}
We thank Jim Bryan, Cyril Closset, Martijn Kool, Horia Magureanu, Greg Moore, Samson Shatashvili and Edward Witten for useful discussions and comments. JM was supported by the Laureate Award 15175 ``Modularity in Quantum Field Theory and Gravity" (2018-2022) of the Irish Research Council, and the Ambrose Monell Foundation (2022-2023). EF is supported by the EPSRC grant ``Local Mirror Symmetry and Five-dimensional Field Theory''.

\appendix \addtocounter{section}{3}
\section{Modular forms II}\label{app:modularforms2}
In this Appendix, we discuss aspects of modular forms relevant for this Part II. It complements Appendix \ref{app:modularforms} of Part I.
Section \ref{chowla_selberg} gives a review of the Chowla--Selberg formula, which is used in Section \ref{sec:ADcontribution} in the expansion of the photon path integral around the AD points. 
In Section \ref{app:derivatives}, we discuss additional properties of modular forms, such as derivatives and special values.  These are used in Section \ref{sec:ADcontribution} for the expansion of modular forms around elliptic points, as well as in Section \ref{sec:uplane_int_mock} in the study of anti-derivatives of the $u$-plane integrand.  Finally, Sections \ref{app:mock} and \ref{app:mockjacobi} list properties of mock modular forms and mock Jacobi forms relevant to the evaluation of $u$-plane integrals.

\subsection{Chowla--Selberg formula}\label{chowla_selberg}
Explicit values of modular forms at the elliptic fixed points $i$ and $\alpha=e^{2\pi i/3}$ can often be expressed as specific combinations of powers of $\pi$ and the Euler gamma function $\Gamma$. This is a consequence of the Chowla--Selberg formula, which describes the value of modular forms at complex multiplication points in terms of products of gamma functions. In order to understand the final formula, we need to introduce some basic notions. We follow \cite[$\S 1$ Section 6]{Bruinier08}. Generalisations to non-holomorphic modular forms \cite{bruinier2013} and mock modular forms \cite{GUERZHOY2010,CHOI2018428}   have been explored. In physics it has been used in the context of the BCOV conjecture \cite{Eriksson:2019jhg}.
Further results can be found in \cite{Chowla1949,Shimura2010,DHoker:2022dxx,shimura1977derivatives,MORENO1983226,lerch1897,Waldschmidt2006,LI2022108264,BRUINIER201738,voight2014,Cohen1993}.

Elliptic curves over $\mathbb C$ can be viewed as quotients  $E=\mathbb C /\Lambda$, with $\Lambda$ a lattice in $\mathbb C$. If $E' = \mathbb C/\Lambda'$ is another curve and $\lambda$ a complex number with $\lambda\Lambda \subseteq \Lambda'$, then multiplication by $\lambda$ induces an algebraic map from $E$ to $E'$. In particular, if $\lambda\Lambda \subseteq \Lambda$, then we get a map from $E$ to itself. 
Elliptic curves  $E=\mathbb C /\Lambda$  which satisfy $\lambda\Lambda \subseteq \Lambda$  for some non-real value of $\lambda$ are said to admit \emph{complex multiplication}.

We can think of the upper half-plane $\mathbb H$ as a moduli space for elliptic curves. Then the points in $\mathbb H$ that correspond to elliptic curves with complex multiplication (from now on called CM points) are the numbers $\mathfrak z\in\mathbb H$ which satisfy a quadratic equation over $\mathbb Z$.

Examples of CM points are the $\slz$-images $\gamma\tau_0$ of the $\slz$ elliptic points $\tau_0$. We can easily check that for any $\gamma=\begin{psmallmatrix}a&b\\c&d\end{psmallmatrix}$, the elliptic points $\gamma i$ and $\gamma\alpha$ satisfy
\begin{equation}\begin{aligned}\label{quadratic_eq_elliptic}
   0&= (c^2 + d^2)(\gamma i)^2 - 2 (a c + b d) \gamma i + a^2 + b^2, \\
    0&=(c^2 - c d + d^2) (\gamma\alpha)^2 + (b (c - 2 d) + 
    a (d-2 c ) ) \gamma\alpha + a^2 - a b + b^2.
\end{aligned}\end{equation}
The equations are strictly quadratic, since $c^2 + d^2=|ci+d|^2$ and $c^2 - c d + d^2=|c\alpha+d|^2$.

\begin{definition}[Discriminant of CM point] The discriminant of a CM point is the smallest discriminant of the quadratic polynomial over $\mathbb Z$ of which it is a root. 
\end{definition}
For the elliptic points in \eqref{quadratic_eq_elliptic}, we can easily compute the discriminants $-4$ for $\fz=\gamma i$ and $-3$ for $\fz=\gamma\alpha$. Generally, discriminants of CM points are congruent to $0$ or $1$ modulo $4$. For example, $i=\frac12\sqrt{-4}$ and $\alpha=\frac12(-1+\sqrt{-3})$.

There is a one-to-one relation between primitive positive definite binary quadratic forms of discriminant $D$ and the set $\mathfrak Z_D$ of CM points with discriminant $D$. For fixed $D$, there are only finitely many classes of such quadratic forms with discriminant $D$. Their number is the \emph{class number} $h(D)$ of the discriminant $D$.\footnote{A generalisation of the class numbers are the Hurwitz class numbers, which are class numbers weighted by the inverse of the order of the automorphism group in $\Gamma_1$.} Due to the bijection with $\mathfrak Z_D$, the class number $h(D)=|\Gamma_1\backslash \mathfrak Z_D|$ is given by the cardinality of the set of $\Gamma_1$-equivalent elements in $\mathfrak Z_D$. 

CM points $\fz$ have many special properties. For instance, the value of $j(\fz)$ or that of any other modular function with algebraic coefficients evaluated at $\fz$ is an \emph{algebraic number}\footnote{An algebraic  number is a  number which is a root of a non-zero polynomial in one variable with integer coefficients. The set of all algebraic numbers forms an algebraic field, commonly denoted by $\overline{\mathbb Q}$.}. Furthermore, these \emph{singular moduli} $j(\fz)$ give explicit generators of the class fields of imaginary quadratic fields. 

Many results depend on the following 
\begin{proposition}
Let $\fz\in\mathbb H$ be a CM point. Then $j(\fz)$ is an algebraic number. 
\end{proposition}
This is quite trivially true for the $\slz$ elliptic points: $j(\gamma i)=j(i)=12^3$ and $j(\gamma\alpha)=j(\alpha)=0$, which are both roots of degree 1 polynomials over $\mathbb Z$.
\begin{proof}
We can use the fact that if $\fz$ is a CM point, it is fixed $M\fz=\fz$ by some integer-valued matrix $M$ with positive determinant. Then $j(\tau)$ and $j(M\tau)$ are both modular functions on the subgroup $\Gamma_1\cap M^{-1}\Gamma_1 M$ of finite index in $\Gamma_1$, so they are algebraically dependent and thus there is a polynomial $P(X,Y)$ in two variables such that $P(j(M\tau),j(\tau))=0$. Since the Fourier expansion of $j$ has integer coefficients, the coefficients of $P$ can be chosen to lie in $\mathbb Q$. Since $M\fz=\fz$,  $j(\fz)$ is a root of the non-zero polynomial $P(X,X)\in \mathbb Q[X]$.
\end{proof}
The idea easily generalises to subgroups of $\Gamma_1$:

\begin{proposition}\label{prop_mod_function}
    Let $\fz\in\mathbb H$ be a CM point and $f$ a modular function for a finite index subgroup of $\slz$ with algebraic Fourier coefficients. Then $f(\fz)\in \overline{\mathbb Q}$. 
\end{proposition}

Interestingly, this statement is equivalent to the following

\begin{proposition}\label{prop_mod_form}
Let $\fz\in\mathbb H$ be a CM point and $f$ a modular form of weight $k$ for a finite index subgroup of $\slz$ with algebraic Fourier coefficients. Then there is a complex number $\Omega_{\fz}$ (depending on $\fz$ only), such that $f(\fz)\in \overline{\mathbb Q}\cdot \Omega_{\fz}^k$.
\end{proposition}
Proposition \ref{prop_mod_form} implies Proposition \ref{prop_mod_function} by specialising to $k=0$. Proposition \ref{prop_mod_function} implies Proposition \ref{prop_mod_form} on the other hand since if $f\in M_k$ and $g\in M_l$ then $f^l/g^k$ is a  modular function and is thus algebraic at $\fz$, such that $f(\fz)^{\frac 1k}$ and $g(\fz)^{\frac 1l}$ are algebraically proportional. 

The number $\Omega_{\fz}$ is defined only up to an algebraic number, and does not change up to an algebraic number if $\fz$ is mapped to $M\fz$ with some positive determinant integer valued matrix $M$. Thus Proposition \ref{prop_mod_form} lifts to arbitrary imaginary quadratic fields $K$,

\begin{proposition}\label{IQF_value}
    Let $K$ be an imaginary quadratic field. Then there exists a number $\Omega_K\in\mathbb C^*$ such that \begin{equation}f(\fz)\in\overline{\mathbb Q}\, \Omega_K^k
    \end{equation}
    for all $\fz\in K\cap \mathbb H$, all $k\in\mathbb Z$ and all modular forms f  of weight $k$ with algebraic Fourier coefficients. 
\end{proposition}

To find $\Omega_K$, we can compute $f(\fz)$ for any modular form. For instance, we can choose the non-holomorphic modular function $\Phi(\tau)=y|\eta(\tau)|^4$. As described above, for general $D$ we can look at CM points with discriminant $D$, of which modulo $\Gamma_1$ there are $h(D)$, the class number of $K$. Since there is no distinguished element of $\mathfrak Z_D$, we take the product over $\Gamma_1\backslash \mathfrak Z_D$ and take the $h'(D)$-th root, where $h'(D)=h(D)/\frac12w(D)$ is either $\frac 13$, $\frac 12$ or $h(K)$ depending on the value $D=-3$, $D=-4$ or $D<-4$ (these numbers take into account the multiplicities of the elliptic fixed points $i$ and $\alpha$). 
Then we have the following 

\begin{theorem}[Chowla--Selberg formula]\label{thm_CS}
Let $K$ be an imaginary quadratic field with discriminant $D$. Then 
\begin{equation}
    \prod_{\fz\in \Gamma_1\backslash\mathfrak Z_D}\left(4\pi \sqrt{|D|}\Phi(\fz)\right)^{2/w(D)}=\prod_{j=1}^{|D|-1}\Gamma\left(\tfrac{j}{|D|}\right)^{\chi_D(j)},
\end{equation}
where $\chi_D(\cdot)\coloneqq \left(\frac{D}{\cdot}\right)$ (the Legendre symbol) is the quadratic character associated to $K$, and $\Gamma$ is the Euler gamma function. 
\end{theorem}
This formula was essentially found by Lerch \cite{lerch1897} and rediscovered by Chowla and Selberg \cite{Chowla1949} in the study of  Epstein's $\zeta$-function.

Since $\eta^2$ is a weight $1$ modular form for $\Gamma_1$, we can find the number $\Omega_K$.

\begin{corollary}
The number $\Omega_K$ in Proposition \ref{IQF_value} can be chosen to be 
\begin{equation}
    \Omega_K=\frac{1}{\sqrt{2\pi|D|}}\left(\prod_{j=1}^{|D|-1}\Gamma\left(\tfrac{j}{|D|}\right)^{\chi_D(j)}\right)^{1/2h'(D)}.
\end{equation}
\end{corollary}

For the two CM points $i\in \mathfrak Z_{-4}$ and $\alpha\in\mathfrak Z_{-3}$ generating the quadratic imaginary fields $\mathbb Q(i)$ and $\mathbb Q(\alpha)$, we can now easily compute their corresponding Chowla--Selberg periods\footnote{They are periods in the sense of \cite{Kontsevich2001}.},
\begin{equation}\label{CSperiodsOaOi}
    \begin{aligned}
        \Oi&=\frac{\Gamma(\tfrac14)^2}{4\pi^{\frac32}},\\
        \Oa&=\frac{3^{\frac14}\Gamma(\tfrac13)^3}{4\pi^2}.
    \end{aligned}
\end{equation}
The fact that the periods for the quadratic fields $\mathbb Q(i)$ and $\mathbb Q(\alpha)$ can be written by powers of a single $\Gamma$ function is a consequence of Euler's reflection formula. For larger discriminant CM points, the periods are generally products of $\Gamma$-functions at multiple rational values. 

Since the values \eqref{CSperiodsOaOi} are known, the value $f(\fz)$ of some elliptic point $\fz$ of a modular form $f$ of weight $k$ can now more easily be determined, since $f(\fz)\Omega_{\mathbb Q(\fz)}^{-k}$ is an algebraic number, which often  can be determined using an integer relation algorithm such as the LLL-algorithm \cite[Section 2.6.1]{Cohen1993}. For instance, we have 
\begin{equation}
    E_4(i)=12\, \Oi^4, \qquad E_6(\alpha)=12^{\frac32}\, \Oa^6,
\end{equation}
where the algebraic coefficients are solutions to the linear and quadratic equations $x=12$ and $x^2=12^3$.

\subsection{Derivatives and theta functions}\label{app:derivatives}
Derivatives of modular functions are described by Ramanujan's differential operator. It increases the holomorphic weight by 2 and it can be explicitly constructed using the theory of Hecke operators \cite{ono2004}. 
For the derivatives of the Jacobi theta functions \eqref{Jacobitheta}, one finds
\begin{equation}\begin{aligned}\label{jacobiderivatives}
D\jt_2^4&=\tfrac 16 \jt_2^4\left( E_2+\jt_3^4+\jt_4^4\right),\\
D\jt_3^4&=\tfrac 16 \jt_3^4\left( E_2+\jt_2^4-\jt_4^4\right),\\D\jt_4^4&=\tfrac 16 \jt_4^4\left( E_2-\jt_2^4-\jt_3^4\right),
\end{aligned}\end{equation}
where $D\coloneqq \frac{1}{2\pi\im}\frac{d}{d\tau}=q\frac{d}{dq}$ and $E_2$ is the quasi-modular Eisenstein series \eqref{Ek} of weight 2. The derivatives of $\jt_j^2$ are given by $D\jt_j^2=\frac{D\jt_j^4}{2\jt_j^2}$ and therefore also in the ring of quasi-modular forms.

All quasi-modular forms for $\Gamma_1=SL(2,\mathbb Z)$ can be expressed as polynomials in $E_2$, $E_4$ and $E_6$. 
Derivatives of modular forms are quasi-modular in a precise way:
\begin{lemma}\label{derivative_mod_form}Let $f\in M_k(\Gamma_1)$ be a modular form of weight $k$. Then the \emph{Serre derivative} $\jt_k f\coloneqq \frac{1}{2\pi i}f'-\frac{k}{12} E_2 f$ is a  modular form  $\jt_k f\in M_{k+2}(\Gamma_1)$.
\end{lemma}
As a consequence, the quasi-modular forms for $\Gamma_1$ form a ring $\widetilde{M_*}(\Gamma_1)=\mathbb C[E_2,E_4,E_6]$. We have the following 
\begin{proposition}\label{propM*closed}
The ring $\widetilde{M_*}(\Gamma_1)$ is closed under differentiation. Specifically, we have
\begin{equation}\begin{aligned}\label{eisensteinderivative}
E_2'&=\frac{2\pi\im}{12}(E_2^2-E_4), \\
E_4'&=\frac{2\pi\im}{3}(E_2 E_4-E_6), \\
E_6'&=\frac{2\pi\im}{2}(E_2 E_6-E_4^2).
\end{aligned}\end{equation}
\end{proposition}

In the main text, we use the property that $E_2(\tn)$ is never zero for any $\tn$ in the $\slz$ orbit of an $\slz$-elliptic point (see also \cite{BASRAOUI}):

\begin{lemma}\label{lemmaE2} Let $\alpha=e^{2\pi i /3}$. The second Eisenstein series $E_2(\gamma i)\neq 0$ and $E_2(\gamma \alpha) \neq 0$ for all $\gamma\in \psl$.
\end{lemma}

\begin{proof}
Consider the transformation property 
\be 
\label{E2trafo}
E_2(\gamma\tau) =(c\tau+d)^2E_2(\tau)-\frac{6\im}{\pi}c(c\tau+d)
\ee
for $\gamma=\begin{psmallmatrix}a&b\\c&d\end{psmallmatrix}\in\slz$.
From $E_2(i)=\frac3\pi$ and $E_2(\alpha)=\frac{2\sqrt3}{\pi}$ we can easily compute
\begin{equation}\begin{aligned}\label{e2elliptic}
    E_2(\gamma i)&=E_2(i) |ci+d|^2, \\
E_2(\gamma\alpha)&=E_2(\alpha)|c\alpha+d|^2.
    \end{aligned}
\end{equation}
None of the factors on the rhs are zero, since $c,d\in\mathbb Z$.
\end{proof}

\subsubsection*{Elliptic Jacobi theta functions}

We use also the elliptic Jacobi theta functions  $\jt_i(\tau,v)$, $\vartheta_j:\mathbb{H}\times
\mathbb{C}\to \mathbb{C}$, $j=1,\dots,4$, which we define as 
\be\label{jt_elliptic}
\begin{split}
&\vartheta_1(\tau,v)=i \sum_{r\in
  \mathbb{Z}+\frac12}(-1)^{r-\frac12}q^{r^2/2}e^{2\pi i
  rv}, \\
&\vartheta_2(\tau,v)= \sum_{r\in
  \mathbb{Z}+\frac12}q^{r^2/2}e^{2\pi i
  rv},\\
&\vartheta_3(\tau,v)= \sum_{n\in
  \mathbb{Z}}q^{n^2/2}e^{2\pi i
  n v},\\
&\vartheta_4(\tau,v)= \sum_{n\in 
  \mathbb{Z}} (-1)^nq^{n^2/2}e^{2\pi i
  n v}. 
\end{split}
\ee
The Jacobi theta functions defined in \eqref{Jacobitheta} of Part I are $\jt_j(\tau)=\jt_j(\tau,0)$ for $j=2,3,4$.
The $S$-transformation reads
\bea\label{jt_S-transformation}
\jt_1(-1/\tau,v/\tau)&=-i\sqrt{-i \tau}e^{\pi i z/\tau^2}\jt_1(\tau,v), \\
\jt_2(-1/\tau,v/\tau)&=\sqrt{-i \tau}e^{\pi i z/\tau^2}\jt_4(\tau,v),\\
\jt_3(-1/\tau,v/\tau)&=\sqrt{-i \tau}e^{\pi i z/\tau^2}\jt_3(\tau,v),\\
\jt_4(-1/\tau,v/\tau)&=\sqrt{-i \tau}e^{\pi i z/\tau^2}\jt_2(\tau,v).
\eea
The periodicity in $v$ is 
\bea
\jt_1(\tau,v+1)&=-\jt_1(\tau,v),\\
\jt_2(\tau,v+1)&=-\jt_2(\tau,v),\\
\jt_3(\tau,v+1)&=+\jt_3(\tau,v),\\
\jt_4(\tau,v+1)&=+\jt_4(\tau,v),
\eea
while the map $v\mapsto v+\tau$ sends
\bea
\jt_1(\tau,v+\tau)&=-e^{-\pi i(\tau+2v)}\jt_1(\tau,v),\\
\jt_2(\tau,v+\tau)&=+e^{-\pi i(\tau+2v)}\jt_2(\tau,v),\\
\jt_3(\tau,v+\tau)&=+e^{-\pi i(\tau+2v)}\jt_3(\tau,v),\\
\jt_4(\tau,v+\tau)&=-e^{-\pi i(\tau+2v)}\jt_4(\tau,v).
\eea
Their zeros are 
\be\label{jt_zeros}
\begin{split}
\jt_1(\tau,m+n\tau)&=0, \\
\jt_2(\tau,m+\tfrac12+n \tau)&=0, \\
\jt_3(\tau,m+\tfrac12+(n+\tfrac12)\tau)&=0, \\
\jt_4(\tau,m+(n+\tfrac12)\tau)&=0, \\
\end{split}
\ee
with $m,n\in\mathbb Z$. The elliptic Jacobi theta functions satisfy various generalisations of the abstruse identity \eqref{jacobiabstruseidentity}. We use the following identity,
\begin{equation}
\label{jt_elliptic_abstruse}
\jt_2(\tau,0)^2\jt_2(\tau,v)^2+\jt_4(\tau,0)^2\jt_4(\tau,v)^2=\jt_3(\tau,0)^2\jt_3(\tau,v)^2.
\end{equation}

\subsection{Mock modular forms}\label{app:mock}
In this Appendix, we introduce various mock modular forms and their relations that are useful for the evaluation of $u$-plane integrals. For more comprehensive treatments, we refer the reader to the available literature. See for instance \cite{Larry,ZwegersThesis,Dabholkar:2012nd,zagier2009}.
In the context of elliptic AD points, we are also interested in evaluating their modular completions at the elliptic fixed points in the fundamental domain.

\subsubsection*{The mock modular forms $F$ and $H$}

Following \cite{ZwegersThesis}, we define the Appell--Lerch sum $M(\tau,u,v)$ as\footnote{See also \cite{Korpas:2019cwg} for more details.} 
\begin{equation}\label{Appell--Lerch-sum}
M(\tau,u,v)=\frac{e^{\pi \im u}}{\vartheta_1(\tau,v)}\sum_{n\in\BZ}\frac{(-1)^nq^{n(n+1)/2}e^{2\pi \im nv}}{1-e^{2\pi \im u}q^n}.
\end{equation}
The important property of the Appell--Lerch function is that if we add the function 
\begin{equation}\label{R-function}
R(\tau,\bar\tau,u,\bar u)=\sum_{n\in\BZ+\frac{1}{2}}\left(\sgn(n)-\text{Erf}\left((n+a)\sqrt{2\pi y}\right) \right)(-1)^{n-\frac{1}{2}}e^{-2\pi \im un}q^{-n^2/2},
\end{equation}
the completion
\begin{equation}\label{hatF_def}
\begin{split}
\widehat{F}_\mu(\tau,\bar\tau,\rho,\bar\rho)&=-\im e^{\pi \im \nu}q^{-\nu^2/2}w^{-\nu}\left(M(\tau,\rho+\mu\tau,\tfrac{1}{2}\tau)\right.\\
&\quad \left.+\tfrac{\im}{2}R(\tau,\bar\tau,\rho+\nu\tau,\bar\rho+\nu\bar\tau)\right),\\
\end{split}
\end{equation}
where $\nu=\mu-\frac{1}{2}$, $w=e^{2\pi i \rho}$, 
 transforms as a Jacobi form.

We study in the main text the mock modular form 
\begin{equation}
    F(\tau)=-\frac{1}{\jt_4(\tau)}\sum_{n\in\mathbb Z}\frac{(-1)^nq^{\frac{n^2}{2}-\frac 18}}{1-q^{n-\frac12}} 
\end{equation}
and its non-holomorphic completion 
\begin{equation}\label{Ffhat}
    \widehat F(\tau,\bar\tau)=F(\tau)-\frac i2\int_{-\bar\tau}^{i\infty}\frac{\eta(w)^3}{\sqrt{-i(w+\tau)}}\mathrm dw,
\end{equation}
which correspond to \eqref{hatF_def} with $\mu=\frac12$.
The completion $\widehat F$ transforms as a non-holomorphic modular form of weight $\left(\tfrac12,0\right)$ for $\Gamma^0(2)$,
\begin{equation}
   \begin{aligned}
        \widehat F(\tau+2,\bar\tau+2)&=-i\widehat F(\tau,\bar\tau), \\
        \widehat F(\tfrac{\tau}{\tau+1},\tfrac{\bar\tau}{\bar\tau+1})&=e^{\frac{\pi i}{4}}\sqrt{\tau+1}\, \widehat F(\tau,\bar\tau).
    \end{aligned}
\end{equation}

Since $\widehat F$ does not transform under the full $\slz$, it is useful to consider similar (mock) modular forms which do.
The holomorphic function $F$ is related to the $q$-series $H^{(2)}\eqqcolon H$ of Mathieu moonshine \cite{Eguchi:2010ej} in a simple way. We have \cite{Dabholkar:2012nd,Korpas:2019cwg}\footnote{Typo in \cite{Dabholkar:2012nd}}
\begin{equation}
\begin{aligned}
    H(\tau)&=2\frac{\jt_2(\tau)^4-\jt_4(\tau)^4}{\eta(\tau)^3}-\frac{24}{\jt_3(\tau)}\sum_{n\in\mathbb Z}\frac{q^{\frac{n^2}{2}-\frac18}}{1+q^{n-\frac12}}\\
    &=2q^{-\frac18}\left(-1+45q+231q^2+770q^3+2277q^4+\dots\right).
    \end{aligned}
\end{equation}
The $q$-series is the   OEIS sequence A169717, and the Fourier coefficients are sums of dimensions of irreducible representations of the sporadic group $M_{24}$.
This function has been studied in detail in \cite{EGUCHI1988125,Eguchi:1988vra,Eguchi:2004yi,Cheng:2010pq,Eguchi:2010fg} and appears in the elliptic genus of the $K3$ sigma model with $\CN=(4,4)$ supersymmetry. 
It relates to $F$ as \cite{Korpas:2019cwg}
\begin{equation}\label{FH_relation}
    F(\tau)=\frac{1}{24}\left( H(\tau)+2\frac{\jt_2(\tau)^4+\jt_3(\tau)^4}{\eta(\tau)^3}\right),
\end{equation}
which gives an alternative representation for $F$,
\begin{equation}
F(\tau)=\frac{\jt_2(\tau)^4}{4\eta(\tau)^3}-\frac{1}{\jt_3(\tau)}\sum_{n\in\mathbb Z}\frac{q^{\frac{n^2}{2}-\frac18}}{1+q^{n-\frac12}}.
\end{equation}
The completion 
\begin{equation}\label{HHhat}
   \widehat H(\tau,\bar\tau)= H(\tau)-12i\int_{-\bar\tau}^{i\infty}\frac{\eta(w)^3}{\sqrt{-i(w+\tau)}}\mathrm dw
\end{equation}
transforms as a non-holomorphic modular form   for $\slz$ of weight $(\tfrac12,0)$,
\begin{equation}\label{Hhat_mod_transf}
    \begin{aligned}
    \widehat H(\tau+1,\bar\tau+1)&=e^{-\frac{\pi i}{4}}\widehat H(\tau,\bar\tau), \\
    \widehat H(-1/\tau,-1/\bar\tau)&=-\sqrt{-i\tau}\widehat H(\tau,\bar\tau).
    \end{aligned}
\end{equation}
Using modular transformations which fix the $\slz$ elliptic points $\tau=i$ and $\tau=\alpha$, 
 we can readily evaluate
\begin{equation}\begin{aligned}\label{hatHelliptic}
  \widehat H(i,\bar i)&=0, \\
 \widehat H(\alpha,\bar \alpha)&=0.
\end{aligned}\end{equation}
The completions $\widehat H$ and $\widehat F$ are of course not unrelated. By comparing \eqref{Ffhat} with \eqref{HHhat} and using \eqref{FH_relation}, we find
\begin{equation}\label{FhatHhat}
    \widehat F(\tau)=\frac{1}{24}\left( \widehat H(\tau)+2\frac{\jt_2(\tau)^4+\jt_3(\tau)^4}{\eta(\tau)^3}\right).
\end{equation}
This is precisely the modular completion of \eqref{FH_relation}. Since the correction 
\begin{equation}\label{r_hmf}
r(\tau)=2\frac{\jt_2(\tau)^4+\jt_3(\tau)^4}{\eta(\tau)^3}
\end{equation}
is already modular, its modular completion vanishes. It is in fact a modular form for $\Gamma^0(2)$ and transforms as 
\begin{equation}
   \begin{aligned}
       r(\tau+2)&=-i\, r(\tau), \\
        r(\tfrac{\tau}{\tau+1})&=e^{\frac{\pi i}{4}}\sqrt{\tau+1}\, r(\tau).
    \end{aligned}
\end{equation}

Another mock modular form with shadow $\eta^3$ is the function $Q^+$  introduced by Malmendier and Ono \cite{Malmendier:2012zz,Malmendier:2008db},
\begin{equation}\begin{aligned}
    Q^+(\tau)&=\frac{1}{12}H(\tau)+\frac76\frac{\jt_2(\tau)^4+\jt_3(\tau)^4}{\eta(\tau)^3}\\
    &=\frac{3\jt_2(\tau)^4+2\jt_3(\tau)^4}{2\eta(\tau)^3}-\frac{2}{\jt_3(\tau)}\sum_{n\in\mathbb Z}\frac{q^{\frac{n^2}{2}-\frac18}}{1+q^{n-\frac12}}\\
    &=q^{-\frac18}(1 + 28 q^{\frac12} + 39 q + 196 q^{\frac32} + 161 q^2+756 q^{\frac52} +\dots).\\
\end{aligned}\end{equation}
Its non-holomorphic completion
\begin{equation}
    \widehat{Q^+}(\tau,\bar\tau)= Q^+(\tau)-i\int_{-\bar\tau}^{i\infty}\frac{\eta(w)^3}{\sqrt{-i(w+\tau)}}\mathrm dw
\end{equation}
evaluates to
\begin{equation}\begin{aligned}\label{Qpelliptic}
      \widehat{Q^+}(i,\bar i)&=7\Oi^{\frac12}, \\
\widehat{Q^+}(\alpha,\bar \alpha)&=\frac{7\sqrt[3]{2}}{\sqrt3}\Oa^{\frac12}, \\
\widehat{Q^+}(1+i,\overline{1+i})&=0,
\end{aligned}\end{equation}
where $\Oi$ and $\Oa$ are the Chowla--Selberg periods \eqref{CSperiodsOaOi}, which we introduced in Subsection \ref{chowla_selberg}.
It is related to $F$ by 
\begin{equation}
    Q^+(\tau)=2F(\tau)+\frac{\jt_2(\tau)^4+\jt_3(\tau)^4}{\eta(\tau)^3}.
\end{equation}

\subsection*{Laurent series of $H$}
The Chowla--Selberg  formula (Proposition \ref{IQF_value}) seems to be valid for the three completions $\widehat F$, $\widehat H$ and $\widehat{Q^+}$ of mock modular forms of weight $\frac12$. See equations  \eqref{Hhat_elliptic_CS}, \eqref{Fhat_values} and \eqref{Qpelliptic}.
We can also use it to compute non-holomorphic Laurent/Taylor series around special points. Regarding the $u$-plane integrand, we are interested in expansions
\be\label{exp_non_holo}
 f(\tau,\bar \tau)=\sum_{m\gg -\infty,n\geq 0} \frac{1}{m!n!}d_{\tau_0}^f(m,n)\,(\tau-\tn)^m\,(\bar \tau-\bar \tn)^n
\ee
around some point $(\tn,\bar\tn)$.
It is clear that if $f$ does not have poles in $\mathbb H$, then the series is only over $m\geq 0$. Furthermore, the Taylor coefficients are given by 
\begin{equation}
d_{\tau_0}^f(m,n)=\frac{\partial^{m+n}f(\tau,\bar\tau)}{\partial \tau^m \partial \bar\tau^n}\Big|_{\substack{\,\,\tau=\tn \\ \bar\tau=\bar\tn}}.
\end{equation}
For $\widehat H$, we can compute some of these coefficients, which we give in Table \ref{H2hcoefficients}.
\begin{table}[ht]\begin{center}
\begin{tabular}{|c|Sc|Sc|Sc|}
\hline 
\diagbox[width=1cm, height=1cm]{$m$}{$n$}&0&1&2 \\ 
\hline
0& 0&$-6\sqrt2i$ &$-6\sqrt2$\\
1& $\sim 6.22669i$&$\frac{3}{\sqrt2}$&$-\frac{9i}{2\sqrt2}$\\
2&$\sim-9.34004$&$\frac{9i}{4\sqrt2}$&$\frac{9}{2\sqrt2}$ \\
 \hline
\end{tabular}
\caption{The coefficients $d_i^{\widehat H}(m,n)\Oi^{-\frac32}$ of the completion $\widehat H$. The coefficients for $m\geq1$ and $n=0$ do not seem to be algebraic.} \label{H2hcoefficients}\end{center}
\end{table}

For some coefficients, we find closed formulas,
\begin{equation}\begin{aligned}\label{H2coef}
    d_i^{\widehat H}(m,1)&=-3i^{m+1}2^{\frac32-2m}(2m-1)!!   \, \Oi^{\frac32},\\
    d_i^{\widehat H}(m,2)&=-3 (-i)^{-m} 2^{\frac{1}{2}-2 m}(2 m-1)!! (m+2)   \, \Oi^{\frac32},
\end{aligned}\end{equation}
where $\Oi$ is the Chowla--Selberg period \eqref{elliptic_periods}, and $m\in\mathbb N_0$.
This is possible due to the fact that $\partial_{\bar\tau}\widehat H$ is proportional to $\overline{\eta(\tau)}^3/\sqrt{y}$, whose derivatives are straightforward to compute.

Similarly, for the expansion around $(\tau,\bar\tau)=(\alpha,\bar\alpha)$ we find
\begin{equation}\begin{aligned}\label{H2_alpha_coef}
    d_\alpha^{\widehat H}(m,1)&=2^{2-m} 3^{\frac{3}{4}-\frac{m}{2}} e^{\frac{1}{8} i \pi  (4 m-3)}
    (2 m-1)!!\, \Oa^{\frac32}, \\
   d_\alpha^{\widehat H}(m,2)&= 2^{2-m} 3^{\frac{1}{4}-\frac{m}{2}} e^{\frac{1}{8} i \pi  (4 m+9)}  (2 m-1)!!(m+2)\Oa^{\frac32}.
\end{aligned}\end{equation}
The coefficients for $n=1$ and $n=2$ for both $i$ and $\alpha$ are simply related by 
\begin{equation}\begin{aligned}
    d_i^{\widehat H}(m,2)&=\,\,-\frac i2\,\,(m+2) d_i^{\widehat H}(m,1), \\  d_\alpha^{\widehat H}(m,2)&=-\frac{i}{\sqrt3}(m+2)d_\alpha^{\widehat H}(m,1).
\end{aligned}\end{equation}
From Table \ref{H2hcoefficients} it is clear that the smallest value of $m$ in the expansion \eqref{exp_non_holo} of $\widehat H$ at $\tau=i$ is $m=1$. At $\tau=\alpha$, the same is true, as is clear from the nonzero constant $d_\alpha^{\widehat H}(1,1)$ in \eqref{H2_alpha_coef}. However, we have 
\bea
     d_\alpha^{\widehat H}(0,0)&=0 \\
     d_\alpha^{\widehat H}(1,0)&=0, \\
     d_\alpha^{\widehat H}(2,0)&\sim 22.8072\, e^{\frac{\pi i}{8}}.
\eea

\subsection{Mock Jacobi forms}\label{app:mockjacobi}
We can generalise the function $H=H^{(2)}$ to include an elliptic argument. For this, we need to find a suitable elliptic generalisation of the holomorphic modular form \eqref{r_hmf} which is the difference of $\widehat F$ and $\widehat H$. Recall the elliptic generalisation
\begin{equation}\label{Ftaurho}
    F(\tau,\rho)=-\frac{w^{\frac{1}{2}}}{\vartheta_4(\tau)}\sum_{n\in\BZ}\frac{(-1)^nq^{n^2/2-\frac{1}{8}}}{1-wq^{n-\frac{1}{2}}}.
\end{equation}
We define the dual modular form 
\begin{equation}
    r_D(\tau)\coloneqq -(-i\tau)^{-\frac12}  r(-1/\tau)=-\frac{\jt_4(\tau)^4+\jt_3(\tau)^4}{\eta(\tau)^3}.
\end{equation}
Then using the Jacobi identity \eqref{jacobiabstruseidentity}, we have
\begin{equation}
    r(\tau)-r_D(\tau)=6\frac{\jt_3(\tau)^4}{\eta(\tau)^3}.
\end{equation}
The difference between $\widehat F(\tau,\bar\tau)$ and $\widehat F_D(\tau,\bar\tau)$ is a weakly holomorphic modular form of weight $\frac12$, since $\widehat F$ and $\widehat F_D$ have the same non-holomorphic part. This can be seen from \eqref{FhatHhat} and the fact that $\widehat H$ is $\slz$-invariant.  This gives 
\begin{equation} \label{FminusF_D}   
\widehat F(\tau,\bar\tau)-\widehat F_D(\tau,\bar\tau)=\frac{\jt_3(\tau)^4}{4\eta(\tau)^3}=\frac{\jt_3(\tau)^3}{2\jt_2(\tau)\jt_4(\tau)}.
\end{equation}
Including the elliptic variables, we find 
\begin{equation}    \label{Fhat-FDhat}
\widehat F(\tau,\bar \tau,\rho,\bar \rho)-\widehat F_D(\tau,\bar \tau,\rho,\bar \rho)=\frac{\jt_3(\tau)^2\jt_3(\tau,\rho)}{2\jt_2(\tau,\rho)\jt_4(\tau,\rho)}.
\end{equation}
This again follows from the fact that both $\widehat F(\tau,\bar \tau,\rho,\bar \rho)$ and $F_D(\tau,\bar \tau,\rho,\bar \rho)$ have the same non-holomorphic part, and $F(\tau,\rho)-F_D(\tau,\rho)$ must transform as a Jacobi form of weight $\frac12$ and for $\rho=0$ give back \eqref{FminusF_D}. The holomorphic functions $F(\tau,\rho)$ and $F_D(\tau,\rho)$ are given in \eqref{f12function} and \eqref{f12functionD}.
The rhs of \eqref{Fhat-FDhat} is a specialisation of the Appell--Lerch sum $M(\tau,u,v)$ \eqref{Appell--Lerch-sum}, which is another way to find it.

We are thus aiming for a function $r(\tau,\rho)$ such that
\begin{equation}
    r(\tau,\rho)-r_D(\tau,\rho)=12\frac{\jt_3(\tau)^2\jt_3(\tau,\rho)}{\jt_2(\tau,\rho)\jt_4(\tau,\rho)}
\end{equation}
and such that $r(\tau,0)$ agrees with \eqref{r_hmf}. This function is given by elliptically extending the arguments of some of the Jacobi theta functions of $r(\tau)$,
\begin{equation}\label{rtaurho}
r(\tau,\rho)=4\frac{\vartheta_2(\tau)^2\vartheta_2(\tau,\rho)^2+\vartheta_3(\tau)^2\vartheta_3(\tau,\rho)^2}{\vartheta_2(\tau,\rho)\vartheta_3(\tau,\rho)\vartheta_4(\tau,\rho)}.
\end{equation}
A short calculation using the elliptic $S$-transformation \eqref{jt_S-transformation} and the elliptic abstruse identity \eqref{jt_elliptic_abstruse} confirms this. Then we know that 
\begin{equation}\label{hatHdef}
    \widehat H(\tau,\bar\tau,\rho,\bar\rho)=24\widehat F(\tau,\bar\tau,\rho,\bar\rho)-r(\tau,\rho),
\end{equation} 
where the holomorphic part reads
\begin{equation}\label{H12Def}
    \begin{aligned}
    H(\tau,\rho)=&\, 
 \,  24\, F(\tau,\rho)-4\frac{\vartheta_2(\tau)^2\vartheta_2(\tau,\rho)^2+\vartheta_3(\tau)^2\vartheta_3(\tau,\rho)^2}{\vartheta_2(\tau,\rho)\vartheta_3(\tau,\rho)\vartheta_4(\tau,\rho)}\\
        =&\, -\frac{4\,w^{1/2}}{(1+w)}q^{-1/8}+\frac{4(4+11w+15w^2+11w^3+4w^4)}{w^{3/2}(1+w)}q^{7/8}+\CO(q^{15/8}),
    \end{aligned}
\end{equation}
with $w=e^{2\pi i\rho}$. As far as checked, the $w$-dependent higher order terms in the $q$-expansion  can be expressed as positive linear combinations of $\mathfrak{su}(2)$ characters. In particular, the polynomials are palindromic. Their coefficients appear to asymptote to those of the series 
\be \begin{aligned}
(1 + w) \prod_{k=1}^\infty \frac{1 + w^k}{1 - w^k}&=\frac{(1+w)(-1;w)_\infty}{2(w;w)_\infty} \\
&=1+3w+6w^2+12w^3+22w^4+38w^5+64w^6+\CO(w^7).
\end{aligned}\ee 
It would be interesting to explore if the coefficients of this refined generating function also have an interpretation within Mathieu Moonshine. The analysis in \cite{MARGOLIN1993370, Taormina:2013mda, Taormina:2013jza} could be a possible avenue for this direction.

We have the following quasi-periodicity properties of the three elliptic functions $F$, $H$ and $r$,
\bea\label{quasi-periodicity}
F(\tau,\rho+1)&=-F(\tau,\rho),\\
    H(\tau,\rho+1)&=-H(\tau,\rho), \\
    r(\tau,\rho+1)&=-r(\tau,\rho),\\
   F(\tau,\rho+\tau)&=-e^{\pi i\tau}e^{2\pi i\rho}F(\tau,\rho)-e^{\frac34\pi i\tau}e^{\pi i\rho}, \\
     H(\tau,\rho+\tau)&=-e^{\pi i\tau}e^{2\pi i\rho}H(\tau,\rho)-24\, e^{\frac34\pi i\tau}e^{\pi i\rho}, \\
     r(\tau,\rho+\tau)&=-e^{\pi i \tau}e^{2\pi i \rho}r(\tau,\rho). \\
     \eea
The mock Jacobi forms $F$ and $H$ are symmetric in $\rho$, which gives a  palindromic property of their Fourier coefficients,
     \bea
     F(\tau,-\rho)&=F(\tau,\rho), \\
     H(\tau,-\rho)&=H(\tau,\rho).
     \eea

\subsubsection*{Poles of $F$ and $H$}
Since the functions $F$ and $H$ are quasi-periodic on the lattice  $\mathbb Z\oplus \tau\mathbb Z$, it is sufficient to study its poles in the fundamental domain of the torus $\mathbb C/\mathbb Z\oplus \tau\mathbb Z$.

The function $H(\tau,\rho)$ has  (possible) poles at 
\begin{equation}
    \rho=\frac{m}{2}+\frac{n}{2}\tau,\quad m,n\in\BZ,\quad (m,n)\neq(0,0),
\end{equation}
which come from the zeros \eqref{jt_zeros} of the Jacobi theta functions in the denominator of $r(\tau,\rho)$ (here, we can already exclude the cases where both $m$ and $n$ are even). Meanwhile,  $F(\tau,\rho)$ has poles at $\rho=(n+\tfrac12)\tau$, which follows from the definition \eqref{Ftaurho}. Restricting to the fundamental domain, the complete pole structure of both $F$ and $H$ are then $\rho\in\{0,\frac12,\frac\tau2,\frac12+\frac\tau2\}$.
Let us study  which of these elliptic values are poles of $r$, $F$ and $H$.

The zeros of the Jacobi theta functions $\jt_i$ are listed in \eqref{jt_zeros}. By construction, the functions $r$, $F$ and $H$ have a smooth limit $\rho\to 0$ and reduce to holomorphic functions with the same symbol. Thus by \eqref{quasi-periodicity}, there are no poles on the lattice $\mathbb Z\oplus \tau\mathbb Z$. 

The potential poles of $r$, $F$ and $H$ are thus on the half-lattice points. 
Consider first $\rho=\frac12$. 
We find 
\begin{equation}\label{F_exp_rho_12}
    F(\tau,\rho)=-2\pi m(\tau)(\rho-\tfrac12)+\CO((\rho-\tfrac12)^3),
\end{equation}
with a nonzero holomorphic function $m$.

Out of the four Jacobi theta functions, only $\jt_2$ has a zero at $\rho=\frac12$. We easily calculate the Taylor expansions,
\bea
\jt_2(\tau,\rho)&=-2\pi\eta(\tau)^3(\rho-\tfrac12)+\CO\left((\rho-\tfrac12)^3\right), \\
\jt_3(\tau,\rho)&=\jt_4(\tau)+\CO\left((\rho-\tfrac12)^2\right), \\
\jt_4(\tau,\rho)&=\jt_3(\tau)+\CO\left((\rho-\tfrac12)^2\right).
\eea
From \eqref{rtaurho} we then find
\begin{equation}
    r(\tau,\rho)=-\frac{4}{\pi \jt_2(\tau)}(\rho-\tfrac12)^{-1}+\CO\left((\rho-\tfrac12)^1\right).
\end{equation}
Combining with \eqref{H12Def} and \eqref{F_exp_rho_12} we thus find that $H(\tau,\rho)$ has a simple pole at $\rho=\tfrac12$ with residue
\begin{equation}
    \underset{\rho=\frac12}{\Res} \, H(\tau,\rho)=\frac{4}{\pi \jt_2(\tau)}.
\end{equation}
At $\rho=\tau$, both $F(\tau,\tau)$ and $H(\tau,\tau)$ are regular. At $\rho=\frac\tau2$, $F$, $r$ and $H$ have a pole. We find
\bea\label{F_pole_tau2}
F(\tau,\rho)=\frac{q^{\frac18}}{2\pi i\jt_4(\tau)}(\rho-\tfrac\tau2)^{-1}+\CO\left((\rho-\tfrac\tau2)^0\right)
\eea
Regarding $r(\tau,\rho)$, we calculate
\bea
\jt_2(\tau,\rho)&=q^{-\frac18}\jt_3(\tau)+\CO\left((\rho-\tfrac\tau2)^1\right), \\
\jt_3(\tau,\rho)&=q^{-\frac18}\jt_2(\tau)+\CO\left((\rho-\tfrac\tau2)^1\right), \\
\jt_4(\tau,\rho)&=2\pi iq^{-\frac18}\eta(\tau)^3(\rho-\tfrac\tau2)+\CO\left((\rho-\tfrac\tau2)^2\right).
\eea
This gives
\begin{equation}\label{r_pole_tau2}
    r(\tau,\rho)=-\frac{8iq^{\frac18}}{\pi \jt_4(\tau)}(\rho-\tfrac\tau2)^{-1}+\CO(1).
\end{equation}
Combining \eqref{F_pole_tau2} and \eqref{r_pole_tau2}, we find
\begin{equation}
    \underset{\rho=\frac\tau2}{\Res} \, H(\tau,\rho)=\frac{4q^{\frac18}}{\pi i \jt_4(\tau)}.
\end{equation}
Finally, consider $\rho_0=\frac12+\frac\tau2$. We find
\begin{equation}
    F(\tau,\rho_0)=\frac{q^{\frac18}}{2i}.
\end{equation}
The Jacobi theta functions have leading terms
\bea
\jt_2(\tau,\rho)&=-iq^{-\frac18}\jt_4(\tau)+\CO\left((\rho-\rho_0)^1\right), \\
\jt_3(\tau,\rho)&=2\pi i q^{-\frac18}\eta(\tau)^3(\rho-\rho_0)+\CO\left((\rho-\rho_0)^2\right), \\
\jt_4(\tau,\rho)&=q^{-\frac18}\jt_2(\tau)+\CO\left((\rho-\rho_0)^1\right).
\eea
This gives 
\begin{equation}
    r(\tau,\rho)=-\frac{4q^{\frac18}}{\pi\jt_3(\tau)}(\rho-\rho_0)^{-1}+\CO(1),
\end{equation}
and therefore 
\begin{equation}
    \underset{\rho=\frac12+\frac\tau2}{\Res} \, H(\tau,\rho)=\frac{4q^{\frac18}}{\pi \jt_3(\tau)}.
\end{equation}

In Table \ref{elliptic_residues}  we summarise the residues of $r$, $F$ and $H$ at all possible poles of these functions. 
In particular, the poles of $H$ are
\bea
    \underset{\rho\in\mathbb C}{\text{poles}}\, H(\tau,\rho)&=\{\tfrac m2+\tfrac n2\tau\}\, \backslash \,\mathbb Z\oplus \tau\mathbb Z \\
    &=\{\tfrac m2+\tfrac n2\tau \, | \, (\tfrac m2,\tfrac n2)\not\in\mathbb Z^2\}.
\eea

\begin{table}[ht]\begin{center}
\renewcommand{\arraystretch}{2}
\begin{tabular}{|Sc|Sc|Sc|Sc|}
\hline 
$\rho_0$ & $\underset{\rho=\rho_0}{\Res}\, F(\tau,\rho)$ & $\underset{\rho=\rho_0}{\Res}\, H(\tau,\rho)$ &$\underset{\rho=\rho_0}{\Res}\, r(\tau,\rho)$\\ \hline \vspace{-2ex}
$\displaystyle0$&0&0 &0\\ 
$\displaystyle\frac12$ & 0 & $\displaystyle\frac{4}{\pi \jt_2(\tau)}$ & $\displaystyle-\frac{4}{\pi \jt_2(\tau)}$  \\  
$\displaystyle\frac\tau2$& $\displaystyle\frac{q^{\frac18}}{2\pi i\jt_4(\tau)}$ & $\displaystyle\frac{4q^{\frac18}}{\pi i\jt_4(\tau)}$ & $\displaystyle\frac{8q^{\frac18}}{\pi i\jt_4(\tau)}$ \\  
$\displaystyle\frac12+\frac\tau2$ & 0 & $\displaystyle\frac{4q^{\frac18}}{\pi \jt_3(\tau)}$ & $\displaystyle-\frac{4q^{\frac18}}{\pi \jt_3(\tau)}$
 \\  \hline
\end{tabular}
\caption{The pole structure of the functions $F$ \eqref{Ftaurho}, $H$ \eqref{H12Def} and $r$ \eqref{rtaurho} including residues at the half-lattice points. } \label{elliptic_residues}\end{center}
\end{table}

Given the precise pole structure of $H$, we can study functions related to $H$ which have no poles for real $\rho$. To this end, we can add to $H(\tau,\rho)$ a meromorphic Jacobi form which cancels those poles of $H(\tau,\rho)$.\footnote{See \cite{Dabholkar:2012nd} for a systematic study of poles of Jacobi forms.} In this way, we arrive at
\begin{equation}
\label{wtHtaurho}
    \widetilde H(\tau,\rho)=H(\tau,\rho) - \frac{8 \eta(\tau)^3}{\theta_1(\tau,\rho)^2} \sum_{j=2,3,4} \left(\frac{\theta_j(\tau,2\rho)}{\theta_j(\tau,\rho)} - \frac{\theta_j(\tau,\rho)^3}{\theta_j(\tau)^3}\right),
\end{equation}
which does not have a pole at either $\rho=0$ or $\rho=\frac12$. Indeed, $\widetilde H(\tau,0)=H(\tau)$ and the Taylor series at $\rho=\frac12$ reads
\begin{equation}
    \widetilde H(\tau,\rho)=4\jt_2(\tau)^3\frac{\jt_3(\tau)^2-\jt_4(\tau)^2}{\jt_3(\tau)^2\jt_4(\tau)^2}+\CO\left((\rho-\tfrac12)^2\right).
\end{equation}
While $\widetilde H$ does not have poles at $\rho=0$ or $\rho=\frac12$, it has poles at $\rho=\frac\tau2$ and $\rho=\frac12+\frac\tau2$. $\widetilde H$ is periodic under $\rho\to \rho+2$. The first few terms of the $q$-expansion of $\widetilde H$ are
\be 
\begin{split} 
&\widetilde H(\tau,\rho)=q^{-1/8}(-w^{1/2}-w^{-1/2}) \\
&\quad +(-16 w^{2}+19w^{3/2}-64
w+26w^{1/2}+160 + \text{palindromic terms})\, q^{7/8}+\dots.
\end{split}
\ee 
Similarly to the refinement $H(\tau,\rho)$ of $H(\tau)$, it is an interesting question whether this refinement $\widetilde H(\tau,\rho)$  has an interpretation in the context of Mathieu moonshine.

\section{Mass expansions}\label{app:mass_expansions}
In this Appendix, we list various expansions of Coulomb branch functions required for the computation of topological correlation functions. In Subsection \ref{app:vc}, we comment on the calculation of the couplings $v_j$ and $w_{jk}$ to the background fluxes. In Appendix \ref{app:large_masses}, we give explicit series of the order parameter and the period $\frac{du}{da}$ in a large mass expansion, which are relevant particularly in Section \ref{sec:large_mass_P2}. 
Due to the modular structure, the equal mass $N_f=2$ theory is the simplest massive case where many results can be obtained exactly. We study this configuration in detail in subsection \ref{sec:equal_mass_expansion}.

\subsection{Couplings to background fluxes}\label{app:vc}
The background couplings  $v_j$ and $w_{jk}$, defined in \eqref{Defvw}, can be calculated  either directly from  the prepotential \cite{Ohta_1997,Ohta:1996hq}, or using  the hypergeometric representation of the periods. For the latter case, from \cite{Masuda:1996xj,Brandhuber:1996ng} we have 
\begin{equation}\label{periods_hyperg}
	\begin{aligned}
		\frac{\partial a}{\partial u} =\,& \frac{1}{2\cdot 3^{1/4}g_2^{1/4}}\,{}_2F_1\left[\frac{1}{12},\frac{5}{12},1;\frac{12^3}{\CJ}\right], \\
		\frac{\partial a_D}{\partial u}=\,&\frac{i}{2\cdot 3^{1/4}g_2^{1/4}}\left(3\log(12)\,{}_2F_1\left[\frac{1}{12},\frac{5}{12},1;\frac{12^3}{\CJ}\right]-F^*\left[\frac{1}{12},\frac{5}{12},1;\frac{12^3}{\CJ}\right]\right),
	\end{aligned}
\end{equation}
where
\begin{equation}
	F^*\left[\alpha,\beta,1;z\right]={}_2F_1\left[\alpha,\beta,1;z\right]\log z+\sum_{n=0}^\infty\frac{(\alpha)_n(\beta)_n}{(n!)^2}z^n\sum_{r=0}^{n-1}\frac{1}{\alpha+r}+\frac{1}{\beta+r}-\frac{2}{r+1}.
\end{equation}
Here, $\CJ=12^3g_2^3/(g_2^3-27g_3^2)$ is the $\CJ$-invariant of the curve. 

The periods can be in principle expanded around any point on the Coulomb branch, and by integration and comparison with the prepotential \eqref{prepotential} they can be found  by differentiation with respect to the masses $m_j$. In practise, for generic masses this is straightforward only for the weak coupling limit \eqref{decouplingConv}.
In the following, we give explicit expressions at weak coupling for $N_f=1$ and $N_f=2$.

\subsubsection*{$\boldsymbol{N_f=1}$}

For $N_f=1$, we find in this way the couplings $v$ and $C$, 
\begin{equation}\label{candVNf=1}\tiny 
\begin{aligned}
     v=&\frac{1}{2}+\frac{4 m }{\pi  \Lambda_1}q^{\frac16}+\frac{8  \left(\Lambda_1^3-4 m^3\right)}{\pi  \Lambda_1^3}q^{\frac36}+\frac{64 \left(374 m^5-135
   \Lambda_1^3 m^2\right)}{45 \pi  \Lambda_1^5}q^{\frac56} -\frac{64 \left(98800 m^7-49140 \Lambda_1^3 m^4+3969 \Lambda_1^6 m\right)}{567 \left(\pi 
   \Lambda_1^7\right)}q^{\frac76}  \\  
   &+\frac{64  \left(-17 \Lambda_1^9+12320 m^9-7840 \Lambda_1^3 m^6+1188 \Lambda_1^6 m^3\right)}{3 \pi  \Lambda_1^9}q^{\frac96}\\ 
   &-\frac{4096  \left(m^2 \left(-120285 \Lambda_1^9+13017172 m^9-10103049 \Lambda_1^3 m^6+2241756 \Lambda_1^6 m^3\right)\right)}{8019 \left(\pi 
   \Lambda_1^{11}\right)}q^{\frac{11}{6}}\\
   &+\frac{8 m \left(824527431 \Lambda_1^{12}+1880795176960 m^{12}-1723243991040 \Lambda_1^3 m^9+503472153600 \Lambda_1^6
   m^6-48655676160 \Lambda_1^9 m^3\right)}{85293 \pi  \Lambda_1^{13}} q^{\frac{13}{6}} \\
   &+\frac{16  \left(12393 \Lambda_1^{15}-1513521152 m^{15}+1599037440 \Lambda_1^3
   m^{12}-579532800 \Lambda_1^6 m^9+81786880 \Lambda_1^9 m^6-3533220 \Lambda_1^{12} m^3\right)}{5 \pi  \Lambda_1^{15}}q^{\frac{15}{6}}+\CO(q^{\frac{17}{6}}), \\
C=&
-4 q^{\frac16}+\frac{128 m^2 }{3 \Lambda_1^2}q^{\frac36}+\frac{128  \left(9 \Lambda_1^3 m-56 m^4\right)}{9 \Lambda_1^4}q^{\frac56}+\frac{16  \left(405 \Lambda_1^6+90112 m^6-27648 \Lambda_1^3 m^3\right)}{81 \Lambda_1^6}q^{\frac76}\\
&-\frac{4096  \left(320 m^8-144 \Lambda_1^3 m^5+11 \Lambda_1^6 m^2\right)}{3
   \Lambda_1^8}q^{\frac96}+\frac{512 m  \left(-26973 \Lambda_1^9+16187392 m^9-9584640 \Lambda_1^3 m^6+1370520 \Lambda_1^6 m^3\right)}{729 \Lambda_1^{10}}q^{\frac{11}{6}} \\
   &-\frac{16  \left(12892365 \Lambda_1^{12}+125513498624 m^{12}-92649553920 \Lambda_1^3 m^9+19568812032 \Lambda_1^6 m^6-707678208 \text{$\Lambda
   $1}^9 m^3\right)}{6561 \Lambda_1^{12}}q^{\frac{13}{6}} \\
   &+\frac{256 m^2  \left(204435 \Lambda_1^{12}+2078277632 m^{12}-1849688064 \Lambda_1^3 m^9+518246400 \text{$\Lambda
   $1}^6 m^6-35288064 \Lambda_1^9 m^3\right)}{63 \Lambda_1^{14}}q^{\frac{15}{6}}+\CO(q^{\frac{17}{6}}),\\
      \end{aligned}
\end{equation}
which we list to order $\CO(q^{\frac52})$.
We have determined them up to order $\CO(q^{4})$.

Relating the $q$-expansion of $v$ to that of the period $\frac{da}{du}$, we can compare the normalisation in the  weak coupling limit, 
\begin{equation}
    v-\frac12 \sim \frac{\sqrt{8}}{\pi i} m \frac{da}{du},
\end{equation}
as $\tau\to i \infty$.

One can similarly make a small mass expansion. For $C$, we find
\be 
\label{Csmallm}
\begin{split}
C&=-4\,q^{1/6}(1-20\,q+2250\,q^{2}+\CO(q^{3})\\
&\quad +128\,q^{5/6}(1-148\,q+27826\,q^2+\CO(q^3))\frac{m}{\Lambda_1} +\CO(m^2).
\end{split}
\ee 
For the small mass expansion of $e^{2\pi i v_1}$, we have
\be 
\label{vsmallm}
\begin{split} 
e^{2\pi i v_1}&=-1-16i\,q^{1/2}+128\,q+1408i\,q^{3/2}+\CO(q^2)\\
&\quad -8i\,q^{1/6}(1+46i\,q^{1/2}-240\,q-3200i\,q^{3/2}+\CO(q^{2}))\frac{m}{\Lambda_1}+\CO(m^2).
\end{split} 
\ee
We note that the small mass expansion commutes with the small $q$ expansion in \eqref{candVNf=1}.

As observed in \cite{Gottsche:2006}, the  coupling $v$ can also be expressed in terms of $\frac{du}{da}$ and a power series in $u$. Since under the large $u$ monodromy, $v$ transforms to $-v+1$ , we can rather consider $v-1$, which transforms to $-(v-1)$ \cite{Aspman:2022sfj}\footnote{We can confirm this directly from the $q$-series. It is also important to notice that at strong coupling, $v$ rather than $v-1$ transforms modularly \cite{Aspman:2022sfj}.}. This is also true for $\frac{du}{da}\to -\frac{du}{da}$, which  has modular weight $-1$ as well. This suggests that the quotient of $v-1$ and $\frac{du}{da}$ is invariant under any monodromy, and thus a function of $u$. A positive power series in $u$ can be excluded, as it would generate an arbitrarily large  principal part. For $N_f=1$ we can make the ansatz  
\begin{equation}
    v-\frac12=-\frac{1}{16\sqrt2 \pi i \Lambda_1}\frac{du}{da}\sum_{n=0}^\infty a_n\left(\frac{\Lambda_1^2}{u}\right)^n,
\end{equation}
where the coefficients $a_n$ depend on the mass $\mu=m/\Lambda_1$.   The first few coefficients are $a_0=0$, $a_1=8 \mu $, $a_2=\frac{8 \mu ^3}{3}-1$, 
$a_3=\frac{8 \mu ^5}{5}+2 \mu ^2$, $a_4=\frac{8 \mu ^7}{7}+\mu ^4-\frac{3 \mu }{4} $ etc. In the massless case, we have the series ($\Lambda_1=1$)
\begin{equation}\label{vdudaseries_m=0}
    -u^2\sum_{n=0}^\infty \frac{a_n}{u^n}=1-\frac{1}{24 u^3}+\frac{7}{2560 u^6}-\frac{3}{14336 u^9}+\CO(u^{-12}).
\end{equation}
It would be interesting to express this series as a hypergeometric function, such as the periods $\frac{da}{du}$ as given in \eqref{periods_hyperg}.
See also \cite{Poghosyan:2023zvy,Masuda:1996xj,suzuki2} for related formulae.

\subsubsection*{$\boldsymbol{N_f=2}$}

If we keep the masses in $N_f=2$ distinct, we  find
\begin{equation}
    v^1=-\frac{n_1}{2}-\frac{8im_1}{\pi\Lambda_2}\left(q^{1/4}+\left(\frac{16}{3\Lambda_2^2}(7m_1^2+3m_2^2)-\frac{8m_2}{m_1}\right)q^{3/4}+\CO(q^{5/4})\right).
\end{equation}
In the equal mass case, the couplings $v_1=v_2=v$, and in particular the couplings $C_{jk}$ can be found exactly. See the series \eqref{vmm}  and \eqref{CijNf2} in Section \ref{sec:equal_mass_expansion}.

\subsection{Large masses}\label{app:large_masses}
In this Appendix, we list expansions of Coulomb branch functions in the large mass limit. We consider $N_f$ hypermultiplets with equal masses $m_i=m$, $i=1,\dots, N_f$. 
When $m$ is large, the scales $\Lambda_{N_f}$ and $\Lambda_0$ are related as
\begin{equation}
    \Lambda_0^4=m^{N_f}\Lambda_{N_f}^{4-N_f}.
\end{equation}
The order parameters $u_{N_f}$ decouple precisely $u_{N_f}\to u_0$. We can find the large $m$ corrections by making the ansatz 
\be\label{uNf_large_m}
u_{N_f}=\sum_{n=0}^\infty c_n^{N_f}(u_0)m^{-2n},
\ee
and iteratively find $c_n^{N_f}$ by satisfying the relation $\CJ_{N_f}(u_{N_f})=\CJ_{0}(u_{0})$ order by order in $m^{-2}$. The fact that there are only even powers on $m$ follows by direct inspection. We list the coefficient functions $c_n^{N_f}$ in the Tables \ref{Nf1_large_m}, \ref{Nf2_large_m} and \ref{Nf3_large_m} for $N_f=1,2,3$.
\begin{table}[ht]\tiny\begin{center}
\renewcommand{\arraystretch}{1}
\begin{tabular}{|>{$\displaystyle}Sl<{$}| >{$\displaystyle}Sl<{$}|}
\hline
n&c_n^1(u)\\
\hline
 0 & u \\
 1 & \frac{1}{16} \left(3 \Lambda_0^4-4 u^2\right) \\
 2 & \frac{1}{128} \left(3 \Lambda_0^4 u-4 u^3\right) \\
 3 & \frac{1}{512} \Lambda_0^4 \left(3 \Lambda_0^4-4 u^2\right) \\
 4 & \frac{7 \left(16 u^5-40 \Lambda_0^4 u^3+21 \Lambda_0^8 u\right)}{32768} \\
 5 & \frac{9 \Lambda_0^{12}+32 u^6-60 \Lambda_0^4 u^4+15 \Lambda_0^8 u^2}{16384} \\
 6 & \frac{13 u \left(4 u^2-3 \Lambda_0^4\right) \left(-89 \Lambda_0^8+48 u^4+8 \Lambda_0^4 u^2\right)}{4194304} \\
 7 & \frac{3 \Lambda_0^4 \left(3 \Lambda_0^4-4 u^2\right) \left(\Lambda_0^8-8 u^4+9 \Lambda_0^4 u^2\right)}{131072} \\
 8 & -\frac{19 u \left(4 u^2-3 \Lambda_0^4\right) \left(5883 \Lambda_0^{12}+3520 u^6-15216 \Lambda_0^4 u^4+8244 \Lambda_0^8
   u^2\right)}{2147483648} \\
 9 & -\frac{11 \left(4 u^2-3 \Lambda_0^4\right) \left(5 \Lambda_0^{16}+32 u^8-76 \Lambda_0^4 u^6-35 \Lambda_0^8 u^4+90
   \Lambda_0^{12} u^2\right)}{16777216} \\
 10 & -\frac{5 u \left(4 u^2-3 \Lambda_0^4\right) \left(550163 \Lambda_0^{16}+396032 u^8+89344 \Lambda_0^4 u^6-2889312 \text{$\Lambda
   $0}^8 u^4+2143504 \Lambda_0^{12} u^2\right)}{274877906944} \\
 \hline
\end{tabular}
\caption{Coefficient functions $c_n^1$ of the large $m$ expansion \eqref{uNf_large_m} of $u_1$ in $N_f=1$} \label{Nf1_large_m}\end{center}
\end{table}
\begin{table}[ht]\tiny\begin{center}
\renewcommand{\arraystretch}{1}
\begin{tabular}{|>{$\displaystyle}Sl<{$}| >{$\displaystyle}Sl<{$}|}
\hline
n&c_n^2(u)\\
\hline
 0 & u \\
 1 & \frac{1}{8} \left(3 \Lambda_0^4-4 u^2\right) \\
 2 & \frac{1}{8} u \left(u^2-\Lambda_0^4\right) \\
 3 & 0 \\
 4 & -\frac{1}{128} u \left(u^2-\Lambda_0^4\right)^2 \\
 5 & 0 \\
 6 & \frac{u \left(u^2-\Lambda_0^4\right)^3}{1024} \\
 7 & 0 \\
 8 & -\frac{5 u \left(u^2-\Lambda_0^4\right)^4}{32768} \\
 9 & 0 \\
 10 & \frac{7 u \left(u^2-\Lambda_0^4\right)^5}{262144} \\
 \hline
\end{tabular}
\caption{Coefficient functions $c_n^2$ of the large $m$ expansion \eqref{uNf_large_m} of $u_2$ in $N_f=2$} \label{Nf2_large_m}\end{center}
\end{table}
\begin{table}[ht]\tiny\begin{center}
\renewcommand{\arraystretch}{1}
\begin{tabular}{|>{$\displaystyle}Sl<{$}| >{$\displaystyle}Sl<{$}|}
\hline
n&c_n^3(u)\\
\hline
0 & u \\
 1 & \frac{9 \Lambda_0^4}{16}-\frac{3 u^2}{4} \\
 2 & \frac{3}{128} u \left(20 u^2-19 \Lambda_0^4\right) \\
 3 & \frac{1}{512} \left(-35 \Lambda_0^8-128 u^4+164 \Lambda_0^4 u^2\right) \\
 4 & \frac{3 u \left(729 \Lambda_0^8+1232 u^4-1960 \Lambda_0^4 u^2\right)}{32768} \\
 5 & -\frac{21 \left(-3 \Lambda_0^{12}+32 u^6-60 \Lambda_0^4 u^4+31 \Lambda_0^8 u^2\right)}{16384} \\
 6 & \frac{u \left(-9637 \Lambda_0^{12}+42432 u^6-88816 \Lambda_0^4 u^4+56020 \Lambda_0^8
   u^2\right)}{4194304} \\
 7 & \frac{3 \Lambda_0^4 \left(\Lambda_0^4-8 u^2\right) \left(\Lambda_0^4-u^2\right) \left(3 \Lambda_0^4-4
   u^2\right)}{131072} \\
 8 & -\frac{3 u \left(206287 \Lambda_0^{16}+1230592 u^8-3662592 \Lambda_0^4 u^6+3829280 \Lambda_0^8
   u^4-1603568 \Lambda_0^{12} u^2\right)}{2147483648} \\
 9 & \frac{\left(u-\Lambda_0^2\right) \left(\Lambda_0^2+u\right) \left(4 u^2-3 \Lambda_0^4\right) \left(-59
   \Lambda_0^{12}+4576 u^6-6292 \Lambda_0^4 u^4+1759 \Lambda_0^8 u^2\right)}{16777216} \\
 10 & -\frac{3 u \left(-621219 \Lambda_0^{20}+32814080 u^{10}-105107200 \Lambda_0^4 u^8+123541632 \Lambda_0^8
   u^6-63642592 \Lambda_0^{12} u^4+13015300 \Lambda_0^{16} u^2\right)}{274877906944} \\
 \hline
\end{tabular}
\caption{Coefficient functions $c_n^3$ of the large $m$ expansion \eqref{uNf_large_m} of $u_3$ in $N_f=3$} \label{Nf3_large_m}\end{center}
\end{table}

Using the series \eqref{uNf_large_m}, we can similarly compute $\frac{du}{da}$ from the definition \eqref{dadu_def}, and expand it for large $m$ as 
\begin{equation}\label{duda_large_m_exp}
    \left(\frac{du}{da}\right)_{N_f}=\left(\frac{du}{da}\right)_{0}\sum_{n=0}^{\infty}d_n^{N_f}(u_0)m^{-2n}.
\end{equation}
We list the coefficients functions $d_n^{N_f}(u_0)$ in the following Table \ref{duda_large_m}.
\begin{table}[ht]\tiny\begin{center}
\begin{tabular}{|>{$\displaystyle}Sl<{$}| >{$\displaystyle}Sl<{$}|>{$\displaystyle}Sl<{$}|>{$\displaystyle}Sl<{$}|}
\hline
n&d_n^1(u)&d_n^2(u)& d_n^3(u)\\
\hline
0 & 1 & 1 & 1 \\
 1 & -\frac{u}{8} & -\frac{u}{4} & -\frac{u}{4} \\
 2 & -\frac{3}{256} \left(\Lambda_0^4+2 u^2\right) & \frac{1}{32} \left(u^2-2 \Lambda_0^4\right) & \frac{1}{32}
   \left(u^2-2 \Lambda_0^4\right) \\
 3 & -\frac{u \left(23 \Lambda_0^4+6 u^2\right)}{2048} & \frac{1}{128} u \left(u^2-2 \Lambda_0^4\right) &
   \frac{1}{128} u \left(u^2-2 \Lambda_0^4\right) \\
 4 & \frac{-123 \Lambda_0^8+140 u^4-876 \Lambda_0^4 u^2}{131072} & \frac{-12 \Lambda_0^8-5 u^4+12
   \Lambda_0^4 u^2}{2048} & \frac{-12 \Lambda_0^8-5 u^4+12 \Lambda_0^4 u^2}{2048} \\
 5 & \frac{3 u \left(-599 \Lambda_0^8+364 u^4-908 \Lambda_0^4 u^2\right)}{1048576} & \frac{-7 u^5+20
   \Lambda_0^4 u^3-20 \Lambda_0^8 u}{8192} & \frac{-7 u^5+20 \Lambda_0^4 u^3-20 \Lambda_0^8
   u}{8192} \\
 \hline
\end{tabular}
\caption{Coefficient functions $d_n^{N_f}$ of the large $m$ expansion \eqref{duda_large_m_exp} of $\left(\frac{du}{da}\right)_{N_f}$ for $N_f=1,2,3$.} \label{duda_large_m}\end{center}
\end{table}

\subsection{Equal mass \texorpdfstring{$N_f=2$}{Nf2}}\label{sec:equal_mass_expansion}
For the case of $N_f=2$ with equal masses, we can give explicit closed-form expressions for many Coulomb branch functions. For instance, we have \cite{Aspman:2021vhs} (see footnote \ref{ftn:sign})
\begin{equation}\begin{aligned}\label{nf2mmquantities}
		u&=-\frac{\Lambda_2^2}{8}\frac{\jt_4^8+\jt_2^4\jt_3^4+(\jt_2^4+\jt_3^4)\sqrt{f_2}}{\jt_2^4\jt_3^4}, \\
		\frac{du}{da}&=-\im\Lambda_2\frac{\sqrt{\jt_2^4+\jt_3^4+\sqrt{f_2}}}{\jt_2^2\jt_3^2}, \\
		\frac{du}{d\tau}&=\pi\im\Lambda_2^2\jt_4^8\frac{2(4\frac{m^2}{\Lambda_2^2}+1)\jt_2^4\jt_3^4+\jt_4^8+(\jt_2^4+\jt_3^4)\sqrt{f_2}}{8\jt_2^4\jt_3^4\sqrt{f_2}}, \\
		\frac{da}{d\tau}&=-\frac{\pi\Lambda_2}{8}\frac{\jt_4^8}{\jt_2^2\jt_3^2}\frac{\sqrt{f_2}(\jt_2^4+\jt_3^4)+\jt_4^8+2\jt_2^4\jt_3^4(1+4\tfrac{m^2}{\Lambda_2^2})}{\sqrt{f_2(\sqrt{f_2}+\jt_2^4+\jt_3^4)}},
\end{aligned}\end{equation}
where we defined $f_2=16\frac{m^2}{\Lambda_2^2}\jt_2^4\jt_3^4+\jt_4^8$ and we suppress the dependence on $\tau$. The physical discriminant is
\begin{equation}
	\Delta=(u-u_1^*)(u-u_2^*)(u-u_3^*)^2,
\end{equation}
where $(u_1^*,u_2^*)=(u_-,u_+)$ with 
$u_\pm=\tfrac{\Lambda_2^2}{8}\pm m\Lambda_2$,   and  $u_3^*=u_*=m^2+\tfrac{\Lambda_2^2}{8}$. Using the solution for $u$ above we can easily express also the discriminant in terms of theta functions. 

The fundamental domain was found in \cite{Aspman:2021vhs} and is shown in Fig. \ref{fig:nf2domain}. Due to the square roots there will also be branch points present, the paths of these as we vary the mass are also indicated in the Figures. 

In the equal mass case, the couplings $v^j$ to the background $\spinc$-structures are equal, $v^1=v^2\eqqcolon v$. Using the hypergeometric expressions for the periods and fixing the magnetic winding number $n=-1$, it is given by 
\begin{equation}\label{vmm}\footnotesize
	v= \frac{1}{2}+\frac{8}{\pi}\frac{m}{\Lambda_2}\left(q^{1/4}+\left(8-\frac{160}{3}\frac{m^2}{\Lambda_2^2}\right)q^{3/4}+\left(42-1536\frac{m^2}{\Lambda_2^2}+\frac{32256}{5}\frac{m^4}{\Lambda_2^4}\right)q^{5/4}+\CO(q^{7/4})\right).
\end{equation}
Furthermore, we have 
\begin{equation}\label{CijNf2}\footnotesize
	\begin{aligned}
		C_{11}=&\,C_{22}=e^{-\pi \im w^{11}}=\im\sqrt{2}\frac{\jt_1(\tau, v)^2}{\jt_4(\tau,v)\jt_4(\tau,0)}=4\im\sqrt{2}\left(q^{1/4}-64\frac{m^2}{\Lambda_2^2}q^{3/4}+\CO(q^{5/4})\right), \\
		C_{12}=&\,e^{-\pi \im w^{12}}=\frac{\jt_4(\tau,v)}{\jt_4(\tau,0)}=1+4q^{1/2}+\left(8-256\frac{m^2}{\Lambda_2^2}\right)q+\CO(q^{3/2}).
	\end{aligned}
\end{equation}
We can see that, even though $v\to \tfrac{1}{2}$ when $m\to 0$, the couplings $C_{ij}$ remain dependent on $\tau$. In the massless limit, we find
\begin{equation}
    \begin{aligned}
        &C_{11}(\tau)\to i\sqrt{2}\frac{\jt_2(\tau)^2}{\jt_3(\tau)\jt_4(\tau)},\\
        &C_{12}(\tau)\to \frac{\jt_3(\tau)}{\jt_4(\tau)},
    \end{aligned}
\end{equation}
where we denote $\jt_j(\tau)=\jt_j(\tau,0)$ as before. 

In the decoupling limit to the pure theory, $m\to \infty$, $\Lambda_2\to 0$, we expect $v\to 0$ and $C_{jk}\to 1$. We see that this agrees with the behaviour of $C_{12}$ (since then $\vartheta_4(\tau,0)=\vartheta_4(\tau,0)$), but not for $C_{11}$, where we instead find zero. 

Another important observation is that we have zeros and poles in $C_{ij}$ for special values of $v$. Namely, we have a double zero for $C_{11}$ when $v=l_1+l_2\tau$ and a pole when $v=l_1+l_2\tau+\tfrac{\tau}{2}$, where $l_j$ are any integers. The value $v=l_1+l_2\tau+\tfrac{\tau}{2}$ also corresponds to a simple zero of $C_{12}$. We thus want to exclude these values from the domain of $v$. The pole structure of $C_{ij}$ might be important for finding the correct modular completion as discussed in Appendix \ref{sec:bkgFluxes}. A similar story is occurs for the $\CN=2^*$ theory \cite{Manschot:2021qqe}, where a stronger claim can be made due to a certain relationship on theta functions involving $v$. 

\subsubsection*{Behaviour near $I_2$ singularity}
We finish this Appendix with some details of the behaviour near the $I_2$ singularities. We can also determine the behaviour of the couplings near other singularities using \eqref{nf2mmquantities}. For $da/du$ near the singularity $u_3^*$, we obtain for example 
\be 
\label{daduI2}
\begin{split}
&\left(\frac{da}{du} \right)_3(\tau_3)=\\
&\qquad \frac{i}{\Lambda_2} \frac{\vartheta_2(\tau_3)^2\vartheta_4(\tau_3)^2}{(\vartheta_2(\tau_3)^4-\vartheta_4(\tau_3)^4+(\vartheta_3(\tau_3)^8-16 (m^2/\Lambda_3^2) \vartheta_2(\tau_3)^4\vartheta_4(\tau_3)^4)^{1/2} )^{1/2}}.
\end{split}
\ee 
This gives for the leading term in the $q$-series,
\be
\label{dauq3}
\left(\frac{da}{du} \right)_3(\tau_3)=-\frac{i}{\sqrt{2\Lambda_2^2-8m^2}}+\CO(q_3^{1/2}).
\ee 
Similarly, we determine for the leading terms of $u_3$,
\be 
\label{uq3}
u_3(\tau_3)=u_3^*+\frac{(2\Lambda_2^2-8m^2)^2}{\Lambda_2^2}\,q_3^{1/2}+\CO(q_3).
\ee 
Using integration of  Matone's relation \eqref{Matone}, we find for the leading term of $a_3$,
\be 
\label{aq3}
a_3(\tau_3)=-\frac{i}{\Lambda_2^2} (2\Lambda_2^2-8m^2)^{3/2}\,q_3^{1/2}+\CO(q_3).
\ee

\section{Comments on non-vanishing background fluxes for \texorpdfstring{$\pt$}{P2}}\label{sec:bkgFluxes}
In Section \ref{sec:nonvanishing_bckg_flux}, we have studied point correlators on $\pt$ with non-vanishing background fluxes. Including background fluxes, for large masses we can use the mock Jacobi form $F_\mu(\tau,\rho)$, defined in \eqref{f12function}, for any value of the background fluxes $k_j$. For small masses however, this function does not have the right behaviour under $\slz$ monodromies. For $\mu=\frac12$, we can rather use the mock Jacobi form $H(\tau,\rho)$ \eqref{H12Def},
and determine the $u$-plane integral in a small mass expansion,
\bea
\label{Nfsmallmasses}
\Phi_{1/2}\big[e^{2pu/\Lambda_{N_f}^2}\big]=\, &\frac{(-1)^{N_f}}{2\,\Lambda_{N_f}^{12-2N_f}} {\rm Coeff}_{q^0} {\rm Ser}_{\bfm}\Bigg[ \left(\frac{du}{da}\right)^{12} \frac{\eta(\tau)^{27}}{P_{N_f}^{\text{M}}}\prod_{i,j=1}^{N_f}C_{ij}^{B(k_i,k_j)} \\
&\times \frac{1}{24}H(\tau,v_1k_1+\dots+v_{N_f}k_{N_f}) e^{2pu/\Lambda_{N_f}^2}\Bigg].
\eea

Alternatively, we can also calculate the series for finite masses, without expanding first around $\bfm=0$,
\bea
\label{Nffixedmasses}
\Phi_{1/2}\big[e^{2pu/\Lambda_{N_f}^2}\big]=\, &\frac{(-1)^{N_f}}{2\,\Lambda_{N_f}^{12-2N_f}} {\rm Coeff}_{q^0} \Bigg[ \left(\frac{du}{da}\right)^{12} \frac{\eta(\tau)^{27}}{P_{N_f}^{\text{M}}}\prod_{i,j=1}^{N_f}C_{ij}^{B(k_i,k_j)} \\
&\times \frac{1}{24}H(\tau,v_1k_1+\dots+v_{N_f}k_{N_f}) e^{2pu/\Lambda_{N_f}^2}\Bigg].
\eea

For $N_f=1$, we use the expansions  for $C=C_{11}$ and $v=v_1$ given in \eqref{Csmallm} and \eqref{vsmallm}. These expansions commute with the expansion for small $q$ in \eqref{candVNf=1}.
We list  the first few point correlators for small mass $m$, as in Eq. (\ref{Nfsmallmasses}), in Table \ref{nf1backgroundul1smallm}, while the ones for fixed mass $m$, as in Eq. (\ref{Nffixedmasses}), are listed in Table \ref{nf1backgroundul1smallmNEWEF}. We find that the small mass and fixed mass calculation agree for $k_1$ even, and is different for $k_1$ odd. We comment on this issue below. 
\begin{table}[ht]\small\begin{center}
\renewcommand{\arraystretch}{2}
\begin{tabular}{|>{$\displaystyle}Sl<{$}|  >{$\displaystyle}Sr<{$}| >{$\displaystyle}Sr<{$}| >{$\displaystyle}Sr<{$}| }
\hline
\ell& k_1=1 & k_1=2 & k_1=3 \\
\hline
  0& -1 & -\frac83 & -\frac{512}{9}\frac{m}{\Lambda_1}\\
  1& 0 & 0 & -\frac{64}{9}\Lambda_1^2-\frac{512}{9}\frac{m^3}{\Lambda_1} \\
 2& -\frac{19}{64}m \Lambda_1^3 & -m\Lambda_1^3-\frac{8}{3}m^4 & -\frac{128}{9} \Lambda_1^2 m^2 \\
 3& \frac{11}{2^{9}}\Lambda_1^6&\frac{11}{192}\Lambda_1^6+\frac{4}{3}m^3\Lambda_1^3&-\frac{8}{27}\Lambda_1^5 m-\frac{64}{3}\Lambda_1^2 m^4\\
 4& - &-\frac{73 }{64} \Lambda_1^6 m^2-3 \Lambda_1^3 m^5-\frac{8 m^8}{3}& \frac{7}{27}\Lambda_1^8-\frac{64}{81}\Lambda_1^5 m^3\\
 \hline
\end{tabular}
\caption{List of the first few $\Phi_{1/2}[u^{\ell}]$ for small mass $N_f=1$ on $\mathbb P^2$ with background flux, $k_1=1,2,3$. The expansion is determined up to $\CO(m^{5})$.
One entry is left undetermined since the $q^0$ of the integrand deviates by $\CO(q^3)$ from its leading term. 
} \label{nf1backgroundul1smallm}\end{center}
\end{table}
We can perform a similar analysis for $N_f=2$ with equal masses, where modular expressions are available (see Appendix \ref{sec:equal_mass_expansion}). See in particular \eqref{vmm} for $v$ and \eqref{CijNf2} for $C_{ij}$.
We list the results in Table \ref{nf2backgroundul1smallm}.

\begin{landscape}
\begin{table}[ht]\tiny\begin{center}
\renewcommand{\arraystretch}{1}
\begin{tabular}{|>{$\displaystyle}Sl<{$}|  >{$\displaystyle}Sr<{$}| >{$\displaystyle}Sr<{$}| >{$\displaystyle}Sr<{$}| >{$\displaystyle}Sr<{$}| }
\hline
\ell& k_1=1 & k_1=2 & k_1=3 \\
\hline
  0 & 1-\frac{\Lambda_1^3}{24 m^3} & -\frac{8}{3} & 0 \\
 1 & -\frac{\Lambda_1^6}{192 m^4}-\frac{\Lambda_1^3}{24 m} & 0 & 0 \\
 2 & -\frac{\Lambda_1^9}{1536 m^5}-\frac{\Lambda_1^6}{96 m^2}+\frac{49 \Lambda_1^3 m}{192} & -\frac{8 m^4}{3}-\Lambda_1^3 m & \frac{8 \Lambda_1^5}{9
   m} \\
 3 & -\frac{19 \Lambda_1^6}{512}-\frac{\Lambda_1^{12}}{12288 m^6}-\frac{\Lambda_1^9}{512 m^3}-\frac{\Lambda_1^3 m^3}{24} & \frac{11 \Lambda_1^6}{192}+\frac{4 \Lambda_1^3 m^3}{3} & \frac{\Lambda_1^8}{9 m^2}+\frac{64 \Lambda_1^5 m}{27} \\
 4 & \frac{2 m^{11}}{27 \Lambda_1^3}-\frac{\Lambda_1^{15}}{98304 m^7}-\frac{\Lambda_1^3 m^5}{24}-\frac{\Lambda_1^{12}}{3072 m^4}+\frac{7 \Lambda_1^6
   m^2}{64}-\frac{5031 \Lambda_1^9}{1048576 m} & -\frac{8 m^8}{3}-3 \Lambda_1^3 m^5-\frac{73 \Lambda_1^6 m^2}{64} & \frac{19 \Lambda_1^8}{27}+\frac{\Lambda_1^{11}}{72 m^3}+\frac{368 \Lambda_1^5 m^3}{81} \\
 \hline
\end{tabular}
\caption{List of the first few $\Phi_{1/2}[u^{\ell}]$ for fixed mass $N_f=1$ on $\mathbb P^2$ with background flux, $k_1=1,2,3$. For $k_1$ even, here $k_1=2$, it agrees precisely with the small mass calculation in Table \ref{nf1backgroundul1smallm}.} 
\label{nf1backgroundul1smallmNEWEF}\end{center}
\end{table}
\begin{table}[ht]\tiny\begin{center}
\renewcommand{\arraystretch}{1}
\begin{tabular}{|>{$\displaystyle}Sl<{$}|  >{$\displaystyle}Sr<{$}| >{$\displaystyle}Sr<{$}| >{$\displaystyle}Sr<{$}|  }
\hline
\ell&  (k_1,k_2)=(2,0) & (k_1,k_2)=(1,1)&(k_1,k_2)=(2,2) \\
\hline
  0& -\frac{2}{3}\frac{\Lambda_2^2}{m^2} & \frac43 +\frac{1}{12}\frac{\Lambda_2^2}{m^2}&0\\
  1&0 &-\frac{4}{3}m^2-\frac{5}{12}\Lambda_2^2+\frac{1}{96}\frac{\Lambda_2^4}{m^2} &0 \\
 2&  -\frac{2}{3}m^2\Lambda_2^2-\frac{1}{4}\Lambda_2^4+\frac{1}{768}\frac{\Lambda_2^6}{m^2} & \frac43 m^4+\frac54  m^2\Lambda_2^2+\frac{77}{384}\Lambda_2^4-\frac{1}{3072}\frac{\Lambda_2^6}{m^2}&-\frac16 \frac{\Lambda_2^6}{m^2} \\
 3&\frac{1}{3}m^2\Lambda_2^4+\frac{37}{384}\Lambda_2^6&-\frac{4}{3}m^6-\frac{25}{12}m^4\Lambda_2^2-\frac13 m^2\Lambda_2^4-\frac{259}{6144}\Lambda_2^6+\frac{251}{1572864}\frac{\Lambda_2^8}{m^2} &\frac23 \Lambda_2^6+\frac{1}{24}\frac{\Lambda_2^8}{m^2}\\
 4&-\frac{2}{3}m^6\Lambda_2^2-\frac34 m^4\Lambda_2^4-\frac{31}{384}m^2\Lambda_2^6-\frac{39}{2048}\Lambda_2^8+\frac{223}{3145728}\frac{\Lambda_2^{10}}{m^2} &\frac43 m^8+\frac{35}{12}m^6\Lambda_2^2+\frac{205}{96}m^4\Lambda_2^4+\frac{71}{512}m^2\Lambda_2^6+\frac{4609}{524288}\Lambda_2^8-\frac{251}{25165824}\frac{\Lambda_2^{10}}{m^2} &-\frac{13}{3}m^2\Lambda_2^6-\frac{25}{48}\Lambda_2^8-\frac{23}{3072}\frac{\Lambda_2^{10}}{m^2}\\
 \hline
\end{tabular}
\caption{List of the first few $\Phi_{1/2}[u^{\ell}]$ for fixed equal mass $N_f=2$ on $\mathbb P^2$ with background flux, $(k_1,k_2)=(2,0),(1,1),(2,2)$.} \label{nf2backgroundul1smallm}\end{center}
\end{table}
\end{landscape}
\subsubsection*{Discussion}
Comparing Tables \ref{nf1backgroundul1smallm} and \ref{nf1backgroundul1smallmNEWEF} for small and fixed mass in $N_f=1$, we find that the point correlators agree if $k_1$ is even, but are different when $k_1$ is odd. In the remainder of this section, we discuss this obstruction in detail.

The previous results of correlation functions in the absence of background fluxes give polynomials in the masses, see Section \ref{sec:p2}. Such expressions have a well-defined small \emph{and} large mass limit.  This is a consequence of the SW curves having a smooth massless and infinite mass limit.
If we include background fluxes, from  Table \ref{nf1backgroundul1smallmNEWEF} we can see that point correlators are polynomials in the mass if $k_1$ is even, and \emph{Laurent} polynomials if $k_1$ is odd. As one can check, these expressions do not have a consistent small or large mass limit. Thus it appears that odd  background fluxes introduce discontinuities in the mass dependence. 

Another possibility for the deviation is the branch point which could contribute to the $u$-plane integral. In Section \ref{sec:behaviour_bp}, we proved that the branch point of the integrand does not contribute for $\bfk_j=0$. Including the background fluxes, this argument needs to be revisited. However, it is not clear why this contribution would depend on the parity of $\bfk_j$.
A perhaps more evident explanation is that the anti-derivative used in the calculation is not valid for finite and small mass, for instance due to the appearance of poles or branch cuts. Let us analyse this possibility.

In Appendix \ref{app:mock}, we study the modular and analytical properties of the mock Jacobi form $H(\tau,\rho)$. The pole structure is derived in Table \ref{elliptic_residues}, with the following result: For any value $\rho=m+n\tau$ with $m,n\in\mathbb Z$, the function $H(\tau,\rho)$ is regular. At the three half-lattice points $\rho\in\{\frac12,\frac\tau2,\frac12+\frac\tau2\}$, the function $H$ has simple poles with residues 
\bea\label{H_residues}
    \underset{\rho=\frac12}{\Res}\, H(\tau,\rho)=\frac{4}{\pi \jt_2(\tau)}, \\
    \underset{\rho=\frac\tau2}{\Res}\, H(\tau,\rho)=\frac{4q^{\frac18}}{\pi i \jt_4(\tau)}, \\
    \underset{\rho=\frac12+\frac\tau2}{\Res}\, H(\tau,\rho)=\frac{4q^{\frac18}}{\pi \jt_3(\tau)}.
\eea
For $N_f=1$, we have $z=k_1 v$, where $v$ is given in \eqref{candVNf=1}.
We thus find that for $k_1$ odd there is a pole in $H(\tau,z)$ as $q\to 0$.
Thus near $q=0$, we have the asymptotics
\begin{equation}
    H(\tau,z)=\begin{cases}
    -2q^{-\frac18}\dots \quad &k_1 \text{ even}, \\
    (-1)^{\frac{k_1-1}{2}}\frac{\Lambda_1}{2k_1 m}q^{-\frac{7}{24}}\quad &k_1 \text{ odd.}
    \end{cases}
\end{equation}
The sign comes from the quasi-periodicity \eqref{quasi-periodicity} of $H$. Clearly, the pole introduces negative powers of the mass $m$, even if it is considered to be small.

By considering the smallest exponents of all factors in the integrand, we can derive a selection rule for the correlators $\Phi_{\frac12}[u^\ell]$ with given $k_1$ to not vanish. For $N_f=1$ on $X=\pt$, we have 
\begin{equation}\begin{aligned}
\nu(\tau)&\sim q^{\frac{1}{24}(4k_1^2-13)}, \\
u(\tau)&\sim q^{-\frac13}, \\
\end{aligned}\end{equation}
Then $\Phi_{\frac12}[u^\ell]$ can only be nonzero if the  leading exponent of $\nu(\tau) H(\tau,k_1v)u(\tau)^\ell$ is $\leq 0$. This gives the selection rule
\bea
    k_1^2&\leq 4+2l, \qquad k_1\text{ even}, \\
     k_1^2&\leq 5+2l, \qquad k_1\text{ odd}.
\eea
These agree in particular with Table \ref{nf1backgroundul1smallmNEWEF}. A similar argument can be made also for $N_f=2$, where some correlators vanish in Table \ref{nf2backgroundul1smallm}.
In particular, one can detect the selection rules for $k_1\geq 3$. For instance, for $N_f=1$ the partition function vanishes,
\begin{equation}
    \Phi_{\frac12,k_1}[1]=0, \qquad k_1 \geq 3.
\end{equation}

One possibility for the difference between small $m$ and large $m$ calculations is a pole of $H$ due to a zero of $v$ for odd $k_1$. In Fig. \ref{fig:vNf1} we sketch that this is indeed the case: We can numerically evaluate the function  $v(\tau)-\frac12$ for $N_f=1$ and a small mass $m$, and find zeros of this function inside the fundamental domain. For small mass $m\ll \mad$, there is a zero $v(\tn)=\frac12$ on the line $\text{Re}\, \tau_0=\frac32$. On the same line, one instance of the $N_f=1$ branch point moves to infinity if the mass $m$ is increased beyond $\mad$. For a small mass, the zero $\tn$ and the branch point $\tbp$ are thus well separated. 
\begin{figure}[ht]\centering
	\includegraphics[scale=1.5]{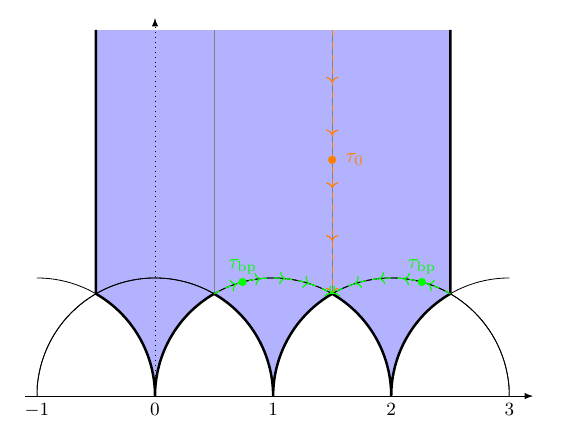}
	\caption{The fundamental domain of $N_f=1$ for a small positive mass $m$. The $\tbp$ are the two branch points (green), which move on the branch point locus (green, dashed) as the mass is varied, see \cite{Aspman:2021vhs} for a detailed discussion. Here, $\tn$ (orange) is the zero of the function $v(\tau)-\frac12$ inside the fundamental domain. As the mass is increased, this zero $\tn$ moves on the line (orange, dashed) with $\text{Re}\, \tau_0=\frac32$. This line coincides with the branch point locus for the  mass $m$ above the AD mass,  $m>\mad=\frac34$. It is thus clearly not the same as the branch point $\tbp$ for $u(\tau)$. The locus of the zero $\tn$ is determined for  $m/\Lambda_1\lesssim 0.3$. For larger masses, the zeros can not be accurately determined, since the coefficients in the $q$-series become very large.}\label{fig:vNf1}
\end{figure}

Based on this numerical evidence, let us assume that for a generic mass $m$, $v(\tau)-\frac12$ has a simple zero inside the domain $\CF_1(m)$. Assuming that $v$ is  holomorphic at $\tau=\frac12$, we can locally write  a  Taylor series
\begin{equation}
v(\tau)-\frac12=c_1(\tau-\tn)+\CO\left((\tau-\tn)^2\right).
\end{equation} 
Combining this with \eqref{H_residues}, we find that for $k_1=1$,
\begin{equation}
H(\tau,z(\tau))=\frac{4}{\pi c_1 \jt_2(\tn)} (\tau-\tn)^{-1} +\CO(1).
\end{equation}
We thus see that the integrand has a pole at the point $\tn$.

This gives an explanation for obtaining different results in a large and small mass calculation: For $m\to0$, we find $\tn\to\frac32+i\infty$. In the regularisation $F_Y(m)$ of the fundamental domain we discuss in Section \ref{sec:funddom}, we choose $Y\ll \text{Im}\,\tn$, such that in the limit $Y\to \infty$ we do not cross the pole $\tn$ of $H(\tau,z(\tau))$. For a finite mass $m$, $\text{Im}\,\tn$ is finite and sending $Y\to\infty$ picks up a contribution from the pole $\tn$. For very large $m$, we can express the integrand as a Taylor series in $1/m$, with leading term the $N_f=0$ integrand. Thus, any possible pole $\tn$ moves to infinity as well. We choose $Y<\text{Im}\,\tn$ in this case, such that there is no contribution due to the pole.
This discussion agrees with the observation that the large $m$ calculation (Table \ref{nf1backgroundul1}) has a well-defined large $m$ limit, the small mass calculation (Table \ref{nf1backgroundul1smallm}) has a well-defined $m\to0$ limit, while the generic mass calculation (Table \ref{nf1backgroundul1smallmNEWEF}) does not have either limits.

A similar numerical study for the other two poles $\rho\in\{\frac\tau2,\frac12+\frac\tau2\}$ of $H$ --- modulo the lattice $\mathbb Z\oplus \tau \mathbb Z $ --- does not result in such a zero: It seems that $v(\tau)=\frac\tau2$ and $v(\tau)=\frac12+\frac\tau2$ do not have solutions in the fundamental domain $\CF_1(m)$.

Let us also briefly discuss the cases of $N_f>1$. For equal mass $N_f=2$, we have
\begin{equation}
    v=\frac{1}{2}+\frac{8}{\pi}\frac{m}{\Lambda_2}\left(q^{1/4}+\left(8-\frac{160}{3}\frac{m^2}{\Lambda_2^2}\right)q^{3/4}+\CO(q^{5/4})\right),
\end{equation}
and $z=(k_1+k_2)v$. We thus see that for $(k_1+k_2)\in 2\BZ+1$ there is again a pole in $H(\tau,z)$ for both $m\to 0$ and $q\to 0$. The values in Table \ref{nf2backgroundul1smallm}
avoid these values. Interestingly, we see that the pole structure of $C_{11}$ for equal mass $N_f=2$, see \eqref{CijNf2}, does not have any overlap with that of $H(\tau,z)$ for any $k_1+k_2$. 

For general $N_f$, we have $v_j-\frac12 \in \CO(q^{-\frac{1}{2(4-N_f)}})$ for arbitrary masses. Then we have $z=\sum_j v_jk_j=\frac12(k_1+\dots+k_{N_f})+\CO(q^{-\frac{1}{2(4-N_f)}})$, such that $H(\tau,z)$ has a pole as $q\to 0$ precisely if  $\sum_j k_j$ is odd. This agrees with the two previous cases. On the other hand, if $\sum_j k_j$ is even, we expect there to not be any poles.

A possible way to cure the issue of the pole at $\rho=\frac12$ is to add to the anti-derivative $H(\tau,\rho)$ a holomorphic $\slz$ Jacobi form of the same weight and index as $H$, which cancels precisely the pole of $H$ but does not introduce any additional poles. As we discuss in Appendix \ref{app:mockjacobi}, $\widetilde H(\tau,\rho)$ \eqref{wtHtaurho} is such a function, and compatible with the monodromies as long as $\sum_j k_j \in 2\mathbb{Z}$. For these cases, using $\widetilde H$ we find the same point correlators as in the Tables \ref{nf1backgroundul1smallm}--\ref{nf2backgroundul1smallm}. These are however precisely the cases where the large mass and small mass calculations agree, rendering the various limits well-defined. Meanwhile, for the cases $\sum_j k_j$ odd where the two calculations disagree, using $\widetilde H$ it is not possible to test if avoiding only the pole at $\rho=\frac12$ is sufficient to obtain point correlators with a suitable infinite mass limit, since $\widetilde H$ is not applicable to those cases. We hope to clarify this issue in future work.

\section{A Kodaira invariant}\label{sec:Kodaira_invariant}

In Section \ref{sec:ADcontribution}, we studied the order of vanishing of the $u$-plane integral at an arbitrary elliptic AD point in an arbitrary configuration. See \eqref{ord_nu} for instance for the order of vanishing of the measure factor $\nu$ at any given AD point. In this Appendix, we determine the values $\ord \ttu$ and $\ord\omega$ from the data of the singular fibres of the elliptic surface containing a singularity of a given type.  Furthermore, we prove that these two orders of vanishing are not independent, the relation is fully determined  by the type of AD point. We keep the discussion completely general, and the results are universal formulas for arbitrary elliptic surfaces. 

We consider the parametrisation $u(\tau)\in\mathbb P^{1}$ such that there is an elliptic point $\tn$  with  $u(\tn)=u_0$, and we define $\ttu=u-u_0$. As a function of $\tau$, we can then study the order of vanishing $\ord\ttu$   as well as that of the holomorphic period
\begin{equation}
    \omega\coloneqq \left(\frac{du}{da}\right)^2 \propto \frac{g_3}{g_2}\frac{E_4}{E_6}.
\end{equation}
We then show that 
\begin{equation}\label{chiTdef}
o_\CT\coloneqq \frac{\ord\ttu}{\ord \omega}
\end{equation}
does \emph{not} depend on the configuration involving an AD point of Kodaira type $\CT$, but rather only on the Kodaira type $\CT$ itself. 
In this Appendix, we use the notation $\ord f$ for the order of vanishing of a function $f$ of $\tau\in\mathbb H$, while we reserve the symbol $o_f$ for the order of vanishing if $f$ is considered as a function on the base $u\in \mathbb P^1$.

\begin{table}%
\centering 
 \begin{tabular}{|Sc | Sc|Sc|Sc|Sc|} %
\hline
$\mathscr A$ & $\tau$  &$o_2$  &$o_3$ & $o_\Delta$   \\  \hline
$I_k$  & $i \infty $&  $0$ &$0$ & $k$        \\  
$I_{k}^\ast$  &$i \infty$&   $2$ &$3$ & $k+6$     \\  
$I_{0}^\ast$  &$\tau$&  $\geq 2$ &$\geq 3$ & $6$        \\  
 $II$ &$\alpha$& $\geq 1$ &$1$ & $2$          \\  
 $II^\ast$ &$\alpha$& $\geq 4$ &$5$ & $10$          \\  
  $III $&$i$  & $1$ &$\geq 2$ & $3$         \\  
$  III^\ast $&$i$& $3$ &$\geq 5$ & $9$         \\  
   $IV$ &$\alpha$& $\geq 2$ &$2$ & $4$         \\  
$ IV^\ast$ &$\alpha$ & $\geq 3$ &$4$ & $8$          \\ \hline 
\end{tabular}
\caption{The Kodaira classification of singular fibres \cite{kodaira63II,kodaira63III}. The symbol $\mathscr A$ denotes the type of singular fibre. The value $\tau\in \mathbb H$ denotes the $\slz$ orbit corresponding to $\mathscr A$. Here, $\alpha=e^{2\pi i/3}$. Finally, we define $(o_2,o_3,o_\Delta)=(\ord g_2,\ord g_3,\ord \Delta)$ as the orders of vanishing of the Weierstra{\ss} invariants $g_2$, $g_3$ and the discriminant $\Delta$ on the base of the elliptic surface.  \label{tab:kodaira}}
\end{table} 

 Recall  Kodaira's classification of singular fibres, as given in Table \ref{tab:kodaira}. Let $\mathscr C=(\mathscr A, \dots)$ be a configuration containing an additive fibre 
\begin{equation}\label{additive_fibres}
\mathscr A\in\{II, III, IV, IV^*, III^*, II^*\}
\end{equation}
not being of type $I_k^*$. For instance, in $\CN=2$ SU(2) SQCD with $N_f\leq 4$ hypermultiplets  we consider the surfaces $\mathscr C=(I_{4-N_f}^*,\mathscr A, \dots)$ with an AD point of Kodaira type $\mathscr A\in \{II, III, IV\}$. 
As is clear from Table \ref{tab:kodaira}, when $\mathscr A\in A_\alpha\coloneqq\{II,  IV, IV^*,  II^*\}$, then $o_3$ and $o_\Delta$ are fixed, while when $\mathscr A\in A_i\coloneqq\{ III,  III^*\}$ then $o_2$ and $o_\Delta$ are fixed, in both cases with  the remaining $o_j$ being bounded from below. 
When $\mathscr A\in A_\alpha$, then the singular fibre corresponds to $\tn\in \Gamma_1\cdot \alpha$, where $\alpha=e^{\frac{2\pi i}{3}}$, while when $\mathscr A\in A_i$ then $\tn\in \Gamma_1\cdot i$, with $\Gamma_1\coloneqq \psl$. Since $\CJ$ and $j$ behave differently near those points $i$ and $\alpha$, these families $A_\alpha$ and $A_i$ of singular fibres thus need to be treated separately. 

In order to compare the orders of vanishing of the various quantities, we use the following Taylor expansions of the modular $j$-invariant,
\begin{equation}\begin{aligned}\label{jinvexpansionselliptic}
    j(\tau)=\frac{1}{3!}j'''(\alpha)(\tau-\alpha)^3+\CO\left((\tau-\alpha)^4\right), \\
    j(\tau)-12^3=\frac{1}{2!}j''(i)(\tau-i)^2+\CO\left((\tau-i)^3\right),
    \end{aligned}
\end{equation}
where crucially 
\begin{equation}\begin{aligned}
j'''(\alpha)&=-2^{13}3^3\sqrt{3} i  \pi ^3 \Oa^6, \\
j''(i)&=-2^93^4 \pi^2 \Oi^4
\end{aligned}\end{equation}
are nonzero, with $\Oa$ and $\Oi$ the Chowla--Selberg periods \eqref{CSperiodsOaOi} associated with the elliptic points $\alpha$ and $i$ (see Appendix \ref{chowla_selberg} for a review).
This is due to the fact that $E_4$ has a simple zero at $\tau=\alpha$ and $E_6$ has a simple zero at $\tau=i$, which follows from the valence formula for modular forms on $\Gamma_1$ (see for example \cite{Bruinier08}).

\subsubsection*{The case $\mathscr A\in A_\alpha$}
When $\mathscr A\in A_\alpha$, then $o_3$ and $o_\Delta$ are fixed by  $\mathscr A$, while $o_2$ can vary. We expand $j(\tau)=\CJ(u)$ around $\tau=\tn$ and $u=u_0$, where $u(\tn)=u_0$. From $\CJ=12^3\frac{g_2^3}{\Delta}$ we have $o_\CJ=3o_2-o_\Delta$.
Comparing with $\ord j=3$ this gives 
\begin{equation}\label{orduAalpha}
    \ord\ttu=\frac{3}{3o_2-o_\Delta}.
\end{equation}
Regarding $\omega$, we have that $\ord \frac{E_4}{E_6}=1$, while $\frac{g_3}{g_2}\sim \ttu^{o_3-o_2}$. We can then insert  \eqref{orduAalpha}, and  find
\begin{equation}
    \ord\omega=\frac{3o_3-o_\Delta}{3o_2-o_\Delta}.
\end{equation}
Interestingly, the dependence of 
\begin{equation}\label{o_Aalpha}
   o=\frac{\ord\ttu}{\ord\omega}=\frac{3}{3o_3-o_\Delta}
\end{equation}
on $o_2$ drops out. Since for a fixed $\mathscr A\in A_\alpha$ the only variable is $o_2$, the ratio $o$ is indeed an invariant $o_{\mathscr A}$ of the singular fibre $\mathscr A\in A_\alpha$ itself. 

\subsubsection*{The case $\mathscr A\in A_i$}
 When $\mathscr A\in A_i$ we have that $j(\tn)=12^3$, such that we need to study the order of vanishing of $\CJ-12^3=12^3 \frac{g_3^2}{\Delta}$. While $\ord(j-12^3)=2$, we can compare this to $ o_{\CJ-12^3}=2o_3-o_\Delta$. This gives the order of vanishing
 \begin{equation}\label{orduAi}
     \ord\ttu=\frac{2}{2o_3-o_\Delta}
 \end{equation}
on the $\tau$-plane. Now $\ord \frac{E_4}{E_6}=-1$, while $\frac{g_3}{g_2}\sim \ttu^{o_3-o_2}$ where we insert \eqref{orduAi}. This gives
\begin{equation}
    \ord\omega=\frac{o_\Delta-2o_2}{2o_3-o_\Delta},
\end{equation}
where both numerator and denominator are positive. We find
\begin{equation}\label{o_Ai}
   o= \frac{\ord\ttu}{\ord\omega}=\frac{2}{o_\Delta-2o_2},
\end{equation}
where again the only free variable $o_3$ drops out. Thus also for a fixed $\mathscr A\in A_i$ the ratio $o$ is an invariant $o_{\mathscr A}$ of the singular fibre $\mathscr A$. 

In Table \ref{tab:ordD}, we compute $o_{\mathscr A}$ for all additive fibres $\mathscr A$ \eqref{additive_fibres}. This proves the above claim that $o_\CT$ is fixed for all given theories $\CT$ containing an equivalent AD point, and extends it to other rational elliptic surfaces as well, such as those describing the rank 1 $E_n$ theories  
\cite{Seiberg:1996bd,Morrison:1996xf,Closset:2021lhd,Magureanu:2022qym} containing Minahan--Nemeschansky (MN) theories \cite{Minahan:1996fg,Minahan:1996cj}, where $\mathscr A\in\{IV^*, III^*, II^*\}$.

\begin{table}%
\centering 
\begin{tabular}{|Sc|Sc|Sc|Sc|Sc|Sc|Sc|  } %
\hline
$\mathscr A$& $II$ & $III$ & $IV$ & $IV^*$ & $III^*$ & $II^*$ \\  \hline
$\displaystyle o_{\mathscr A}$  & $\displaystyle 3$& $\displaystyle 2$ & $\displaystyle\frac32$ & $\displaystyle \frac34$ & $\displaystyle \frac23$ & $\displaystyle \frac35$\\ 
$\displaystyle \Delta_{\mathscr A}$ & $\displaystyle \frac65$ & $\displaystyle \frac43$ & $\displaystyle \frac32$ & $\displaystyle 3$&$\displaystyle 4$&$\displaystyle 6$\\ \hline
\end{tabular}
\caption{The quotient $o_{\mathscr A}=\frac{\ord \ttu}{\ord \omega}$ depends only on the type $\mathscr A$ of additive fibre, not on the configuration $\mathscr C=(\mathscr A, \dots)$ containing it. We compare this invariant to the dimension of the CB operator of the corresponding rank-1 SCFT. \label{tab:ordD}}
\end{table} 

We can combine the results of the quotient $o$ in both cases $A_\alpha$ and $A_i$: From Table \ref{tab:kodaira} we find that for $A_\alpha$, we have $o_\Delta=2o_3$ and for $A_i$ we have $o_\Delta=3o_2$. Inserting these into \eqref{o_Aalpha} and \eqref{o_Ai} gives the much simpler formula
\begin{equation}
    o=\frac{6}{o_\Delta},
\end{equation}
which holds for all cases \ref{additive_fibres}. This can be easily confirmed from the results in Table \ref{tab:ordD}.

By comparing these ratios of orders of vanishing with the dimensions of Coulomb branch operators of the corresponding SCFT (see \cite{Caorsi:2019vex} for an overview), we find the relation 
\begin{equation}
    o_{\mathscr A}=\frac{\Delta_{\mathscr A}}{2(\Delta_{\mathscr A}-1)}.
\end{equation}
Similar formulas for the dimensions $\Delta_{\mathscr A}$ in terms of the SW curve exist \cite{Eguchi:1996vu,Caorsi:2017bnp,Cecotti:2021ouq}. For instance, the dimension $\Delta_{\mathscr A}$ is related to the order of vanishing $ o_\Delta$ of the discriminant on the base as \cite{Aharony_2008}
\begin{equation}
    \Delta_{\mathscr  A}=\frac{1}{1-\frac{o_\Delta}{12}}.
\end{equation}

\providecommand{\href}[2]{#2}\begingroup\raggedright\endgroup
\end{document}